\documentclass[12pt]{article}
\usepackage[authoryear,round]{natbib}
\usepackage{amssymb}
\usepackage{graphicx}
\usepackage{latexsym}
\usepackage[left=1.5in,top=1.5in,right=1.5in,bottom=1.5in]{geometry}

\usepackage{xcolor}
\definecolor{Bleu}{RGB}{0,0,204}
\definecolor{Violet}{RGB}{102,0,204}
\definecolor{Rouge}{RGB}{204,0,0}
\definecolor{Highlight}{RGB}{251,0,0}
\usepackage[hypertexnames=false]{hyperref}
\hypersetup{
colorlinks,
  citecolor=Bleu,
  linkcolor=Bleu,
  urlcolor=Violet} 

\usepackage{amsmath}
\usepackage{amsthm,thmtools}
\renewcommand\thmcontinues[1]{Continued}
\usepackage{amssymb} 
\usepackage{booktabs}
\usepackage{color}
\usepackage{tikz}
\usepackage{rotating}
\usepackage{enumitem}
\usepackage{bbm}

\usepackage{dsfont} 

\usepackage{titlesec} 

\newcommand{\minus}{\scalebox{0.5}[1.0]{$-$}}
\newcommand{\VEm}{\operatorname{VE}_{\minus}}

\usetikzlibrary{shapes}
\usetikzlibrary{arrows}
\definecolor{darkblue}{rgb}{0,0.4,0.9}
\definecolor{gray10}{rgb}{0.1,0.1,0.1}
\definecolor{gray20}{rgb}{0.2,0.2,0.2}
\definecolor{gray30}{rgb}{0.3,0.3,0.3}
\definecolor{gray40}{rgb}{0.4,0.4,0.4}
\definecolor{gray60}{rgb}{0.6,0.6,0.6}
\definecolor{gray80}{rgb}{0.8,0.8,0.8}
\definecolor{gray90}{rgb}{0.9,0.9,.9}
\definecolor{gray95}{rgb}{0.95,0.95,.95}
\definecolor{gray96}{rgb}{0.96,0.96,.96}
\definecolor{lgreen} {RGB}{180,210,100}
\definecolor{dblue}  {RGB}{20,66,129}
\definecolor{ddblue} {RGB}{11,36,69}
\definecolor{lred}   {RGB}{220,0,0}
\definecolor{nred}   {RGB}{224,0,0}
\definecolor{norange}{RGB}{230,120,20}
\definecolor{nyellow}{RGB}{255,221,0}
\definecolor{ngreen} {RGB}{98,158,31}
\definecolor{dgreen} {RGB}{78,138,21}
\definecolor{nblue}  {RGB}{28,130,185}
\definecolor{jblue}  {RGB}{20,50,100}
\definecolor{nnyellow}{RGB}{235,200,0}
\definecolor{purple}{RGB}{150, 0, 120}
\definecolor{sgGreen} {RGB}{20, 180, 50}
\definecolor{revised}{rgb}{0,0,0.9}

\newtheorem{theorem}{Theorem}

\newtheorem{lemma}[theorem]{Lemma}

\theoremstyle{definition}

\declaretheorem[name=Remark,qed={$\Box$}]{remark} 

\newcommand{\openr}{\hbox{${\rm I\kern-.2em R}$}}
\newcommand{\openn}{\hbox{${\rm I\kern-.2em N}$}}
\newcommand{\logit}{\operatorname{logit}}
\newcommand{\Rem}{\operatorname{Rem}}
\newcommand\independent{\protect\mathpalette{\protect\independenT}{\perp}}
\def\independenT#1#2{\mathrel{\rlap{$#1#2$}\mkern2mu{#1#2}}}
\newcommand{\norm}[1]{\left\lVert#1\right\rVert}

\newcommand{\nmin}{n_{\textnormal{min}}}
\DeclareMathOperator{\Var}{Var}

\DeclareMathOperator{\argmin}{argmin}
\DeclareMathOperator{\E}{\mathbb{E}}
\DeclareMathOperator{\VE}{VE}
\DeclareMathOperator{\UR}{UR}

\DeclareMathOperator{\Ind}{\mathds{1}}

\titleformat{\subsubsection}[runin] 
  {\normalfont\bfseries}{\thesubsubsection}{1em}{}

\allowdisplaybreaks

\usepackage[normalem]{ulem}
\usepackage{longtable}
\usepackage{authblk}

\title{Partial Bridging of Vaccine Efficacy to New Populations}
\author{Alexander R. Luedtke}
\author{Peter B. Gilbert}
\affil{Vaccine and Infectious Disease Division, Fred Hutchinson Cancer Research Center}
\date{}

\begin{document}



\def\spacingset#1{\renewcommand{\baselinestretch}%
{#1}\small\normalsize} \spacingset{1}


\maketitle

\bigskip
\begin{abstract}
Suppose one has data from one or more completed vaccine efficacy trials and wishes to estimate the efficacy in a new setting. Often logistical or ethical considerations make running another efficacy trial impossible. Fortunately, if there is a biomarker that is the primary modifier of efficacy, then the biomarker-conditional efficacy may be identical in the completed trials and the new setting, or at least informative enough to meaningfully bound this quantity. Given a sample of this biomarker from the new population, we might hope we can bridge the results of the completed trials to estimate the vaccine efficacy in this new population. Unfortunately, even knowing the true conditional efficacy in the new population fails to identify the marginal efficacy due to the unknown conditional unvaccinated risk. We define a curve that partially identifies (lower bounds) the marginal efficacy in the new population as a function of the population's marginal unvaccinated risk, under the assumption that one can identify bounds on the conditional unvaccinated risk in the new population. Interpreting the curve only requires identifying plausible regions of the marginal unvaccinated risk in the new population. We present a nonparametric estimator of this curve and develop valid lower confidence bounds that concentrate at a parametric rate. We use vaccine terminology throughout, but the results apply to general binary interventions and bounded outcomes.
\end{abstract}

\noindent%
\vfill

\newpage
\section{Introduction}
Randomized clinical trials are the gold standard for evaluating the effect of a new intervention in a population. However, it is not always possible or appropriate to conduct a Phase 3 efficacy trial. In these cases, it is desirable to ``bridge'' results from earlier trials to a new setting that may differ in population, treatment version, or exposure levels. Being able to extrapolate efficacy signals to a new population is a special case of generalizability or external validity \citep{Cole&Stuart2010}, and is often referred to as transportability \citep{Bareinboim&Pearl2012,Bareinboim&Pearl2013,Pearl&Bareinboim2014}. Following the clinical trials literature, we will instead refer to this extrapolation as ``bridging'' \citep{Gilbert&Hudgens2008,Gilbertetal2011,Fleming&Powers2012}.

There have been several recent trials establishing the efficacy of the CYD-TDV vaccine in children. In particular, CYD14 found 56\% per-protocol efficacy of the vaccine in Asian children aged 2-14 years \citep{Capedingetal2014}, CYD15 found 61\% per-protocol efficacy the same vaccine in Latin American children aged 9-16 years \citep{Villaretal2015}, and CYD23 found 30\% per-protocol efficacy in Thai children aged 4-11 years \citep{Sabchareonetal2012}. All trials featured a 2:1 vaccine:placebo randomization scheme, with respective sample sizes of approximately 10,000, 21,000, and 4,000 and respective virologically confirmed dengue event counts of 250, 397, and 134. In each trial, efficacy was shown to increase with age and prior exposure to dengue \citep{Capedingetal2014,Villaretal2015,Sabchareonetal2012}. While several countries have licensed the CYD-TDV vaccine for individuals from 9--45 or 9--60 \citep{WHO2016}, many countries have still not licensed this vaccine due to debates over the appropriate ages of indication.  Nonetheless, the high estimated efficacy in the existing trials may make further Phase 3 trials unethical. There is thus compelling motivation to bridge the results of the earlier efficacy trials completed in children to the adult population.

The focus of most existing work on bridging is on bridging the results from a single efficacy trial to a new population. Unsurprisingly, some assumptions must be made to bridge between the two studies. The primary standard assumption is that (i) the two populations have equivalent outcome regressions, i.e. that the mean outcome within each treatment/covariate/immune response biomarkers stratum is constant, and (ii) that the support of the treatment/covariate/immune response biomarker distribution in the population to which the researchers wish to bridge is contained in the corresponding support in the population in which the trial was completed \citep{Bareinboim&Pearl2012,Bareinboim&Pearl2013,Pearl&Bareinboim2014}. Recently, an efficient, double robust estimator was established when these assumptions hold \citep{Rudolph&vanderLaan2016}. There has also been some work on establishing transportation formulas when there are multiple settings in which full observations are observed in more than one setting, which in our setting means that there is more than one completed efficacy trial \citep{Lee&Honavar2013,Bareinboim&Pearl2014}.

Assumption (i) is often not plausible in the infectious disease setting unless one adjusts for the level of pathogen exposure, but in most efficacy trials there is only coarse exposure data (e.g., age in dengue trials). Furthermore, adjusting for level of pathogen exposure can violate Condition (ii) in some scenarios. To overcome this problem, \cite{Gilbert&Huang2016} posited a bridging assumption on the ratio scale, arguing that the vaccine versus placebo risk ratios in the two populations should be approximately equal within each covariate-specific principal stratum because this ratio reflects the vaccine's effect on biological susceptibility to pathogen infection or disease. Because knowing the risk ratio does not uniquely determine the outcome regression, there remains an undetermined degree of freedom. One possibility is to estimate the outcome regression among unvaccinated participants in the new population using epidemiological data \citep{Gilbert&Huang2016}. 

There is often severe underreporting of infection or disease incidence in the available epidemiologic surveillance data. For example, recent studies comparing active and passive surveillance (by national surveillance programs) of dengue incidence have seen up to 19-fold underreporting in the passive surveillance \citep{Nealonetal2016,Sartietal2016}. Though one could attempt to account for the underreporting by specifying an inflation factor for the passive surveillance data, there will likely be uncertainty around what exactly this inflation factor should be. For this reason, in this work we focus on a less rigid approach to account for the unknown unvaccinated risk allocation. In particular, we derive the worst-case allocation under constraints that we will impose in the next section. These constraints can be estimated via epidemiologic data or existing trial data. This then yields the lower bound on the vaccine efficacy in the new population that we wish to estimate using double robust methods \citep[e.g.,][]{Robins&Rotnitzky95,vdL02,vanderLaan&Rose11}. Identifying a parameter that provides a bound on an unidentifiable parameter is often referred to as partial identification \citep{Manski2003}, and these identification problems have received considerable attention over the last several decades \citep[e.g.,][]{Manski1990,Horowitz&Manski2000,Maclehose2005}. Because we are specifically focused on bridging a parameter from one or several populations to a new setting, we refer to this exercise as partial bridging, though really our objective is a special case of partial identification.


\subsubsection*{Organization of article.}
Section~\ref{sec:setup} describes the observed data and presents a lower bound on the vaccine efficacy in the new population. 
Section~\ref{sec:monotone} derives a first-order approximation to the lower bound of inference that will play a key role in our estimation scheme. Section~\ref{sec:est} describes our estimation scheme, with Section~\ref{sec:nonconstantVE} focusing on the case where the conditional vaccine efficacy surface is not known to be constant across completed efficacy trials, and Section~\ref{sec:veconst} describing how to improve the precision of the lower confidence bound when the conditional efficacy surface is known to be the same across the completed efficacy trials. Section~\ref{sec:sim} presents a simulation study. Section~\ref{sec:disc} concludes with some brief remarks.

All of our proofs can be found in Appendix~\ref{app:proofs}. Appendix~\ref{app:est} presents technical regularity conditions used to establish the validity of our estimator. Appendix~\ref{app:extensions} extends our results to two-phase sampling designs and monotonic vaccine efficacy curves. Appendix~\ref{app:nestedccsim} presents simulation results under a two-phase sampling design.

\section{Problem setup} \label{sec:setup}
\subsection{Notation and bridging assumptions}
Before presenting our problem, we introduce a few basic pieces of notation. For functions $f$ and $g$ mapping from some space $\mathcal{X}$ to another space $\mathcal{Z}$, we write $f\equiv g$ to denote equality of $f(x)$ and $g(x)$ for all $x$. We also write $f\equiv 0$ to denote that $f(x)=0$ for all $x$. We use ``$\triangleq$'' to denote a definition, e.g. $f(x)\triangleq x$ denotes the identity map. When we want to emphasize that $f$ is a function, we refer to $f$ as $w\mapsto f(w)$. As is standard in the empirical process literature, we let $P f$ denote the expectation of $f(X)$ when $X$ is drawn from the distribution $P$ \citep[e.g.,][]{vanderVaartWellner1996}. 

Consider the data structure $(W,A,Y)$, where $W$ is a (possibly multivariate) baseline (pre-randomization) biomarker, $A$ is a vaccination indicator, and $Y$ is the outcome of interest that occurs subsequent to vaccination. 
While all of the results in this work hold for general bounded $Y$, we focus on the case that $Y$ is an indicator of infection or disease, since this is our primary case of interest. To avoid introducing additional burdensome notation, we assume that the outcome $Y$ is observed on all individuals, though we note in the Discussion (Section~\ref{sec:disc}) that the extension to right-censored outcomes is straightforward. 
Suppose that we have already observed trials in Populations $1,\ldots,S$. Each trial $s$ consists of $n_s$ independent and identically distributed (i.i.d.) observations $O_s[1],\ldots,O_s[n_s]$ of the structure $O_s\triangleq (W,A,Y)\sim P_s^0$, where $P_s^0$ is known to belong to the model $\mathcal{M}_s$ that at most places restrictions on the conditional distribution of $A$ given $W$ and such that each $P_s'\in\mathcal{M}_s$ satisfies $\delta<\min_{a\in\{0,1\}}P_s'(A=a|W)$ with $P_s^0$ probability one. Thus, we assume that all efficacy trials tested the same vaccine versus control/placebo and collected the same W and Y following the same protocol. While it is not essential that the data from each $P_s^0$ be drawn from a randomized trial, we refer to data from $P_s^0$ as data from ``trial $s$'' because that is our primary application of interest.

From each trial $s$ we have an estimate of the vaccine efficacy, given by $1-\frac{\E_s^0[\E_s^0[Y|A=1,W]]}{\E_s^0[\E_s^0[Y|A=0,W]]}$, where $\E_s^0$ is the expectation operator under $P_s^0$. If $A$ is independent of $W$, then this reduces to estimating a covariate-unadjusted vaccine efficacy. 
Our objective is to estimate the vaccine efficacy in a new population $\star$. Were we to run a trial in this population, we would observe i.i.d. copies of $(W,A,Y)\sim P_\star^{0,F}$, and we could then estimate the vaccine efficacy
\begin{align*}
\Psi(P_\star^{0,F})
\triangleq 1-\frac{\E_\star^{0,F}\left[\E_\star^{0,F}[Y|A=1,W]\right]}{\E_\star^{0,F}\left[\E_\star^{0,F}[Y|A=0,W]\right]},
\end{align*}
where we use $\E_\star^{0,F}$ to denote the expectation operator under $P_\star^{0,F}$. In practice we do not observe a trial from population $\star$, but rather a size $n_\star$ i.i.d. sample of observations $O_\star[1],\ldots,O_\star[n_\star]$ containing only $O_\star\triangleq W\sim P_\star^0$, where $P_\star^0$ is the marginal distribution of $W$ under $P_\star^{0,F}$. We denote the nonparametric model for $P_\star^0$ by $\mathcal{M}_\star$. Our objective is to obtain a lower bound on $\Psi(P_\star^{0,F})$ based on assumptions of how $P_\star^{0,F}$ is related to the distributions $P_1,\ldots,P_S$ from the completed trials. We refer to these assumptions as bridging assumptions.
Because we only make assumptions that are biologically justifiable in a wide variety of examples, we are often unable to identify $\Psi(P_\star^{0,F})$ with any parameter mapping of the collection of distributions $\mathcal{P}^0\triangleq (P_\star^0,P_1^0,\ldots,P_S^0)$. We are, however, able to identify a parameter mapping that lower bounds this quantity under our bridging assumptions. Obtaining such bounds is the objective of the partial identification literature \citep{Manski2003}. A consequence of our partial bridging is that our estimator of the lower bound does not generally converge to $\Psi(P_\star^{0,F})$ even as $n_\star,n_1,\ldots,n_S$ all grow to infinity.

Before proceeding, we introduce some notation. We let $\mathcal{P}^0\triangleq (P_\star^0,P_1^0,\ldots,P_S^0)$ and $n\triangleq n_{\star}+\sum_{s=1}^S n_s$. We treat each of the sample sizes $n_s$ as deterministic. We let $\mathcal{M}\triangleq\mathcal{M}_\star\times\prod_{s=1}^S \mathcal{M}_s$ denote the model for $\mathcal{P}^0$. 
For members $P_s'$ of the general collection $\mathcal{P}'\triangleq\{P_\star',P_1',\ldots,P_S'\}$, we let $\E_s'$ denote the expectation under $\mathcal{P}'$. We also let $\mathcal{S}=\{\star,1,\ldots,S\}$. For Trials $s=1,\ldots,S$, we define the conditional vaccine efficacy as follows
\begin{align*}
\VE_s^0(w)\triangleq 1 - \frac{\E_s^0[Y|A=1,w]}{\E_s^0[Y|A=0,w]}.
\end{align*}
We define the unidentifiable conditional vaccine efficacy in population $\star$ similarly, with the expectations under $P_s^0$ above replaced by expectations under $P_\star^{0,F}$.

We now define a lower bound on the curve $w\mapsto \VE_\star^{0,F}(w)$. This lower bound is defined using data from the completed trials $s=1,\ldots,S$, and also possibly vaccinated/unvaccinated conditional risks from a user-defined pseudo-population, which we denote $s=0$. This pseudo-population can be used to make our vaccine efficacy curve assumption more plausible. Let $\mathbbmss{v}_s : \mathcal{W}\rightarrow [0,1]$, $s=0,1,\ldots,S$, be a collection of functions satisfying the convexity constraint $\sum_{s=0}^S \mathbbmss{v}_s(w)=1$ for all $w$. Let $d : \mathcal{A}\times \mathcal{W}\rightarrow[0,1]$ be a function for which $d(0,w)$ and $d(1,w)$ respectively represent the unvaccinated and vaccinated risk in the user-defined pseudo-trial. Define
\begin{align*}
\VEm^0(w)&\triangleq 1-\frac{\mathbbmss{v}_0(w) d(1,w) + \sum_{s=1}^S \mathbbmss{v}_s(w) \E_s[Y|A=1,w]}{\mathbbmss{v}_0(w) d(0,w) + \sum_{s=1}^S \mathbbmss{v}_s(w) \E_s[Y|A=0,w]}.
\end{align*}
For each $w$, the above expression gives weight $\mathbbmss{v}_s(w)$ for the pseudo-trial $s=0$ and each completed trial $s=1,\ldots,S$. That is, $\mathbbmss{v}_s(w)$ indicates the hypothetical size of the stratum of $w$ in trial $s$ relative to the stratum of $w$ in the other trials and the pseudo-trial. If $\mathbbmss{v}_0\equiv 0$, then the pseudo-trial gets zero weight, whereas if $\mathbbmss{v}_0(w)$ is large for many $w$, then the vaccinated and unvaccinated risks $d(1,w)$ and $d(0,w)$ in the pseudo-trial play a major role in determining $w\mapsto \VEm^0(w)$. Our first bridging assumption states that $w\mapsto \VEm^0(w)$ lower bounds $w\mapsto \VE_\star^{0,F}(w)$.
\begin{enumerate}[label=B\arabic*),ref=B\arabic*,series=en:bridgconds]
	\item\label{it:brVEbridge} $P_\star^0\{\VE_\star^{0,F}(W)\ge \VEm^0(W)\}=1$.
\end{enumerate}
If $\VE_\star^{0,F}\equiv \VE_1^0\equiv \ldots\equiv \VE_S^0$, then the stronger result $P_\star^0\{\VE_\star^{0,F}(W)= \VEm^0(W)\}=1$ holds for any $w\mapsto \VEm^0(w)$ for which the pseudo-trial is given zero weight, i.e. any $\mathbbmss{v}_0,\mathbbmss{v}_1,\ldots,\mathbbmss{v}_S$ satisfying $\mathbbmss{v}_0\equiv 0$. We give a brief derivation of this result in Section~\ref{sec:veconst}.

While the above conditions are helpful for bridging the marginal vaccine efficacy from the completed trials to population $\star$, they do not allow one to uniquely identify this marginal vaccine efficacy, even when the conditional efficacy is point identifiable in the sense that $P_\star^0\{\VE_\star^{0,F}(W)= \VEm^0(W)\}=1$. In particular, we still need to understand the behavior of the conditional risk among unvaccinated individuals in population $\star$. While getting a reasonable estimate of this baseline risk may be difficult in practice, experts may be able to give bounds on how small or large this risk may be. A flexible way of communicating this expert knowledge is as follows. Let $\mathbbmss{u}_0,\mathbbmss{u}_1,\ldots,\mathbbmss{u}_S$ be an expert-specified set of functions mapping from $\mathcal{W}$ to $\mathbb{R}$. Define the upper bound $\upsilon^0\triangleq \Upsilon(\mathcal{P}^0)$ on the unvaccinated risk pointwise by
\begin{align*}
\upsilon^0(w)\triangleq \mathbbmss{u}_0(w) + \sum_{s=1}^S \mathbbmss{u}_s(w) \E_s^0[Y|A=0,w].
\end{align*}
For the lower bound, let $\ell_0,\ell_1,\ldots,\ell_S$ be an expert-specified set of functions mapping from $\mathcal{W}$ to $\mathbb{R}$. The lower bound $\lambda^0\triangleq \Lambda^{\mathcal{P}^0}$ on the unvaccinated risk is defined pointwise by
\begin{align*}
\lambda^0(w)\triangleq \ell_0(w) + \sum_{s=1}^S \ell_s(w) \E_s^0[Y|A=0,w],
\end{align*}
where for each $w$ we require that $\lambda^0(w)\le \upsilon^0(w)$. We assume throughout that $\upsilon^0(w)-\lambda^0(w)$ is uniformly bounded below by some $\delta>0$. Forcing $\delta>0$ will not prove to be restrictive because $\Upsilon$ only needs to serve as an upper bound for the conditional baseline risk (known to fall in the closed unit interval). Having $\delta>0$ ensures that point identification of the vaccine efficacy $\Psi(P_\star^{0,F})$ is impossible when the vaccine efficacy curve is nonconstant because we cannot identify the baseline risk in population $\star$, i.e. $w\mapsto \E_\star^{0,F}[Y|A=0,w]$.

The baseline risk assumption can be written as follows.
\begin{enumerate}[resume*=en:bridgconds]
	\item\label{it:brdatadepub} $\lambda^0(W)\le \E_{\star}^{0,F}[Y|A=0,W]\le \upsilon^0(W)$ with $P_{\star}^0$ probability one.
\end{enumerate}
Our next bridging assumption essentially states that the support of population $\star$ must be contained in the support of the completed trials, though, as we explain below, is slightly more general than this stated condition.
\begin{enumerate}[resume*=en:bridgconds]
	\item\label{it:brcommonsupport} For each $s=1,\ldots,S$ and each event $E$ on $\mathcal{W}$, $P_s^0(E)=0$ implies that $\E_\star^0[\max\{\ell_s(W),\mathbbmss{u}_s(W),\mathbbmss{v}_s(W)\}\Ind_E] = 0$.
\end{enumerate}
The above is weaker than assuming that the marginal distribution of $W$ under $P_s^0$ dominates $P_\star^0$ for each $s=1,\ldots,S$ since $\ell_s(W)$, $\mathbbmss{u}_s(W)$, and $\mathbbmss{v}_s(W)$ may be selected to be $0$ for populations where there is no support. 

We also assume that $w\mapsto \VEm^0(w)$ satisfies a boundedness condition.
\begin{enumerate}[resume*=en:bridgconds]
	\item\label{it:brvebdd} $w\mapsto \VEm^0(w)$ is uniformly bounded on the support of $\mathcal{W}$.
\end{enumerate}
The above ensures that integrals involving $w\mapsto \VEm^0(w)$ make sense -- while it could be replaced by a moment condition, the boundedness assumptions simplifies our proofs and seems to give up little since it is rare that $\VEm^0(w)$ can be arbitrarily close to $-\infty$.

\subsection{Partial bridging formula}
We now present lower bounds on the vaccine efficacy $\Psi(P_\star^{0,F})$ in population $\star$. We first present a loose lower bound that is attainable under only the bridging assumptions discussed thus far, and we then add a final bridging assumption that will often yield a tighter bound.

Given the bridging assumptions presented thus far, the tightest obtainable bound on $\Psi(P_\star^{0,F})$ is given by the solution to the following optimization problem:
\begin{alignat*}{2}
&\textnormal{Minimize }\hspace{0.5em} &&\frac{\E_{\star}^0\left[g(W)\right]}{\E_{\star}^0[f(W)]}\textnormal{ in }f,g : \mathcal{W}\rightarrow[0,1] \\
&\textnormal{subject to }\hspace{0.5em} &&\lambda^0(W)\le f(W)\le \upsilon^0(W),\;\,1-\frac{g(W)}{f(W)}\ge \VEm^0(W),
\end{alignat*}
where we use the convention that $x/0=0$ if $x=0$ and $x/0=+\infty$ if $x>0$. We now argue that the solution to this problem is undesirably loose in general, and that a further restriction will generally yield a tighter bound. For simplicity, we give the argument in the special case where one is unwilling to assume a lower bound on the unvaccinated risk so that $\lambda^0\equiv 0$. We then provide an alternative optimization problem that will typically give a tighter (larger) lower bound, present a closed-form solution to this new optimization problem, and finally we show why the above optimization problem is undesirably loose even in the case that $\lambda^0$ is not the constant function zero. To ease discussion, the remainder of this paragraph supposes that there exists a $w_{\minus}\in\mathcal{W}$ such that $P_\star^0\{\VEm^0(W)=\VEm^0(w_{\minus})\}>0$ and $P_\star^0\{\VEm^0(W)<\VEm^0(w_{\minus})\}=0$, i.e. that $\VEm^0(w_{\minus})$ is equal to the minimal value of $\VEm^0(w)$ on $\mathcal{W}$. In this case one can quickly see that a $w\mapsto f(w)$ optimizing the above is positive at $w_{\minus}$ and is zero everywhere else and the corresponding $w\mapsto g(w)$ is equal to $f(w_{\minus})\VEm^0(w_{\minus})$ at $w_{\minus}$ and $0$ everywhere else, so that at this $f,g$ we have that $1-\frac{\E_{\star}^0\left[g(W)\right]}{\E_{\star}^0[f(W)]}=\VEm^0(w_{\minus})$.

To overcome this problem, we will add an interpretable constraint to the optimization problem. For each $\mu>0$ in the unit interval, we give a lower bound on $\Psi(P_\star^{0,F})$ that is valid if $\mu$ is equal to the marginal unvaccinated risk $\E_\star^0[\E_\star^{0,F}[Y|A=0,W]]$. Each $\mu$-specific lower bound can then be interpreted as valid provided the already stated bridging assumptions hold and $\mu$ is equal to the the marginal unvaccinated risk. We derive our lower bound for $\Psi(P_\star^{0,F})$ by finding the worst-case conditional unvaccinated risk, namely by solving the following optimization problem:
\begin{alignat*}{2}
&\textnormal{Minimize }\hspace{0.5em} &&\frac{\E_{\star}^0\left[g(W)\right]}{\E_{\star}^0[f(W)]}\textnormal{ in }f,g : \mathcal{W}\rightarrow[0,1] \\
&\textnormal{subject to }\hspace{0.5em} &&\E_{\star}^0[f(W)] = \mu,\;\,\lambda^0(W)\le f(W)\le \upsilon^0(W),\;\,1-\frac{g(W)}{f(W)}\ge \VEm^0(W).
\end{alignat*}
The solution to our earlier problem in the special case where $\lambda^0(w)$ is always zero, namely $f(w_{\minus})>0$ and $f(w)=0$ for $w\not=w_{\minus}$, will fail to satisfy the constraint $\mu= \E_{\star}^0[f(W)]$ once $\mu$ is large enough.

We now introduce notation to express a solution to our refined optimization problem. We define all of these quantities at $\mathcal{P}^0$, but the definitions at general $\mathcal{P}'\in\mathcal{M}$ are completely analogous. For each $\beta : \mathcal{W}\rightarrow[0,1]$ and $w\in\mathcal{W}$, let 
\begin{align*}
&\UR^{\beta}(\mathcal{P}^0)(w)\triangleq \lambda^0(w) + \left[\upsilon^0(w)-\lambda^0(w)\right]\beta(w)\triangleq \UR^{\beta,0}(w).
\end{align*}
The set of all unvaccinated risk functions allowed by our bounds $\lambda^0,\upsilon^0$ is equal to $\{\UR^{\beta,0} : \beta\}$. For each $\UR^{\beta,0}$, the marginal unvaccinated risk is given by
\begin{align*}
&\Omega^{\beta}(\mathcal{P}^0)\triangleq \E_{\star}^0\left[\UR^{\beta,0}(W)\right]\triangleq \omega^{\beta,0}.
\end{align*}
We also define
\begin{align*}
\Gamma^\beta(\mathcal{P}^0)\triangleq \E_{\star}^0\left[\UR^{\beta,0}\VEm^0(W)\right]\triangleq \gamma^{\beta,0}.
\end{align*}
Often $\omega^{\beta,0}-\gamma^{\beta,0}$ can be interpreted as a marginal vaccinated risk, though there is not in general any guarantee that this quantity is bounded in $[0,1]$. We discuss this subtlety further following Lemma~\ref{lem:lb}.

We now define $\beta^0 : \mathcal{W}\rightarrow [0,1]$, which we will shortly show to be a $\beta$ indexing the worst-case unvaccinated risk. The function $\beta^0$ assigns unvaccinated risk according to the upper bound $\upsilon^0$ to as many covariate strata with small $\VEm^0(w)$ as allowed by the marginal unvaccinated risk constraint $\mu$. This threshold is defined by
\begin{align*}
&\Theta(\mathcal{P}^0)\triangleq \sup\Big\{\theta\in\mathbb{R} : \omega^{w\mapsto\Ind_{\{\VEm^0(w)<\theta\}},0}\le \mu\Big\}\triangleq \theta^0,
\end{align*}
where $\sup\emptyset = -\infty$ by convention. Let $\beta_\eta\triangleq w\mapsto  \Ind_{\{\VEm^0(w)<\theta^0\}} + \eta \Ind_{\{\VEm^0(w)=\theta^0\}}$, and define $\eta^0$ to be the smallest element of the set $\argmin_{\eta\in[0,1]}\left(\omega^{\beta_\eta,0}-\mu\right)^2$. 
Let $\beta^0\triangleq \beta_{\eta^0}$. One can show that, if $\E_\star^0[\lambda^0(W)]\le \mu\le \E_\star^0[\upsilon^0(W)]$, then $\omega^{\beta^0,0}=\mu$.

Finally, we define the bridging parameter that we will aim to estimate at $\mathcal{P}^0$, and note that the definition at general $\mathcal{P}'$ is entirely analogous:
\begin{align*}
\Phi(\mathcal{P}^0)&\triangleq \frac{\Gamma^{\beta^0}(\mathcal{P}^0)}{\Omega^{\beta^0}(\mathcal{P}^0)} = \frac{\gamma^{\beta^0,0}}{\omega^{\beta^0,0}}\triangleq \phi^0,
\end{align*}
where we remind the reader of the dependence of $\beta^0$ on $\mathcal{P}^0$. We now establish that $\phi^0$ provides a valid lower bound for $\Psi(P_\star^{0,F})$.
\begin{lemma} \label{lem:lb}
Suppose \ref{it:brVEbridge}, \ref{it:brdatadepub}, \ref{it:brcommonsupport}, and \ref{it:brvebdd} hold. If $\mu=\E_\star^0[\E_\star^{0,F}[Y|A=0,W]]$, then $\Psi(P_\star^{0,F})\ge \phi^0$.
\end{lemma}
The \hyperlink{proof:lemlb}{proof of Lemma~\ref*{lem:lb}} is given in Appendix~\ref{app:lb}. Typically the lower bound $\phi^0$ will be attainable, in the sense that there exists a distribution for population $\star$ satisfying \ref{it:brVEbridge}, \ref{it:brdatadepub}, and $\mu=\E_\star^0[\E_\star^{0,F}[Y|A=0,W]]$ such that the marginal vaccine efficacy is equal to $\phi^0$. In this case $\phi^0$ is the solution to the refined optimization problem presented earlier in this section. In particular, it will often be the case that a distribution with marginal vaccine efficacy $\phi^0$ is that with unvaccinated risk $\E_\star^{0,F}[Y|A=0,w] = \UR^{\beta^0,0}(w)$ and vaccine efficacy curve $w\mapsto \VEm^0(w)$. The only time no such distribution exists is when the claimed vaccinated risk $w\mapsto \UR^{\beta^0,0}(w)[1-\VEm^0(w)]$ fails to obey the bounds of the model, which may be larger than $1$ if (i) $\UR^{\beta^0,0}(w)$ is large and $1-\VEm^0(w)>0$ or (ii) $1-\VEm^0(w)$ is large. In practice it is unlikely that either of these quantities will be large: the unvaccinated risk will generally be small for rare outcomes, and the vaccine efficacy within a biomarker stratum will rarely be extremely negative, as this indicates that the vaccine is extremely harmful within this stratum and represents a situation where the bridging application would likely not be of interest and hence the method would not be applied.

\begin{remark}
Returning to our first optimization problem in this section, i.e. the optimization problem that did not include the constraint $\E_\star^0[f(W)]=\mu$, we see that the minimum occurs by minimizing the solution to the latter problem over all values of $\mu>0$. We then see that the solution to the first optimization problem is equal to the lowest point on our $\mu$-dependent curve, so that indeed the added restriction will generally improve our bound. For this observation to be true, one might anticipate that the reason that the solution to the first optimization problem is $\VEm(w_{\minus})$ in the case where the unvaccinated risk lower bound is the constant function zero, i.e. $\lambda^0\equiv 0$, is that $\phi_{\mu}^0$ is monotonic in $\mu$, where we write $\phi_{\mu}^0$ to emphasize the dependence of $\phi^0$ on $\mu$. Indeed, we show in Lemma~\ref{lem:lbmon} in Appendix~\ref{app:lb} that $\phi_\mu^0$ is monotonically nondecreasing in $\mu$ in this case. This observation enriches the interpretation of the lower bound $\phi_{\mu'}^0$ when $\lambda^0\equiv 0$, since in this case $\phi_{\mu'}^0$ can be interpreted as a valid lower bound for the marginal vaccine efficacy provided $\mu'\le \mu=\E_\star^0[\E_\star^{0,F}[Y|A=0,W]]$ and the conditions of Lemma \ref{lem:lb} hold at $\mu$: as $\phi_{\mu'}^0\le \phi_{\mu}^0$ and $\phi_{\mu}^0$ lower bounds the vaccine efficacy, $\phi_{\mu'}^0$ must also lower bound the vaccine efficacy. Thus, in this special case, $\phi_{\mu'}^0$ is a valid lower bound whenever $\mu'\le \E_\star^0[\E_\star^{0,F}[Y|A=0,W]]$.
\end{remark}

\section{First-order expansion of lower bound} \label{sec:monotone}
In this section, we present a result that we will use to derive a first-order expansion of the parameter $\Phi$. This expansion plays a key role in our estimation procedure. Before presenting this result, we quickly define a gradient for a general parameter $\Pi : \mathcal{M}\rightarrow\mathbb{R}$ in the $S+1$ sample problem. We only define this and other gradients in this section at $\mathcal{P}^0$, but the extension to a general $\mathcal{P}'$ only requires notational changes. For each $s\in\mathcal{S}$, let $h_s : \mathcal{O}_s\rightarrow\mathbb{R}$ satisfy $\E_s[h_s(O_s)]=0$ and $\sup_{o_s} |h_s(o_s)|\le 1$. For $\epsilon\in(-1,1)$, define $\frac{dP_s^\epsilon}{dP_s^0}(o_s) = 1+\epsilon h_s(o_s)$, $s\in\mathcal{S}$. Let $\mathcal{P}^\epsilon\triangleq \left(P_{\star}^\epsilon,P_1^\epsilon,\ldots,P_S^\epsilon\right)$. We call $\Pi$ pathwise differentiable if, for $s\in\mathcal{S}$, there exist functions $\nabla \Pi_s^0\in L^2(P_s^0)$ satisfying
\begin{align*}
\left.\frac{d}{d\epsilon}\Pi(\mathcal{P}^\epsilon)\right|_{\epsilon=0}&= \sum_{s\in\mathcal{S}} \int \nabla \Pi_s^0(o_s) h_s(o_s) dP_s^0(o_s).
\end{align*}

Our pathwise differentiability result will require one of the following three regularity conditions on the marginal distribution of $W$ in population $\star$. The first is given below.
\begin{enumerate}[label=C\arabic*),ref=C\arabic*,series=en:regconds]
	\item\label{it:omegacontnotflat} $\E_\star^0[\lambda^0(W)]< \mu < \E_\star^0[\upsilon^0(W)]$ and, for all $\theta$ in a neighborhood of $\theta^0$,
	\begin{align*}
	0&<\liminf_{t\rightarrow 0} \frac{P_{\star}^0\{\VEm^0(W)< \theta+t\}-P_{\star}^0\{\VEm^0(W)< \theta\}}{t} \\
	&\le \limsup_{t\rightarrow 0} \frac{P_{\star}^0\{\VEm^0(W)< \theta+t\}-P_{\star}^0\{\VEm^0(W)< \theta\}}{t} <\infty.
	\end{align*}
\end{enumerate}
The above implies both that $\theta^0$ is finite and that $\omega^{\beta^0,0}=\mu$.

For the choice of $\mu$ to be feasible, i.e. for it to be possible that $P_\star^{0,F}\in\mathcal{M}_\star$ satisfies both \ref{it:brdatadepub} and $\E_\star^0[\E_\star^{0,F}[Y|A=0,W]]=\mu$, we generally need that $\E_\star^0[\lambda^0(W)]\le \mu\le \E_\star^0[\upsilon^0(W)]$: one cannot otherwise have both $\E_\star^0[\E_\star^{0,F}[Y|A=0,W]]=\mu$ and $\lambda^0(w)\le \E_\star^{0,F}[Y|A=0,w]\le \upsilon^0(w)$ for all $w$. Nonetheless, it is useful to understand the first-order behavior of the parameter $\Phi$ for all values of $\mu$ since the marginal distribution of $W\sim P_\star^0$ must be estimated in practice. Thus, we offer conditions for both the case that $\mu$ is so large that it violates the upper bound on $\E_\star^{0,F}[Y|A=0,w]$, i.e. $\mu > \E_\star^0[\upsilon^0(W)]$, and then for the case that $\mu$ is so small that it violates the lower bound on $\E_\star^{0,F}[Y|A=0,w]$, i.e. $\mu < \E_\star^0[\lambda^0(W)]$. 
\begin{enumerate}[label=\ref*{it:omegacontnotflat}'),ref=\ref*{it:omegacontnotflat}'] 
	\item\label{it:mubig} The upper bound is too small: $\E_\star^0[\upsilon^0(W)]<\mu$.
\end{enumerate}
\begin{enumerate}[label=\ref*{it:omegacontnotflat}''),ref=\ref*{it:omegacontnotflat}'']
	\item\label{it:musmall} The lower bound is too large: $\E_\star^0[\lambda^0(W)]>\mu$.
\end{enumerate}
\begin{remark}
None of these three conditions (\ref{it:omegacontnotflat}, \ref{it:mubig}, \ref{it:musmall}) allow $P_\star^0(\VEm^0(W)=\theta^0)>0$. This is closely related to the non-pathwise differentiability of many parameters of interest in personalized medicine under so-called exceptional laws \citep{Robins04,Luedtke&vanderLaan2015b}, i.e. distributions for which the conditional average treatment effect is zero in some positive probability stratum of measured covariates.  The condition that $P_\star^0\{\VEm(W)=\theta^0\}=0$ may be unlikely to hold in settings where $\theta^0$ is below the lower limit of quantification of $W$, namely because small values of $W$ indicate a true value of zero for the biomarker and any deviation from zero is due to noise. If $\theta^0$ falls below the lower limit if quantification, then we expect that $P_\star^0\{\VEm^0(W)=\theta^0\}>0$ and indeed we will have no guarantee that $\Psi(P_\star^{0,F})\ge \phi^0$. In Appendix~\ref{app:mono}, we describe an alternative bridging parameter that requires that $W$ is univariate and $\VEm^0$ is monotonic rather than that $P_\star^0\{\VEm^0(W)=\theta^0\}>0$.

Finally, we note that, while the conditions of Theorem~\ref{thm:pd} are sufficient for the pathwise differentiability of $\Phi$, they are not necessary. For example, if $W$ is discrete and takes on a finite number of values, then one can give conditions under which $\Phi$ is pathwise differentiable even if $P_\star^0(\VEm^0(W)=\theta^0)>0$, though these conditions still appear to require that only one $w\in\mathcal{W}$ satisfies $\VEm^0(w)=\theta^0$. We do not consider this subtle case further in this work.  
\end{remark}

The following objects, defined for each $s=1,\ldots,S$, will be useful for establishing the gradient of $\Phi$:
\begin{align}
&D_{\UR,s}^{\beta,\mathcal{P}^0}(o_s)\triangleq \left[\ell_s(w) + \beta(w)\{\mathbbmss{u}_s(w)-\ell_s(w)\}\right]\frac{\Ind_{\{a=0\}}}{P_s^0(a|w)}\left(y - \E_s^0[Y|a,w]\right),\;\beta : \mathcal{W}\rightarrow\mathbb{R}, \nonumber \\
&D_{\VE,s}^{\mathcal{P}^0}(o_s)\triangleq \frac{\mathbbmss{v}_s(w)\left[\Ind_{\{a=0\}}\VEm^0(w)-\Ind_{\{a=1\}}\right]}{P_s^0(a|w)\left[\mathbbmss{v}_0(w)d(0,w) + \sum_{s'=1}^S \mathbbmss{v}_{s'}(w) \E_{s'}^0[Y|A=0,w]\right]}(y-\E_s^0[Y|a,w]). \label{eq:DVEdef}
\end{align}
For ease of notation, we let $D_{\UR,s}^{\beta,0}\triangleq D_{\UR,s}^{\beta,\mathcal{P}^0}$ and $D_{\VE,s}^0\triangleq D_{\VE,s}^{\mathcal{P}^0}$.
For any $\beta : \mathcal{W}\rightarrow[0,1]$, define
\begin{alignat*}{2}
&\nabla \Omega_\star^{\beta,0}(o_\star)= \UR^{\beta,0}(w) - \omega^{\beta,0},& \\
&\nabla \Omega_s^{\beta,0}(o_s)= \frac{dP_{\star}^0}{dP_s^0}(w)D_{\UR,s}^{\beta,0}(o_s),\;&s=1,\ldots,S, \\
&\nabla \Gamma_\star^{\beta,0}(w)\triangleq \UR^{\beta,0}(w)\VEm^0(w)-\Gamma^{\beta}(\mathcal{P}^0),& \\
&\nabla \Gamma_s^{\beta,0}(o_s)\triangleq \frac{dP_{\star}^0}{dP_s^0}(w)\left[D_{\UR,s}^{\beta,0}(o_s)\VEm^0(w)+\UR^{\beta,0}(w)D_{\VE,s}^0(o_s)\right],\;&s=1,\ldots,S.
\end{alignat*}
We now give a theorem establishing that the parameter $\Phi$ is pathwise differentiable at $\mathcal{P}^0$.
\begin{theorem} \label{thm:pd}
If \ref{it:brcommonsupport} and \ref{it:brvebdd} hold and either \ref{it:omegacontnotflat}, \ref{it:mubig}, or \ref{it:musmall} holds, then $\Phi$ is pathwise differentiable and, for each $s\in\mathcal{S}$, the $P_s^0$ gradient is given by
\begin{align*}
\nabla \Phi_s^0(o_s)=& \begin{cases}
\dfrac{\nabla \Gamma_s^{\beta^0,0}(o_s)}{\omega^{\beta^0,0}} - \theta^0\dfrac{\nabla \Omega_s^{\beta^0,0}(o_s)}{\omega^{\beta^0,0}},&\mbox{ if }\E_\star^0[\lambda^0(W)]<\mu<\E_\star^0[\upsilon^0(W)], \\[1.3em]
\dfrac{\nabla \Gamma_s^{\beta^0,0}(o_s)}{\omega^{\beta^0,0}} - \gamma^{\beta^0,0}\dfrac{\nabla \Omega_s^{\beta^0,0}(o_s)}{[\omega^{\beta^0,0}]^2},&\mbox{ otherwise.}
\end{cases}.
\end{align*}
\end{theorem}
The \hyperlink{proof:thmpd}{proof of Theorem~\ref*{thm:pd}} is given in Appendix~\ref{app:pd}.

\begin{remark}
Though we have assumed that $\upsilon^0(w)-\lambda^0(w)$ is bounded away from zero, we will now briefly remark on a violation of this assumption, namely the case that $\ell_s\equiv \mathbbmss{u}_s$ so that $\upsilon^0\equiv \lambda^0$. In this case the gradients of $\Phi$ are given by the same expression as in Theorem~\ref{thm:pd} for the case where $\theta^0=+\infty$ so that $P_\star^0(\VEm^0(W)<\theta^0)=1$. The estimation scheme that we present in the next section will remain valid even when $\upsilon^0\equiv \lambda^0$ provided one uses the appropriate gradients, namely the gradients from Theorem~\ref{thm:pd} when $\theta^0=+\infty$, when constructing the confidence lower bound. If $\VE_\star^{0,F}=\VEm^0$, then the marginal vaccine efficacy in population $\star$ is point identified, and so our procedure is analogous to the earlier work by \cite{Rudolph&vanderLaan2016}, with the slight distinction that we focus on a multiplicative rather than additive parameter.
\end{remark}

\section{Estimation} \label{sec:est}
\subsection{Conditional vaccine efficacy not necessarily constant across $s=1,\ldots,S$} \label{sec:nonconstantVE}
We first consider the case where we do not assume that $\VE_1^0\equiv \ldots\equiv \VE_S^0$, but rather make the weaker assumption \ref{it:brVEbridge} for a prespecified set of functions $\mathbbmss{v}_0,\mathbbmss{v}_1,\ldots,\mathbbmss{v}_S$.

\subsubsection{Estimator.} \label{sec:eststeps}
Below and throughout, we let $Q_s^n$ denote the empirical distribution of the observations $O_s[1],\ldots,O_s[n_s]$ for each $s\in\mathcal{S}$. We also denote the expectation operator under $\mathcal{P}^n$ by $\E_s^n$. We also refer to parameters evaluated at the collection of distributions $\mathcal{P}^n$ rather than parameters evaluated at $\mathcal{P}^0$ by replacing the superscript zero by superscript $n$, e.g. we write $\upsilon^n\triangleq \Upsilon(\mathcal{P}^n)$ rather than $\upsilon^0\triangleq \Upsilon(\mathcal{P}^0)$. We do the same for gradients, e.g. $\nabla \Phi_s^n$ rather than $\nabla \Phi_s^0$. When we define objects that may be confused with parameter mappings applied directly to $\mathcal{P}^n$, we add a hat: in particular, we will define $\widehat{\beta}^n$, $\widehat{\eta}^n$, $\widehat{\theta}^n$, $\widehat{\omega}^{\widehat{\beta}^n,n}$, $\widehat{\gamma}^{\widehat{\beta}^n,n}$, and $\widehat{\phi}^n$.

Below we present an estimation scheme to be used when either (i) the chosen $\ell_s(w)$ is a constant multiple of the chosen $\mathbbmss{u}_s(w)$, where the multiple is independent of $s$ and $w$, or (ii) the chosen $\ell_s\equiv 0$ for all $s$ so that $\lambda^0\equiv 0$. In Appendix~\ref{app:altalg}, we present alternatives to the targeted minimum loss-based (TMLE) steps \ref{it:tmle1}, \ref{it:tmle2}, and \ref{it:tmle3} that replace the univariate logistic regression by a bivariate logistic regression. The bivariate regression is inappropriate when (i) holds because the two predictors in the proposed logistic regression will be perfectly correlated, and so, while analytically correct, the method may encounter numerical challenges in practice. When (ii) holds running the bivariate regression presented in the appendix is simply unnecessary because the latter covariate is always zero. The proof of asymptotic linearity of our estimation scheme assumes that the user runs the appropriate estimation scheme, either the below or that in Appendix~\ref{app:altalg}.
\begin{enumerate}
	\item\label{it:initests} Let $(a,w)\mapsto \E_s^{n,\textnormal{init}}[Y|a,w]$, $w\mapsto P_s^{n,\textnormal{init}}(A=1|w)$, and $w\mapsto \frac{dP_\star^{n,\textnormal{init}}}{dP_s^{n,\textnormal{init}}}(w)$ represent estimates of $(a,w)\mapsto \E_s^0[Y|a,w]$, $w\mapsto P_s^0(A=1|w)$, and $w\mapsto \frac{dP_\star^0}{dP_s^0}(w)$, respectively.
	\item\label{it:tmle1} Fit a univariate logistic regression with outcome $\left(y_s[i] : s=1,\ldots,S;\,i=1,\ldots,n_s\right)$, covariate $\left(\frac{\Ind_{\{a_s[i]=0\}}\mathbbmss{u}_s(w)}{n_s P_s^{n,\textnormal{init}}(A=0|w_s[i])}\frac{dP_\star^{n,\textnormal{init}}}{dP_s^{n,\textnormal{init}}}(w_s[i]) : s=1,\ldots,S;\,i=1,\ldots,n_s\right)$, and fixed, subject-level intercept $\left(\logit\left(\E_s^{n,\textnormal{init}}[Y|A=0,w_s[i]]\right) : s=1,\ldots,S;\,i=1,\ldots,n_s\right)$. Denote the fitted coefficient in front of the covariate by $\epsilon_n$.
	\item\label{it:tmle2} For each $s=1,\ldots,S$, let $(a,w)\mapsto \E_s^{n,\epsilon_n}[Y|a,w]$ denote the function
	\begin{align*}
	(a,w)\mapsto \logit^{-1}\left[\logit\left(\E_s^{n,\textnormal{init}}[Y|a,w]\right) + \epsilon_n\frac{\Ind_{\{a=0\}}\mathbbmss{u}_s(w)}{n_s P_s^{n,\textnormal{init}}(A=0|w)}\frac{dP_\star^{n,\textnormal{init}}}{dP_s^{n,\textnormal{init}}}(w)\right].
	\end{align*}
	\item\label{it:tmle3} Let $\mathcal{P}^n=(P_\star^n,P_1^n,\ldots,P_s^n)$ denote any collection of distributions satisfying that, for all $(a,w)$, $\E_s^n[Y|a,w] = \E_s^{n,\epsilon_n}[Y|a,w]$, $P_s^n(A=1|w) = P_s^{n,\textnormal{init}}(A=1|w)$, and $\frac{dP_\star^n}{dP_s^n}(w) = \frac{dP_\star^{n,\textnormal{init}}}{dP_s^{n,\textnormal{init}}}(w)$.
	\item\label{it:omegan} For each $\beta : \mathcal{W}\rightarrow[0,1]$, let $\widehat{\omega}^{\beta,n}\triangleq \omega^{\beta,n} + \sum_{s\in\mathcal{S}} Q_s^n \nabla  \Omega_s^{\beta,n}$, and note that $\widehat{\omega}^{\beta,n}$ rewrites as $Q_\star^n \UR^{\beta,n} + \sum_{s=1}^S Q_s^n \nabla \Omega_s^{\beta,n}$.
	\item\label{it:thetan} Let $\widehat{\theta}^n\triangleq\sup\{\theta : \widehat{\omega}^{w\mapsto \Ind_{\{\VEm^n(w)<\theta\}},n} \le \mu\}$, where $\sup\emptyset=-\infty$.
	\item\label{it:etan} Let $\widehat{\eta}^n$ be the smallest element of the set $\argmin_{\eta\in[0,1]}\left(\widehat{\omega}^{\beta_\eta,n}-\mu\right)^2$, where $\beta_\eta\triangleq w\mapsto \Ind_{\{\VEm^n(w)<\widehat{\theta}^n\}} + \eta \Ind_{\{\VEm^n(w)=\widehat{\theta}^n\}}$.
	\item\label{it:betan} Let $\widehat{\beta}^n\triangleq \beta_{\widehat{\eta}^n}$.
	\item\label{it:onestep} Estimate $\gamma^{\widehat{\beta}^n,0}$ with
	\begin{align*}
	\widehat{\gamma}^{\widehat{\beta}^n,n}&\triangleq \gamma^{\widehat{\beta}^n,n} + \sum_{s\in\mathcal{S}} Q_s^n \nabla \Gamma_s^{\widehat{\beta}^n,n} \\
	&= n_\star^{-1}\sum_{i=1}^{n_\star} \UR^{\widehat{\beta}^n,n}(w_\star[i])\VE^n(w_\star[i]) + \sum_{s=1}^S Q_s^n \nabla \Gamma_s^{\widehat{\beta}^n,n}.
	\end{align*}
	\item\label{it:phin} Estimate $\phi^0$ with $\widehat{\phi}^n\triangleq \frac{\widehat{\gamma}^{\widehat{\beta}^n,n}}{\widehat{\omega}^{\widehat{\beta}^n,n}}$.
\end{enumerate}
The initial estimates in Step~\ref{it:initests} can be obtained by any methods deemed appropriate by the investigators. We encourage the use of data adaptive regression and machine learning techniques to estimate these features of $\mathcal{P}^0$, since the study of our estimator in the next section relies on consistency conditions given in Appendix~\ref{app:est} that are most plausible when one does not restrict the estimates in Step~\ref{it:initests} to those obtained by classical parametric models. 

\begin{remark}
The above represents a hybrid between a TMLE and a one-step estimator. When $\theta^0=\pm \infty$, it uses a TMLE to estimate $\omega^{\beta^0,0}$ (Steps \ref{it:tmle1}, \ref{it:tmle2}, \ref{it:tmle3}). It uses a one-step estimator for $\gamma^{\widehat{\beta}^n,0}$ (Step \ref{it:onestep}) and to correct the bias of $\omega^{\beta,n}$ for $\omega^{\beta,0}$ when defining $\widehat{\beta}^n$ (Step \ref{it:omegan}). For Step \ref{it:onestep}, one could specify a fluctuation of the initial estimate of $\mathcal{P}^0$ with score given by the sum of the empirical means of the $S+1$ gradients to replace the one-step estimator by a TMLE to have a full-fledged TMLE. One could analyze this estimator similarly to that of our hybrid estimator \citep{vanderLaan&Rose11}. The TMLE is preferred to a one-step estimator in many problems because it yields a plug-in estimator of the form $\gamma^{\widehat{\beta}^n,n}$ for some distribution $\mathcal{P}^n$, and is thus guaranteed to respect the known bounds on $\gamma^{\widehat{\beta}^n,0}$. Most fluctuation submodels indexing the TMLE appear to require iterating the fluctuation step until convergence. Alternatively, one could use the newly developed universal least favorable submodel, which has the advantage of always being fit in one step \citep{vdL&Gruber2016}. We leave the development of full-fledged (non-hybrid) TMLEs for this problem to future work.
\end{remark}

\begin{remark}
The upper and lower bounds $\upsilon^0$ and $\lambda^0$ on the conditional unvaccinated risk imply bounds on the plausible range of values of $\mu$. These bounds can be estimated using $\widehat{\omega}^{\beta,n}$, where $\beta$ is taken to be the constant function returning one and zero, respectively. Under the same conditions as those given in the upcoming Theorem~\ref{thm:al}, one can show that, for each of these choices of $\beta$, $\widehat{\omega}^{\beta,n}-\omega^{\beta,0}\approx \sum_{s\in\mathcal{S}} (Q_s^n-P_s^0) \nabla \Omega_s^{\beta,0}$, so that confidence intervals for $\upsilon^0$ and $\lambda^0$ can be developed using the analogous asymptotic normality results to those to be used to develop a confidence lower bound for $\phi^0$ in Section~\ref{sec:conflb}.
\end{remark}

\subsubsection{Asymptotic linearity of our estimator.}
In the appendix we give several additional assumptions used to establish that the estimator $\widehat{\phi}^n$ satisfies a multiple sample version of asymptotic linearity, namely that
\begin{align}
\widehat{\phi}^n-\phi^0&= \sum_{s\in\mathcal{S}} (Q_s^n-P_s^0) \nabla \Phi_s^0 + o_P(\nmin^{-1/2}), \label{eq:multisampleal}
\end{align}
where, for $\nmin\triangleq \min\{n_s : s\in\mathcal{S}\}$, we use $o_P(\nmin^{-1/2})$ to denote a random variable $X_n$ that is a function of all of the $n$ observations from Populations $\star,1,\ldots,S$ with the property that, for each $\delta>0$, the event $\{\nmin^{1/2}X_n>\delta\}$ has probability converging to zero whenever $\nmin$ grows to infinity. We discuss those conditions below the theorem.

The following theorem establishes that our estimator minus the truth behaves as a sum of empirical means across the $S+1$ data sets.
\begin{theorem}\label{thm:al}
If \ref{it:brcommonsupport}, \ref{it:brvebdd}, \ref{it:CSconsistency}, \ref{it:thetaDonsker}, \ref{it:thetaempproc}, and \ref{it:VEnbdd} hold and either \ref{it:omegacontnotflat}, \ref{it:goodquantile}, and \ref{it:betangood} hold, \ref{it:mubig} holds, or \ref{it:musmall} holds, then $\widehat{\phi}^n$ satisfies (\ref{eq:multisampleal}).
\end{theorem}
The \hyperref[app:thmalproof]{proof of Theorem~\ref*{thm:al}} is given in Appendix~\ref{app:thmalproof}. We now summarize the additional conditions used in the theorem. Condition~\ref{it:CSconsistency} requires that the initial estimate of $\mathcal{P}^0$ is consistent and that certain remainder terms converge at a sufficient rate, namely $o_P(\nmin^{-1/2})$. Only requiring remainder terms to behave like $o_P(\nmin^{-1/2})$ makes our conditions on estimates of parameters from a full efficacy trial $s$ more plausible if $n_s\gg \nmin$, since we typically expect these parameters to converge at an $n_s^{-1/2}$ rate in a parametric model and a slower rate in a nonparametric model, but the slower rate may still be much faster than $\nmin^{-1/2}$. Furthermore, the fact that $W$ is univariate leads us to believe that, under some smoothness, we expect the nonparametric rate for estimating the needed functionals of $\mathcal{P}^n$ to be only slightly slower than the parametric rate. Conditions \ref{it:thetaDonsker} and \ref{it:thetaempproc} represent an empirical process condition and mild consistency condition on the initial estimate $\mathcal{P}^n$. Condition~\ref{it:goodquantile} ensures that the estimate $\widehat{\beta}^n$ of $\beta^0$ satisfies $\widehat{\omega}^{\widehat{\beta}^n,n}\approx \mu$, where $\widehat{\omega}^{\widehat{\beta}^n,n}$ is a one-step estimate of $\omega^{\widehat{\beta}^n,0}$. This is to be expected under \ref{it:omegacontnotflat} given that $\omega^{\beta^0,0}=\mu$ in this case, much as we expect the empirical cumulative distribution function evaluated at the sample median to be equal to $1/2 + o_P(n^{-1/2})$. In fact, we often expect $\widehat{\omega}^{\widehat{\beta}^n,n}$ to be exactly equal to $\mu$, and this can be formally checked in any given application by looking at $\widehat{\omega}^{\widehat{\beta}^n,n}$. Condition~\ref{it:betangood} requires that $\VEm^n$ estimates $\VEm^0$ sufficiently well so that $\UR^{\widehat{\beta}^n,0}$ is approximately equal to the worst-case unvaccinated risk function $\UR^{\beta^0,0}$. Appendix~\ref{app:moreinterpretablebetangood} relates this condition to the margin conditions appearing in the classification literature. Condition~\ref{it:VEnbdd} simply states that the estimate of the vaccine efficacy curve $w\mapsto \VEm^0(w)$ is bounded away from negative infinity, which is plausible given that $w\mapsto \VEm^0(w)$ is bounded under \ref{it:brvebdd}.

\subsubsection{Constructing a lower confidence bound.} \label{sec:conflb}
We now describe how Theorem \ref{thm:al} allows the construction of a lower confidence bound for the vaccine efficacy lower bound $\phi^0$. For all $s\in\mathcal{S}$, let $\sigma_s^2\triangleq \E_s^0[\nabla \Phi_s^0(O)^2]$, i.e. the variance of $\nabla \Phi_s^0(O)$ when $O\sim P_s^0$. Accepting slight notational overload, we define $\sigma_n^2\triangleq \sum_{s\in\mathcal{S}} \frac{\nmin}{n_s} \sigma_s^2$. We propose using $\widehat{\phi}^n - 1.64 \nmin^{-1/2}\hat{\sigma}_n$ as a lower confidence bound for $\phi^0$, where $\hat{\sigma}_n\triangleq \sum_{s\in\mathcal{S}} \frac{\nmin}{n_s} \hat{\sigma}_s^2$ for empirical estimates $\hat{\sigma}_s^2\triangleq Q_s^n [\nabla \Phi_s^n-Q_s^n \nabla \Phi_s^n]^2$ of $\sigma_s^2$. Our lower confidence bound on $\phi^0$ deviates from $\widehat{\phi}^n$ by the order $\nmin^{-1/2}\sigma_n$. As $\sigma_n<\sum_{s\in\mathcal{S}} \sigma_s^2<\infty$ and $\widehat{\phi}^n$ deviates from $\phi^0$ on the order $O(\nmin^{-1/2})$ by our asymptotic linearity result, the deviation of our lower confidence bound from $\phi^0$ is of the order of the inverse square root of the smallest trial. Given that the trials in Populations $1,\ldots,S$ will often have been efficacy studies, we will often expect the smallest sample to be of subjects drawn from population $\star$.

We can repeat the above procedure for a range of plausible $\mu$ values and report a graph of how the lower bound varies with the choice of $\mu$. We remind the reader that the estimate $\widehat{\phi}^n$ and each $\hat{\sigma}_s^2$ rely on $\mu$, and so these values must be recalculated for each choice of $\mu$. 
Figure~\ref{fig:avglb} from our simulation study suggests an informative way to display the vaccine efficacy lower bound, with the ``average lower bound'' $y$-axis (averaged across Monte Carlo replications of our simulation) replaced by the estimated lower bound. Given a range $\mathcal{U}$ of plausible values for $\mu$, one can read a lower uncertainty interval \citep{Vansteelandtetal2006} for the marginal vaccine efficacy $\Psi(P_\star^{0,F})$ in population $\star$ by looking at the smallest $\mu$-specific confidence lower bound across all $\mu\in\mathcal{U}$.

We now justify our lower confidence bound. Note that (\ref{eq:multisampleal}) rewrites as
\begin{align}
\frac{\sqrt{\nmin}[\widehat{\phi}^n-\phi^0]}{\sigma_n}&= \sum_{s\in\mathcal{S}} \sum_{i=1}^{n_s} \frac{\sqrt{\nmin}}{n_{s} \sigma_n} \nabla \Phi_s^0(O_s[i]) + o_P(1). \label{eq:normalizedal}
\end{align}
For each $s$, observe that
\begin{align*}
\Var_{s}^0\left[\frac{\sqrt{\nmin}}{n_{s} \sigma_n} \nabla \Phi_s(O_s[i])\right]&= \frac{1}{n_s} \frac{n_{s}^{-1}\sigma_{s}^2}{[n_\star^{-1}\sigma_\star^2 + \sum_{s'=1}^S n_{s'}^{-1} \sigma_{s'}^2]}.
\end{align*}
The latter fraction on the right-hand side is bounded between $0$ and $1$. Thus the variance on the left shrinks at rate $O(n_s^{-1})$. The above and the independence of all of the observations also shows that the double sum on the right-hand side of (\ref{eq:normalizedal}) has variance $1$.

These two facts readily allow one to show that the conditions of the Lindeberg central limit theorem hold, thereby establishing that $\frac{\sqrt{\nmin}[\widehat{\phi}^n-\phi^0]}{\sigma_n}$ converges to a standard normal distribution \citep{Billingsley1999}. It follows that a valid 95\% level lower confidence bound for $\phi^0$ is given by $\widehat{\phi}^n - 1.64 \nmin^{-1/2}\sigma_n$. In practice, $\sigma_n$ is not known, but we can estimate it with $\hat{\sigma}_n$. Glivenko-Cantelli conditions on the functions $\nabla \Phi_s^n$ ensure that $\hat{\sigma}_n^2\rightarrow \sigma_n^2$ in probability, and under this convergence Slutsky's theorem then ensures that a 95\% lower bound for $\phi^0$ is given by $\widehat{\phi}^n - 1.64 \nmin^{-1/2}\hat{\sigma}_n$.

\subsubsection{Intuition for the estimator's construction and sketch of asymptotic linearity proof.}
Our estimator is a hybrid of a TMLE and a one-step estimator. Before discussing our specific estimator, we give some intuition by discussing one-step estimation of a general pathwise differentiable parameter $\Pi : \mathcal{M}\rightarrow \mathbb{R}$. Suppose that $\mathcal{P}^n$ is an initial estimate of $\mathcal{P}^0$ and we wish to estimate $\pi^0\triangleq \Pi(\mathcal{P}^0)$. One possible estimate is obtained via plug-in estimation: $\pi^n\triangleq \Pi(\mathcal{P}^n)$. If $\mathcal{P}^n$ is a good estimate of $\mathcal{P}^0$, often in the sense that certain regression functions under $\mathcal{P}^n$ are close to the corresponding regressions under $\mathcal{P}^0$, then we generally expect that $\pi^n - \pi^0\approx -\sum_{s\in\mathcal{S}} P_s^0 \nabla \Pi_s^n$, where $\nabla \Pi_s^n$ are gradients of $\Pi$ evaluated at $\mathcal{P}^n$. It will often be the case that the right-hand side of this approximation is not $O_P(\nmin^{-1/2})$, and thus that $\pi^n$ converges to $\pi^0$ at a slower than $\nmin^{-1/2}$ rate. In an effort to improve the initial estimate $\pi^n$, a one-step estimator adds the empirical version of $\sum_{s\in\mathcal{S}} P_s^0 \nabla \Pi_s^n$ to the initial estimate $\pi^n$, i.e. one uses as estimate $\widehat{\pi}^n\triangleq \pi^n + \sum_{s\in\mathcal{S}} Q_s^n \nabla \Pi_s^n$, so that $\widehat{\pi}^n - \pi^0\approx \sum_{s\in\mathcal{S}} (Q_s^n-P_s^0) \nabla \Pi_s^n$. Under empirical process conditions, one can replace the gradients $\nabla \Pi_s^n$ under $\mathcal{P}^n$ by the gradients $\nabla \Pi_s^0$ under $\mathcal{P}^0$ so that $\widehat{\pi}^n - \pi^0\approx \sum_{s\in\mathcal{S}} (Q_s^n-P_s^0) \nabla \Pi_s^0$ and it is possible to apply a central limit theorem to understand the behavior of a scaled version of the one-step estimator minus the truth. Typically TMLEs follow a similar derivation, but the initial estimate $\mathcal{P}^n$ of $\mathcal{P}^0$ is carefully selected so that the bias correction $\sum_{s\in\mathcal{S}} Q_s^n \nabla \Pi_s^0$ is small or zero.

The remainder of the discussion is focused on our proposed estimator. Steps \ref{it:tmle1}, \ref{it:tmle2}, and \ref{it:tmle3} are designed to ensure that the score $\sum_{s=1}^S Q_s^n \nabla \Omega_s^{\beta,n}$ is equal to zero when $\beta\equiv 1$, which is an important score equation to solve if $\widehat{\theta}^n=+\infty$. If each $\ell_s(w)$ is a constant multiple of $\mathbbmss{u}_s(w)$ or if each $\ell_s\equiv 0$, then the same score equation is equal to zero when $\beta\equiv 0$, which is an important score equation to solve when $\widehat{\theta}^n=-\infty$. Otherwise, Steps \hyperlink{it:tmle1prime}{\ref*{it:tmle1}'}, \hyperlink{it:tmle1prime}{\ref*{it:tmle2}'}, and \hyperlink{it:tmle1prime}{\ref*{it:tmle3}'} in Appendix \ref{app:altalg} ensure that the score equation is solved at both $\beta\equiv 0$ and $\beta\equiv 1$. If $\widehat{\theta}=+\infty$, then $\Ind_{\{\VEm^n(w_\star[i])<\widehat{\theta}^n\}}=1$ for all $i=1,\ldots,n_\star$ so that the preceding argument implies that $\sum_{s=1}^S Q_s^n \nabla \Omega_s^{\widehat{\beta}^n,n} = 0$ and Step \ref{it:omegan} trivially yields that $Q_{\star}^n \nabla \Omega_{\star}^{\widehat{\beta}^n,n}=0$. Similarly, $\widehat{\theta}^n=-\infty$ implies that $\sum_{s=1}^S Q_s^n \nabla \Omega_s^{\widehat{\beta}^n,n} = 0$. 

We now discuss the choice of $\widehat{\beta}^n$. For each $\beta$, $\widehat{\omega}^{\beta,n}$ is a one-step estimator for $\omega^{\beta,0}$. It is therefore not surprising that, under the conditions given in Appendix~\ref{app:est}, $\widehat{\omega}^{\beta,n} = \omega^{\beta,0} + O_P(\nmin^{-1/2})$. If $\theta^0$ is finite, then one can show $\widehat{\theta}^n\approx \theta^0$ (see Lemma~\ref{lem:thetancons} in Appendix~\ref{app:moreinterpretablebetangood}). Furthermore, if $\mu$ is too large so that $\theta^0=+\infty$, then our conditions imply that $\widehat{\theta}^n=\infty$ with probability approaching one (see Lemma~\ref{lem:QstarnWeq1} in the Appendix~\ref{app:thmalproof}), and an analogous result holds if $\mu$ is too small.

Given an estimate of $\widehat{\beta}^n$, it remains to estimate $\phi^0$. We show in Appendix~\ref{app:thmalproof} that
\begin{align*}
\widehat{\phi}^n-\phi^0&\approx \frac{\widehat{\gamma}^{\widehat{\beta}^n,n}-\gamma^{\widehat{\beta}^n,0}}{\omega^{\beta^0,0}} - \frac{\gamma^{\beta^0,0}}{[\omega^{\beta^0,0}]^2}\left[\widehat{\omega}^{\widehat{\beta}^n,n} - \omega^{\beta^0,0}\right] + \frac{\gamma^{\widehat{\beta}^n,0} - \gamma^{\beta^0,0}}{\omega^{\beta^0,0}}.
\end{align*}
For the first term on the right, we note that, $\widehat{\gamma}^{\beta,n}$ is a one-step estimator for $\gamma^{\beta,0}$ for each $\beta$, and so it is not surprising that this estimator is asymptotically linear for each $\beta$. Furthermore, the fact that $\widehat{\beta}^n$ converges to the fixed quantity $\beta^0$ suggests that these same asymptotic linearity statements should hold with $\beta$ replaced by $\widehat{\beta}^n$ (can be formally shown via empirical process arguments). Under \ref{it:omegacontnotflat} and \ref{it:goodquantile}, $\widehat{\omega}^{\widehat{\beta}^n,n}\approx \mu = \omega^{\beta^0,0}$ so that the second term above is negligible. Under \ref{it:mubig} or \ref{it:musmall}, the $\widehat{\beta}^n$ is correctly specified with probability approaching one so that analyzing the second term is like analyzing the one-step estimator $\widehat{\omega}^{\beta^0,n}$ of $\omega^{\beta^0,0}$, which is asymptotically linear under reasonable conditions. Additionally, \ref{it:mubig} or \ref{it:musmall} implies that the final term above is equal to zero with probability approaching one. For the final term above under \ref{it:omegacontnotflat}, we note that
\begin{align*}
\gamma^{\widehat{\beta}^n,0} - \gamma^{\beta^0,0}&= \beta^0\left[\omega^{\widehat{\beta}^n,0} - \omega^{\beta^0,0}\right] + \left\{\gamma^{\widehat{\beta}^n,0} - \gamma^{\beta^0,0} - \beta^0\left[\omega^{\widehat{\beta}^n,0} - \omega^{\beta^0,0}\right]\right\},
\end{align*}
where the latter term is negligible by assumption \ref{it:betangood} (though a more interpretable sufficient condition is given in Appendix~\ref{app:moreinterpretablebetangood}). We establish the behavior of $\omega^{\widehat{\beta}^n,0} - \omega^{\beta^0,0}$ by using that $\widehat{\omega}^{\widehat{\beta},n}\approx \mu = \omega^{\beta^0,0}$, and so it suffices to study the behavior of $-[\widehat{\omega}^{\widehat{\beta}^n,0}-\omega^{\widehat{\beta}^n,0}]$. The study of this term is standard since $\widehat{\omega}^{\beta,0}$ is a one-step estimator of the $\beta$-specific parameter $\omega^{\beta,0}$.

\subsection{Conditional vaccine efficacy constant across $s=1,\ldots,S$} \label{sec:veconst}
The lower bound in the above section relies on the validity of the bridging assumption \ref{it:brVEbridge}, which is indexed by $d$ and $\mathbbmss{v}_0,\mathbbmss{v}_1,\ldots,\mathbbmss{v}_S$. Suppose that the following condition holds:
\begin{enumerate}[resume*=en:bridgconds]
	\item\label{it:VEconst} The vaccine efficacy curve $w\mapsto \VE_s(w)$ is constant across trials $s=1,\ldots,S$ and the chosen $\mathbbmss{v}_0\equiv 0$.
\end{enumerate}
Choosing $\mathbbmss{v}_0\equiv 0$ indicates that the user believes that the vaccine efficacy curves in the completed trials lower bound the vaccine efficacy curve in population $\star$.

We now show that the above condition allows one to choose $\mathbbmss{v}\triangleq (\mathbbmss{v}_s : s=1,\ldots,S)$ to maximize statistical efficiency for estimating our lower bound. In what follows we make the dependence of $w\mapsto \VEm^0(w)$ on $\mathbbmss{v}$ explicit by writing $\VE_{\minus}^{\mathbbmss{v},0}$. The result relies on the fact that $\VE_{\minus}^{\mathbbmss{v},0}$ is invariant in $\mathbbmss{v}$. In particular, for each $w$ we have that
\begin{align*}
\VE_{\minus}^{\mathbbmss{v},0}(w)
&= \frac{\sum_{s=1}^S \mathbbmss{v}_s(w) \E_s[Y|A=0,w]\VE_s^0(w)}{\sum_{s=1}^S \mathbbmss{v}_s(w) \E_s[Y|A=0,w]} \\
&= \VE_1^0(w)\frac{\sum_{s=1}^S \mathbbmss{v}_s(w) \E_s[Y|A=0,w]}{\sum_{s=1}^S \mathbbmss{v}_s(w) \E_s[Y|A=0,w]},
\end{align*}
which is equal to $\VE_1^0(w)$. Thus, under \ref{it:VEconst}, if \ref{it:brVEbridge} holds for some set of functions $\mathbbmss{v}$, then it holds for all $\mathbbmss{v}$. While the above derivation implies that the lower bound parameter $\Phi$ is also invariant to $\mathbbmss{v}$, the same is not true of the gradients presented in Theorem~\ref{thm:pd}. As the gradient determines the efficiency of our estimator, and therefore the width of our confidence interval, it follows that one can choose $\mathbbmss{v}$ to (approximately) maximize the efficiency of our procedure provided \ref{it:VEconst} holds. Let $\hat{\sigma}_n^2(v)$ and $\sigma_n^2(v)$ respectively denote the values of $\hat{\sigma}_n^2$ and $\sigma_n^2$ defined in Section~\ref{sec:conflb}, but now making the dependence on $\mathbbmss{v}$ explicit in the notation. If the consistency conditions of Theorem~\ref{thm:al} hold uniformly over all $\mathbbmss{v}$ in some class $\mathcal{V}$, then one expects that
\begin{align}
\sup_{v\in\mathcal{V}}|\hat{\sigma}_n^2(v) - \sigma_n^2(v)|\rightarrow 0. \label{eq:GC}
\end{align}
As $\hat{\sigma}_n^2(v)$ can be written as an empirical mean of random functions, the primary condition that one needs to add to Theorem~\ref{thm:al} for this convergence to be valid is that $\mathcal{V}$ is not too large, namely that functions in the sets in $\mathcal{V}$ belong to a Glivenko-Cantelli class. We will return to this requirement shortly. For now, suppose that we have selected a class $\mathcal{V}$ small enough so that (\ref{eq:GC}) is plausible. One can then select $\mathbbmss{v}_n$ as the minimizer of $\hat{\sigma}_n^2(v)$ over $v\in\mathcal{V}$. If $\mathcal{V}$ satisfies a more restrictive condition, namely that the functions in the sets it contains are Donsker, then we expect that we can run the estimator described in Section~\ref{sec:eststeps} at the selected $\mathbbmss{v}_n$ and report the lower confidence bound as defined in Section~\ref{sec:conflb}. A simple suggestion for a class $\mathcal{V}$ that satisfies the conditions of the previous theorem is that $\mathcal{V}$ consists of constant functions, namely each $\mathbbmss{v}_s(w)\in[0,1]$ does not depend on $w$. One could alternatively parameterize $\mathbbmss{v}_s$ using, e.g., a linear logistic regression formulation so that $\mathbbmss{v}_s(w)\propto \logit^{-1}(\beta_0 + \beta_1 w)$ to further improve efficiency.  

The arguments that we have provided here are only a sketch: more work would be needed to make them precise, though these arguments are fairly standard when studying M-estimators \citep{vanderVaartWellner1996}, and a detailed study specific to an estimator with a nuisance parameter selected to minimize the variance were given in \cite{Rubin&vanderLaan08}. A careful analysis would show that, when \ref{it:VEconst} holds and $\mathcal{V}$ is not too large, the conditions needed to ensure the validity of the procedure discussed in this section are not much stronger than the conditions needed to establish the validity of a procedure at fixed $\mathbbmss{v}$. Indeed, in Section~\ref{sec:sim} we show via simulation that this procedure yields less conservative lower bounds while still maintaining nominal coverage.

\section{Simulation} \label{sec:sim}
\subsection{Simulation settings}
We evaluated the performance of our method via simulation in \texttt{R} \citep{R2014}. We have $S=2$ completed efficacy trials in our simulation. We first run our simulation without missingness in the biomarker $W$, and then we simulate data from a two-phase sampling scheme.

Let $Z$ be a standard normal random variable. The marginal distribution of $W$ has the same distribution as $\logit^{-1}(Z-2)$ when $s=1$, $\logit^{-1}(5[Z-1])$ when $s=2$, and $\logit^{-1}(2[Z-1/2])$ when $s=\star$. All trials assign treatment with probability $1/2$, regardless of the value of baseline covariates. The vaccine efficacy is the same across Populations 1 and 2. In particular, both have vaccine efficacy
\begin{align*}
\VE(w)&= 1-\frac{\logit^{-1}\left(-1-w-3[0.3 + (w-0.2)^+]\right)}{\logit^{-1}\left(-1-w\right)},
\end{align*}
where $x^+$ denotes the positive part of $x$. As $\E_s^0[Y|A=1,w]$ equals $[1-\VE(w)]\E_s^0[Y|A=0,w]$, it suffices to define $\E_s^0[Y|A=0,w]$ for each of the three trials. We let $\E_1^0[Y|A=0,w] = \logit^{-1}\left(-1-w\right)$ and $\E_2^0[Y|A=0,w] = \logit^{-1}\left(-1\right)\approx 0.27$.

\begin{sloppypar}
We consider two sample size settings, namely $(n_\star,n_1,n_2)$ equal to $(100,2000,2000)$ and $(200,4000,4000)$. We respectively refer to these sample sizes as the ``Smaller Sample Size'' and ``Larger Sample Size'' in our figures. We use three choices for the coefficients indexing our unvaccinated risk lower and upper bounds. The first setting, labeled ``Loosest'' in our figures, sets $\ell_s\equiv\mathbbmss{u}_s\equiv 0$, $s=1,2$, and $\ell_0\equiv 0$, $\mathbbmss{u}_0\equiv 1$. The second, labeled ``Moderate'' in our figures, sets $\ell_s\equiv 0.25$ and $\mathbbmss{u}_s\equiv 0.75$, $s=1,2$, and $\ell_0\equiv \mathbbmss{u}_0\equiv 0$. The third, labeled ``Tight'' in our figures, sets $\ell_s(w)=0.4$ and $\mathbbmss{u}_s(w)=0.6$, $s=1,2$, and $\ell_0\equiv \mathbbmss{u}_0\equiv 0$.
\end{sloppypar}

We suppose that the user \textit{a priori} believes that $\VE_\star^{0,F}(w)\ge \VE_1(w)$ and that $\VE_s(w)$ is invariant in $s$. In this case $\mathbbmss{v}_0(w)=0$, while $w\mapsto \VEm^0(w)$ is invariant in the choice of $\mathbbmss{v}$. We implement the version of our estimator presented in the main text since $\ell_s(w)$ is a multiple of $\mathbbmss{u}_s(w)$ for all simulation settings, both for fixed $\mathbbmss{v}_1(w)=\mathbbmss{v}_2(w)=1/2$, all $w$, indexing $w\mapsto \VEm^0(w)$ and also for the procedure described in Section~\ref{sec:veconst} that adaptively selects a $w$-invariant convex combination $\mathbbmss{v}_1(w),\mathbbmss{v}_2(2)$ to maximize the estimator's precision. We respectively refer to these estimation schemes as ``Fixed'' and ``Adaptive'' in our figures. The marginal distributions of the biomarker within each population are estimated using a ratio of kernel density estimates as provided by the \texttt{ks} package \citep{kspackage}, with the option \texttt{unit.interval} set to \texttt{TRUE} so that the estimator respects the bounds on the biomarker $W$ and the bandwidth selected according to the univariate plug-in selector presented in \cite{Wand&Jones1994}. The ratio is then standardized so that the empirical mean of each $s$-specific estimate of $\frac{dP_\star^0}{dP_s^0}$ has empirical mean $1$ within the sample of observations from Trial $s$. Although the probability of treatment given covariates is known in each trial, we estimate these quantities because of the known efficiency gains resulting from doing so \citep{vdL02}. In particular, we regress $A$ against $W$ using the ensemble algorithm found in the \texttt{SuperLearner} package \citep{vanderLaan&Polley&Hubbard07,SuperLearner2013}, and provide the algorithm with a candidate library containing $\texttt{SL.glm}$ and $\texttt{SL.glm.interaction}$. We estimate the outcome regressions by regressing $Y$ against $A$ and $W$ using the \texttt{SuperLearner} package, using a candidate library containing $\texttt{SL.mean}$, $\texttt{SL.glm}$, $\texttt{SL.glm.interaction}$, $\texttt{SL.step.interaction}$, $\texttt{SL.gam}$, and $\texttt{SL.nnet}$.

We then repeat our simulation with missingness in the biomarker in the completed efficacy trials via a 1:1 nested case:control sampling design. In particular, we suppose that biomarker values are observed on all $m$ cases (participants with $Y=1$), and, within each trial $s$, on a random sample of controls of size $m$. We run the method described in Appendix~\ref{app:extensions}, first using the known observation weights, and then estimating these weights by using the known unity weights among cases and, among controls, running a logistic regression of observation status against a linear interaction model of observation status against an indicator for belonging to Trial $s=1,2$ and an observed covariate $X=\logit^{-1}[\logit(W) + Z]$, where $Z$ is a standard normal. The covariate $X$ is meant to represent a biomarker that is inexpensive to measure on all participants, but that is highly predictive of the biomarker $W$ needed for the bridging exercise.

The function $\VEm^0$ used in our simulation is monotonic, and therefore the results in Appendix~\ref{app:mono} may be used to weaken the assumptions for our method's applicability. The results were nearly identical for this data generating distribution and so are omitted.

\subsection{Simulation results}
Figure~\ref{fig:covg} demonstrates that our lower confidence bounds achieve approximately 95\% or better coverage for the lower bound on the vaccine efficacy across simulation settings and values of $\mu$. The bound is generally conservative for the Moderate and Tight settings. For a given data generating distribution, there is one true value of the marginal unvaccinated risk $\mu_0$, and the interpretation of our curve as a 95\% lower confidence bound is valid if the lower bound coverage is approximately 95\% when $\mu=\mu_0$. 
The methods that choose the convex combination in the $w\mapsto \VEm^0(w)$ bridging parameter to minimize the standard error had comparable coverage than the methods that used a fixed convex combination.

\begin{figure}
	\centering
	\includegraphics[width=\linewidth]{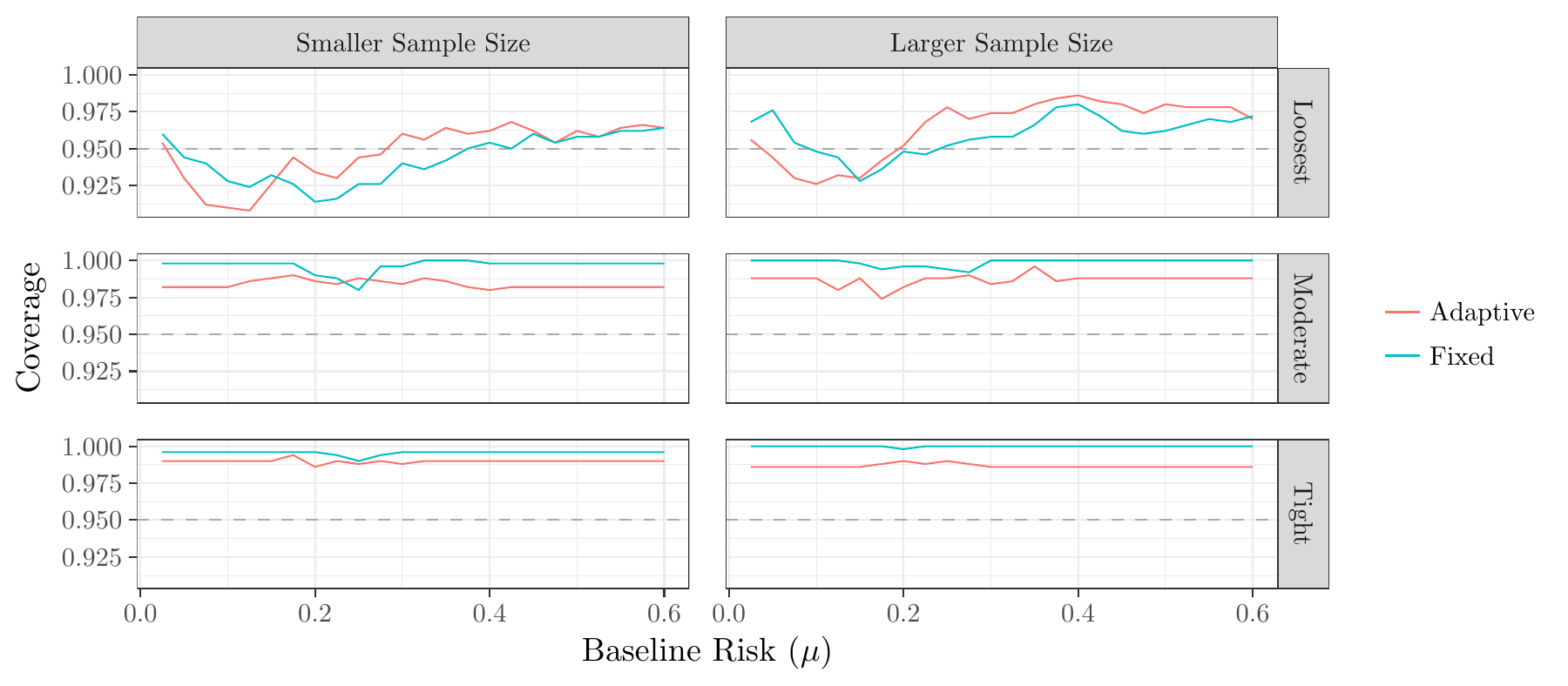}
	\caption{Coverage of our lower confidence bound for $\phi^0$, i.e. the lower bound on the vaccine efficacy, in our simulation. Conducted at both smaller and larger sample sizes, respectively with $(n_\star,n_1,n_2)$ equal to $(100,2000,2000)$ and $(200,4000,4000)$, and for different choices of $\ell_s$ and $\mathbbmss{u}_s$, determining the tightness of the unvaccinated risk bounds. Horizontal dashed lines drawn at 95\% coverage.}
	\label{fig:covg}
\end{figure}

Figure~\ref{fig:avglb} demonstrates that the adaptive methods tend to have a slightly higher average lower bound than the fixed methods that use a convex combination of $(1/2,1/2)$ for the $\VE$ bridging assumption. Given the approximately appropriate coverage of both methods for the (typically conservative) lower bound on the true vaccine efficacy, it seems prudent in practice to use the less conservative adaptive approach.

\begin{figure}
	\centering
	\includegraphics[width=\linewidth]{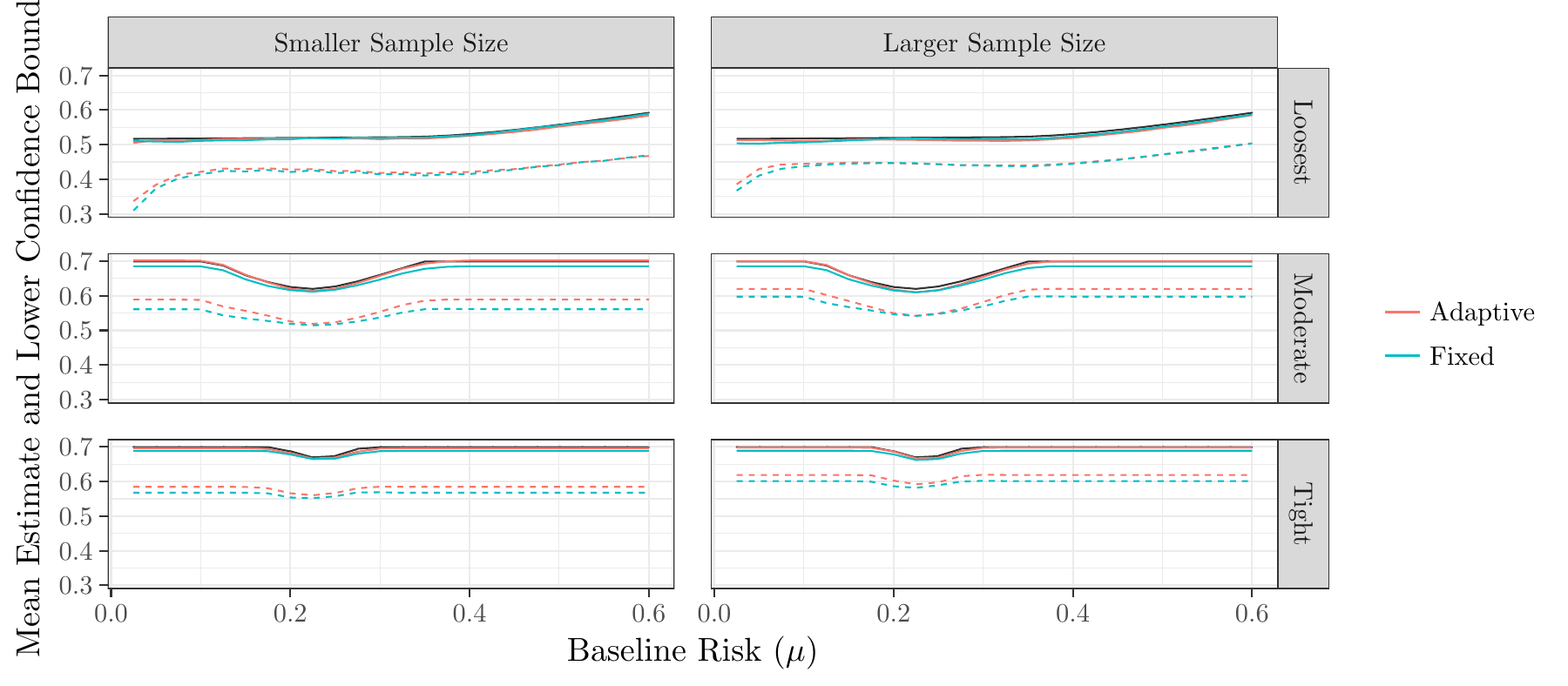}
	\caption{Average estimates (solid lines) and lower confidence bounds (dashed lines) for $\phi^0$, i.e. the lower bound on the vaccine efficacy, in our simulation. Conducted at both smaller and larger sample sizes, respectively with $(n_\star,n_1,n_2)$ equal to $(100,2000,2000)$ and $(200,4000,4000)$, and for different choices of $\ell_s$ and $\mathbbmss{u}_s$, determining the tightness of the unvaccinated risk bounds. Black trend lines denote true $(\ell_s,\mathbbmss{u}_s,\mu)$-specific lower bound.}
	\label{fig:avglb}
\end{figure}

In Appendix~\ref{app:nestedccsim}, we display analogous plots for the fixed two-phase sampling method. Results were similar for the adaptive method. These plots show that the two-phase sampling procedure still maintains the desired coverage level, with a moderate loss in precision in our setting. This is not surprising given that the user has less data for the two-phase methods.

\section{Discussion} \label{sec:disc}
We have presented a method for using data from completed efficacy trials to partially bridge the vaccine efficacy to a new population. We first developed conditions that identify a population level parameter representing a lower bound on the vaccine efficacy in the new population. We then provided a nonparametric estimator of this quantity that respects the fact that this is a multiple sample problem and enables the use of modern data adaptive regression and density estimation techniques to estimate the underlying baseline biomarker distribution, worst-case unvaccinated risk, and vaccine efficacy lower bound.

The validity of our population level lower bound on the vaccine efficacy in the new population relies on three main conditions. One of these conditions essentially states that the support of the key biomarker in the new setting is contained in that of the old efficacy trials so that it is possible to learn about the new setting from these trials. Another states that the results from the earlier trials can be used to derive a lower bound of the conditional vaccine efficacy curve in the new setting. A special case of this condition is that the conditional vaccine efficacy curve is constant across trials. While these two conditions alone allow one to get a lower bound on the marginal vaccine efficacy, this lower bound can be very loose when the biomarker is truly a modifier of vaccine efficacy because it involves finding the worst case distribution for the conditional unvaccinated risk among the class of all such possible conditional risk distributions. We thus add a third condition, namely that the marginal unvaccinated risk is equal to a user-specified constant $\mu$. Adding this constraint can greatly tighten the lower bound. The user can then report lower confidence bounds resulting from our procedure at a range of plausible values of $\mu$.

In this work, we have assumed that one is able to define a lower bound on the vaccine efficacy conditional only on baseline covariates and not on (counterfactual) post-vaccination immune responses. There are advantages and disadvantages to this problem setup. The advantage is that, if the lower bound assumption is valid, investigators can partially bridge the vaccine efficacy to a new setting without vaccinating any individuals in the new population and subsequently waiting to measure the immune response. Avoiding running a phase I immunogenicity trial can save significant financial resources and can also accelerate the introduction of a potentially life-saving vaccine into the new population. The disadvantage of bridging using only using baseline covariates is that the post-vaccination immune response is often more strongly associated with vaccine efficacy, and so a partial bridging assumption that accounts for these responses may be more plausible or less conservative. \cite{Gilbert&Huang2016} consider bridging within a principal stratification framework, thereby allowing for conditioning on counterfactual immune responses. In a future work, we will present a doubly robust partial bridging procedure that allows for conditioning on post-vaccination biomarkers.

Our vaccine efficacy lower bound depends on the marginal unvaccinated risk level $\mu$. While in many settings subject-matter experts may be able to suggest a plausible range of values for $\mu$, in other cases there may not be sufficient knowledge of the marginal risk for this to be possible. To deal with this issue, one could follow the suggestion of \cite{Gilbert&Huang2016} to formally incorporate epidemiologic surveillance data to learn such a range. While there may not be biomarker and baseline risk information jointly available, one may at times have reasonably accurate active surveillance data available to estimate the marginal unvaccinated risk in the new population. In this case, one could estimate our $\mu$ parameter rather than report a range of plausible values. One could then either treat this estimate as the truth, or modify our inferential procedure to incorporate the added uncertainty. While active surveillance data should be incorporated when it exists, often only data from national surveillance systems is available. These data may not be sufficiently accurate to inform plausible values of $\mu$. For example, recent studies of national surveillance of dengue incidence have shown up to 19-fold underreporting compared to active surveillance \citep{Nealonetal2016,Sartietal2016}.

The results in this paper readily extend to a large class of contrasts $f(\mu_1,\mu_0)$ between the marginal vaccinated risk $\mu_1$ and unvaccinated risk $\mu_0$. We have studied the multiplicative efficacy $f(\mu_1,\mu_0)=1-\mu_1/\mu_0$. If one is willing to make the three bridging assumptions made in this paper, namely a support condition, lower bounded conditional \textit{multiplicative} vaccine efficacy, and specified unvaccinated risk, then obtaining a lower bound for more general $f(\mu_1,\mu_0)$ is straightforward. Suppose that, for each $\mu_0$, $f(\mu_1,\mu_0)$ is sufficiently smooth and
monotonically decreasing in $\mu_1$ (the higher the risk among vaccinated individuals, the lower the efficacy). For a user-defined $\mu$, one can then conservatively estimate $f(\mu_1,\mu)$ by upper bounding $\mu_1$, which is indeed what we do in Section~\ref{sec:setup}. When $\mu=\mu_0$, one then gets a lower bound for the efficacy parameter of interest. All typical contrasts, including additive and odds-ratio type contrasts, appear to satisfy this monotonicity condition.

Finally, we note that, while we have focused our discussion on the intervention $A$ being a vaccination indicator, the results given in this paper apply immediately to any binary intervention $A$ and multiplicative efficacy parameter.

\section*{Acknowledgements}
This work was partially supported by the National Institute of Allergy and Infectious Disease at the National Institutes of Health under award  numbers R37AI054165 and UM1 AI068635.  The content is solely the responsibility of the authors and does not necessarily represent the official views of the National Institutes of Health. 
The authors thank Ying Huang, Kara Rudolph, Ying Chen, Holly Janes, Dobromir Dimitrov, Jonathan Sugimoto, Laura Matrajt and Paul Edlefsen for their thoughtful comments on the early stages of this project.

\bibliographystyle{Chicago}
 \bibliography{persrule}

\newpage
\appendix
\setcounter{equation}{0}
\renewcommand{\theequation}{A.\arabic{equation}}
\setcounter{theorem}{0}
\renewcommand{\thetheorem}{A.\arabic{theorem}}
\renewcommand{\thecorollary}{A.\arabic{theorem}}
\renewcommand{\thelemma}{A.\arabic{theorem}}
\renewcommand{\theproposition}{A.\arabic{theorem}}
\renewcommand{\theconjecture}{A.\arabic{theorem}}

\section*{Appendix}
\section{Additional regularity conditions for Theorem~\ref{thm:al}} \label{app:est}
For general $\mathcal{P}'\in\mathcal{M}$, below we write $\VEm'$ and $D_{\VE,s}'$ for the analogues of $w\mapsto \VEm^0(w)$ and $D_{\VE,s}^0$ under $\mathcal{P}'$ rather than $\mathcal{P}^0$. Define
\begin{align*}
&\Rem_1(\mathcal{P}',\mathcal{P}^0)(w) \\
&\;\triangleq \sum_{s=1}^S \left|\E_s'[Y|A=0,w]-\E_s^0[Y|A=0,w]+\E_s^0\left[\frac{\Ind_{\{A=0\}}}{P_s'(A|w)}\left(Y - \E_s^0[Y|A,w]\right)\middle|w\right]\right| \\
&\;= \sum_{s=1}^S \left|\left(1-\frac{P_s^0(A=0|w)}{P_s'(A=0|w)}\right)\left(\E_s'[Y|A=0,w] - \E_s^0[Y|A=0,w]\right)\right|, \\
&\Rem_2(\mathcal{P}',\mathcal{P}^0)(w) \\
&\;\triangleq c(w) - \VEm^0(w) + \sum_{s=1}^S \E_s^0\left[D_{\VE,s}'(O)\middle|w\right] \\
&\;= \left[\VEm'(w)-\VEm^0(w)\right]\left[1-\frac{\mathbbmss{v}_0(w)d(0,w) + \sum_{s=1}^S \mathbbmss{v}_s(w)\E_s^0\left[Y|A=0,w\right]}{\mathbbmss{v}_0(w)d(0,w) + \sum_{s=1}^S \mathbbmss{v}_s(w)\E_s'\left[Y|A=0,w\right]}\right] \\
&\hspace{2em}- \frac{\VEm'(w)\sum_s \mathbbmss{v}_s(w) \left[1-\frac{P_s^0(A=0|w)}{P_s'(A=0|w)}\right]\left(\E_s'[Y|A=0,w] - \E_s^0[Y|A=0,w]\right)}{\mathbbmss{v}_0(w)d(0,w) + \sum_{s=1}^S \mathbbmss{v}_s(w)\E_s^0\left[Y|A=0,w\right]}.
\end{align*}
We now state the additional regularity conditions needed for Theorem~\ref{thm:al}.
\begin{enumerate}[resume*=en:regconds]
	\item\label{it:CSconsistency} $\mathcal{P}^n$ is consistent in the sense that all of the following terms are $o_P(\nmin^{-1/2})$:
	\begin{align*}
	&\norm{\Rem_1(\mathcal{P}^n,\mathcal{P}^0)}_{2,P_\star^0}, \\
	&\norm{\Rem_2(\mathcal{P}^n,\mathcal{P}^0)}_{2,P_\star^0}.
	\end{align*}
	Furthermore, for each $s=1,\ldots,S$, all of the following terms are also $o_P(\nmin^{-1/2})$:
	\begin{align*}
	&\sup_{\beta : \mathcal{W}\rightarrow[0,1]}\norm{w\mapsto \E_s^0[D_{\UR,s}^{\beta,n}|w]-\E_s^0[D_{\UR,s}^{\beta,0}|w]}_{2,P_s^0} \norm{\frac{dP_{\star}^n}{dP_s^n}-\frac{dP_{\star}^0}{dP_s^0}}_{2,P_s^0}, \\
	&\norm{\frac{dP_{\star}^n}{dP_s^n}-\frac{dP_{\star}^0}{dP_s^0}}_{2,P_s^0} \sup_{\beta : \mathcal{W}\rightarrow[0,1]} \norm{w\mapsto \E_s^0[D_{\UR,s}^{\beta,n}(O)|w]-\E_s^0[D_{\UR,s}^{\beta,0}(O)|w]}_{2,P_s^0}, \\
	&\norm{\frac{dP_{\star}^n}{dP_s^n}-\frac{dP_{\star}^0}{dP_s^0}}_{2,P_s^0} \norm{w\mapsto \E_s^0[D_{\VE,s}^n(O)|w]-\E_s^0[D_{\VE,s}^0(O)|w]}_{2,P_s^0}, \\
	&\norm{\upsilon^n-\upsilon^0}_{2,P_s^0} \norm{\VE^n-\VEm^0}_{2,P_s^0}.
	\end{align*}
\end{enumerate}
Like Donsker conditions used in one-sample problems, \ref{it:thetaDonsker} requires certain centered and scaled empirical means to be tight random elements. The next assumption makes this requirement formal.
\begin{enumerate}[resume*=en:regconds]
	\item\label{it:thetaDonsker} $(Q_s^n-P_s^0)\nabla \Omega_s^{\beta^0,n}=O_P(\nmin^{-1/2})$, $(Q_s^n-P_s^0)\nabla \Omega_s^{\widehat{\beta}^n,n}=O_P(\nmin^{-1/2})$ for each $s\in\mathcal{S}$.
\end{enumerate}
In our multiple sample problem, one can analyze $(Q_s^n-P_s^0) \nabla \Omega_s^{\beta,n}$ for $\beta$ fixed ($\beta=\beta^0$) or random ($\beta=\widehat{\beta}^n$) for each $s$ by conditioning on all observations not belonging to trial $s$, and then applying a maximal inequality \citep[which relies on the bracketing or uniform entropy integral of the class to which $\nabla \Omega_s^{\beta,n}$ belongs, see][]{vanderVaartWellner1996} to bound the randomness in this term resulting from observations in trial $s$.

We also require a condition that is similar to asymptotic equicontinuity conditions implied by the use of Donsker classes in one-sample problems.
\begin{enumerate}[resume*=en:regconds]
	\item\label{it:thetaempproc} $(Q_s^n-P_s^0)[\nabla \Omega_s^{\widehat{\beta}^n,n}-\nabla \Omega_s^{\beta^0,0}]$, $(Q_s^n-P_s^0)[\nabla \Gamma_s^{\widehat{\beta}^n,n}-\nabla \Gamma_s^{\beta^0,0}]$ for each $s\in\mathcal{S}$.
\end{enumerate}
Using the same arguments outlined following \ref{it:thetaDonsker}, one can show that, if $\nabla \Omega_s^{\widehat{\beta}^n,n}$ and $\nabla \Gamma_s^{\widehat{\beta}^n,n}$ respectively converge to $\nabla \Omega_s^{\beta^0,0}$ and $\nabla \Gamma_s^{\beta^0,0}$ in $P_s^0$ mean-square, we will allow be able to establish \ref{it:thetaempproc}. While we do not formally give conditions under which this mean-square convergence occurs, the primary assumption needed in addition to the consistency conditions used in \ref{it:CSconsistency} is that $\widehat{\theta}^n$ is consistent for $\theta^0$. We give sufficient conditions for this convergence in Lemma~\ref{lem:thetancons}.

The next assumption is made for simplicity, though it is very mild, especially under \ref{it:brvebdd}, which states that $w\mapsto \VEm^0(w)$ is bounded away from $-\infty$.
\begin{enumerate}[resume*=en:regconds]
	\item\label{it:VEnbdd} $\inf_w \VE^n(w)>-c>-\infty$ with probability approaching one, where the constant $c<\infty$ does not depend on sample size.
\end{enumerate}

We make two additional assumptions when \ref{it:omegacontnotflat} holds. We first assume that the empirical version of $\omega^{\beta^0,0}$ is also close to $\mu$, which seems reasonable given that $\omega^{\beta^0,0}=\mu$ under \ref{it:omegacontnotflat}.
\begin{enumerate}[resume*=en:regconds]
	\item\label{it:goodquantile} $\widehat{\omega}^{\widehat{\beta}^n,n} - \mu=o_P(\nmin^{-1/2})$.
\end{enumerate}
This condition can formally be supported in practice by looking at the size of $\widehat{\omega}^{\widehat{\beta}^n,n} - \mu$. The next regularity condition ensures that $\widehat{\beta}^n$ is a sufficiently good estimate of $\beta^0$ so that $\UR^{\widehat{\beta}^n,0}$ can be interpreted as being approximately equal to the the worst-case conditional unvaccinated risk.
\begin{enumerate}[resume*=en:regconds]
	\item\label{it:betangood} $\Rem_3^n\triangleq \gamma^{\widehat{\beta}^n,0} - \gamma^{\beta^0,0} - \theta^0[\omega^{\widehat{\beta}^n,0} - \omega^{\beta^0,0}]=o_P(\nmin^{-1/2})$.
\end{enumerate}
We give more interpretable sufficient conditions for \ref{it:betangood} in Appendix~\ref{app:moreinterpretablebetangood}. These conditions are a variant of the margin assumptions used in the classification literature.

\section{Proofs} \label{app:proofs}
\subsection{Proof of lower bound} \label{app:lb}
\begin{proof}[Proof of Lemma~\ref{lem:lb}]\hypertarget{proof:lemlb}{}
Because $\mu=\E_\star^0[\E_\star^{0,F}[Y|A=0,W]]$ and \ref{it:brdatadepub} holds, $\omega^{\beta^0,0}=\mu$. By \ref{it:brVEbridge}, $\omega^{\beta^0,0}\Psi(P_\star^{0,F})\ge \E_\star^0[\E_\star^{0,F}[Y|A=0,W]\VEm^0(W)]$. Below we respectively use $\mathcal{E}_1$, $\mathcal{E}_2$, and $\mathcal{E}_3$ to denote the events $\{\VEm^0(W)<\theta^0\}$, $\{\VEm^0(W)=\theta^0\}$, and $\{\VEm^0(W)>\theta^0\}$, we see that
\begin{align*}
\omega^{\beta^0,0}&\Psi(P_\star^{0,F}) - \gamma^{\beta^0,0}\ge \E_\star^0[\E_\star^{0,F}[Y|A=0,W]\VEm^0(W)] - \gamma^{\beta^0,0} \\
=&\, \E_\star^0\left[(\Ind_{\mathcal{E}_1} + \eta^0 \Ind_{\mathcal{E}_2})\left\{\E_\star^{0,F}[Y|A=0,W]-\upsilon^0(W)\right\}\VEm^0(W)\right] \\
&+ \E_\star^0\left[(\Ind_{\mathcal{E}_3} + [1-\eta^0] \Ind_{\mathcal{E}_2})\left\{\E_\star^{0,F}[Y|A=0,W]-\lambda^0(W)\right\}\VEm^0(W)\right].
\end{align*}
By \ref{it:brdatadepub} and the fact that $\eta^0\in [0,1]$, replacing each instance of $\VEm^0(W)$ on the right-hand side by $\theta^0$ yields a lower bound of the form $\theta^0\left[\mu - \omega^{\beta^0,0}\right]$. As $\omega^{\beta^0,0}=\mu$, this lower bound is equal to zero. Dividing both sides by $\omega^{\beta^0,0}>0$ gives the result.
\end{proof}

\begin{lemma} \label{lem:lbmon}
If $\lambda^0\equiv 0$, \ref{it:brcommonsupport}, and \ref{it:brvebdd}, then $\phi_\mu^0$ is monotonically nondecreasing in $\mu$.
\end{lemma} 
\begin{proof}[Proof of Lemma \ref{lem:lbmon}]\hypertarget{proof:lemlbmon}{}
Fix $\theta_1\le \theta_2$ and $x_1,x_2\in[0,1]$. If $\theta_1=\theta_2$, then suppose also that $x_1\le x_2$. For $k=1,2$, define $\beta_k\triangleq w\mapsto \Ind_{\{\VEm^0(w)<\theta_k\}} + x_k \Ind_{\{\VEm^0(w)=\theta_k\}}$. Note that $\beta_2(w)-\beta_1(w)\ge 0$ and is strictly positive only if $\VEm^0(W)\ge \theta_1$. Observe that
\begin{align*}
\frac{\gamma^{\beta_1,0}}{\omega^{\beta_1,0}}-\frac{\gamma^{\beta_2,0}}{\omega^{\beta_2,0}}&= \frac{\left[\omega^{\beta_2,0}-\omega^{\beta_1,0}\right]\gamma^{\beta_1,0} - \omega^{\beta_1,0}\left[\gamma^{\beta_2,0}-\gamma^{\beta_1,0}\right]}{\omega^{\beta_1,0}\omega^{\beta_2,0}} \\
&= \frac{\E_\star^0\left[\{\beta_2(W)-\beta_1(W)\}\upsilon^0(W)\{\gamma^{\beta_1,0}-\VEm^0(W)\omega^{\beta_1,0}\}\right]}{\omega^{\beta_1,0}\omega^{\beta_2,0}} \\
&\le \frac{\E_\star^0\left[\{\beta_2(W)-\beta_1(W)\}\upsilon^0(W)\right]\left[\gamma^{\beta_1,0}-\theta_1\omega^{\beta_1,0}\right]}{\omega^{\beta_1,0}\omega^{\beta_2,0}}.
\end{align*}
Noting that
\begin{align*}
\gamma^{\beta_1,0}&= \E_\star^0[\upsilon^0(W)\{\Ind_{\{\VEm^0(w)<\theta_1\}} + x_1 \Ind_{\{\VEm^0(w)=\theta_1\}}\}\VEm^0(W)] \\
&\le \theta_1 \E_\star^0[\upsilon^0(W)\{\Ind_{\{\VEm^0(w)<\theta_1\}} + x_1 \Ind_{\{\VEm^0(w)=\theta_1\}}\}] = \theta_1\omega^{\beta_1,0},
\end{align*}
we see that $\frac{\gamma^{\beta_1,0}}{\omega^{\beta_1,0}}\le \frac{\gamma^{\beta_2,0}}{\omega^{\beta_2,0}}$.

We now write $\theta_\mu^0$ and $x_\mu^0$ to make the dependence of $\theta^0$, $\eta^0$ on $\mu$ explicit. As $\upsilon^0(w)>0$ for all $w$, clearly $\mu_1<\mu_2$ implies that $\theta_{\mu_1}^0\le \theta_{\mu_2}^0$ and, if $\theta_{\mu_1}^0= \theta_{\mu_2}^0$, then also that $x_{\mu_1}^0\le x_{\mu_2}^0$. This completes the proof.
\end{proof}

\subsection{Pathwise derivative of $\Phi$ (Theorem~\ref{thm:pd})} \label{app:pd}

Throughout this section, we refer to parameters evaluated at the collection of distributions $\mathcal{P}^\epsilon$ rather than parameters evaluated at $\mathcal{P}^0$ by replacing the superscript zero by superscript $\epsilon$, e.g. we write $\upsilon^\epsilon\triangleq \Upsilon(\mathcal{P}^\epsilon)$ rather than $\upsilon^0\triangleq \Upsilon(\mathcal{P}^0)$.

We now prove Theorem~\ref{thm:pd}. The proof references results that we prove later in this section.
\begin{proof}[Proof of Theorem~\ref{thm:pd}]\hypertarget{proof:thmpd}{}
We consider the cases that \ref{it:omegacontnotflat} holds and that \ref{it:mubig} holds separately. The proof in the case where \ref{it:musmall} holds is nearly identical that under \ref{it:mubig}, and so the proof is omitted.\newline
\textit{Case 1:} \ref{it:omegacontnotflat} holds. By \ref{it:omegacontnotflat} and the fact that $\upsilon^0$ is bounded away from zero, $\omega^{\beta^0,0} = \mu$ and $\theta^0$ is finite. By Lemma~\ref{lem:thetacons}, $\theta^\epsilon$ is finite for all $\epsilon$ small enough so that  $\omega^{\beta^\epsilon,\epsilon} = \mu$ for these $\epsilon$. Hence, $\phi^\epsilon - \phi^0 = \frac{\gamma^{\beta^\epsilon,\epsilon}-\gamma^{\beta^0,0}}{\mu}$ for all $\epsilon$ sufficiently small. Dividing both sides by $\epsilon$, taking the limit as $\epsilon\rightarrow 0$, and applying Theorem~\ref{thm:pdLambda} (Section~\ref{sec:pdLambda}) shows that the gradients of $\Phi$ at $\mathcal{P}^0$ are given by the desired expressions.\newline
\textit{Case 2:} \ref{it:mubig} holds. The key to proving Theorem~\ref{thm:pd} under \ref{it:mubig} is proving that $P_\star^\epsilon\{\VEm^\epsilon(W)<\theta^\epsilon\}=P_\star^0\{\VEm^0(W)<\theta^0\}=1$ for all $\epsilon$ small enough. We formally prove this in Lemma~\ref{lem:thetaepsatbdry} (Section~\ref{sec:thetaepsatbdry}). It follows that
\begin{align*}
\Phi(\mathcal{P}^\epsilon) - \Phi(\mathcal{P}^0)&= \frac{\E_\star^\epsilon\left[\upsilon^\epsilon(W)\VE^\epsilon(W)\right]}{\E_\star^\epsilon\left[\upsilon^\epsilon(W)\right]} - \frac{\E_\star^0\left[\upsilon^0(W)\VEm^0(W)\right]}{\E_\star^0\left[\upsilon^0(W)\right]}.
\end{align*}
Dividing both sides by $\epsilon$ and taking the limit as $\epsilon\rightarrow 0$ yields the same expression that would be used to evaluate the gradients of the parameter $\mathcal{P}'\rightarrow \frac{\int \Upsilon(\mathcal{P}')(w)\VEm'(w) dP_\star'(w)}{\int \Upsilon(\mathcal{P}')(w) dP_\star'(w)}$. It is straightforward to derive expressions for the gradients of this parameter via the delta method.
\end{proof}

\subsubsection{Lemma used to Prove Theorem~\ref{thm:pd} under \ref{it:mubig}.} \label{sec:thetaepsatbdry}
\begin{lemma}\label{lem:thetaepsatbdry}
If \ref{it:brcommonsupport}, \ref{it:brvebdd}, and \ref{it:mubig}, then $P_\star^\epsilon\{\VEm^\epsilon(W)<\theta^\epsilon\}=1$ for all $\epsilon$ small enough.
\end{lemma}
\begin{proof}[Proof of Lemma~\ref{lem:thetaepsatbdry}]
Standard pathwise differentiability calculations show that
\begin{align*}
\frac{\E_\star^\epsilon[\upsilon^\epsilon(W)]-\omega^{\beta^0,0}}{\epsilon}&= \frac{\E_\star^\epsilon[\upsilon^\epsilon(W)]-\E_\star^0[\upsilon^0(W)]}{\epsilon} \\
&= \sum_{s\in\mathcal{S}} P_s^\epsilon \nabla \Omega_s^{\beta^0,0} + o(1),
\end{align*}
and the fact that each $\nabla \Omega_s^{\beta^0,0}(O_s)$ has mean zero for $O_s\sim P_s^0$ implies that the right-hand side is $O(1)$. By \ref{it:mubig}, $\omega^{\beta^0,0}<\mu$. Hence, $\E_\star^\epsilon[\upsilon^\epsilon(W)] < \mu + O(\epsilon)$. It follows that, for all $\epsilon$ sufficiently small, $\E_\star^\epsilon[\UR^{w\mapsto\Ind_{\{\VEm^\epsilon(w)<\theta\}},\epsilon}(W)]<\mu$ for all $\theta$. Hence, $\theta^\epsilon=+\infty$ for all $\epsilon$ sufficiently small, completing the proof.
\end{proof}

\subsubsection{Theorem used to prove Theorem~\ref{thm:pd} under \ref{it:omegacontnotflat}.} \label{sec:pdLambda}
The following theorem establishes the pathwise differentiability of the parameter $\mathcal{P}'\mapsto \Gamma^{\beta'}(\mathcal{P}')$, where the $\beta'$ in the subscript is the analogue of $\beta^0$ but defined at parameter input $\mathcal{P}'$ rather than at $\mathcal{P}^0$.
\begin{theorem} \label{thm:pdLambda}
If \ref{it:brcommonsupport}, \ref{it:brvebdd}, and \ref{it:omegacontnotflat}, then $\mathcal{P}'\mapsto \Gamma^{\beta'}(\mathcal{P}')$ is pathwise differentiable at $\mathcal{P}^0$ with $P_s^0$ gradients $o_s\mapsto \nabla \Gamma_s^{\beta^0,0}(o_s) - \theta^0 \nabla \Omega_s^{\beta^0,0}(o_s)$, $s\in\mathcal{S}$.
\end{theorem}
\begin{proof}[Proof of Theorem~\ref{thm:pdLambda}]
Observe that
\begin{align*}
\frac{\gamma^{\beta^\epsilon,\epsilon}-\gamma^{\beta^0,0}}{\epsilon}&= \frac{\gamma^{\beta^0,\epsilon}-\gamma^{\beta^0,0}}{\epsilon}+\frac{\gamma^{\beta^\epsilon,\epsilon}-\gamma^{\beta^0,\epsilon}}{\epsilon}.
\end{align*}
We wish to obtain an expression for the limit on the left as $\epsilon\rightarrow 0$. If the limits exist for both terms on the right, then the limit on the left is given by their sum. A standard delta method argument shows that the limit as $\epsilon\rightarrow 0$ on the right is given by $\sum_{s\in\mathcal{S}} P_s^0 \nabla \Gamma_s^{\beta^0,0} h_s$. Lemma~\ref{lem:quantilederiv} shows that the limit exists for the latter term and provides its closed form expression.
\end{proof}

\begin{lemma} \label{lem:quantilederiv}
Under the conditions of Theorem~\ref{thm:pdLambda},
\begin{align*}
\lim_{\epsilon\rightarrow 0}\frac{\gamma^{\beta^\epsilon,\epsilon}-\gamma^{\beta^0,\epsilon}}{\epsilon}&= -\theta^0\sum_{s\in\mathcal{S}} \int \nabla \Omega_s^{\beta^0,0}(o_s) h_s(o_s) dP_s^0(o_s).
\end{align*}
\end{lemma}
\begin{proof}[Proof of Lemma~\ref{lem:quantilederiv}]
Telescoping yields that
\begin{align*}
\gamma^{\beta^\epsilon,\epsilon}-\gamma^{\beta^0,\epsilon}=&\, \E_{\star}^\epsilon\left[\left\{\UR^{\beta^\epsilon,\epsilon}(W)-\UR^{\beta^0,\epsilon}(W)\right\}\VE^\epsilon(W)\right] \\
=&\, \E_{\star}^\epsilon\left[\left\{\UR^{\beta^\epsilon,\epsilon}(W)-\UR^{\beta^0,\epsilon}(W)\right\}\left\{\VE^\epsilon(W)-\theta^0\right\}\right] \\
&+ \theta^0\int \left\{\UR^{\beta^\epsilon,\epsilon}(w) dP_\star^\epsilon(w)-\UR^{\beta^0,0}(w) dP_\star^0(w)\right\} \\
&- \theta^0\int \left\{\UR^{\beta^0,\epsilon}(w) dP_\star^\epsilon(w)-\UR^{\beta^0,0}(w) dP_\star^0(w)\right\}.
\end{align*}
We denote the first and second terms in the final equality by $T_1^\epsilon$ and $T_2^\epsilon$. If we can show that $T_1^\epsilon$ and $T_2^\epsilon$ are $o(\epsilon)$, then the result is immediate by dividing both sides by $\epsilon$ and applying the chain rule to the third term.

The remainder of this proof uses a positive constant $c$ that may vary line by line. We first establish that $T_1^\epsilon=o(\epsilon)$. 
Straightforward calculations show that
\begin{align*}
|T_1^\epsilon|\le (1+c|\epsilon|)\E_{\star}^0\Bigg[&\left\{\left|\UR^{\beta^\epsilon,0}(W)-\UR^{\beta^0,0}(W)\right| + c|\epsilon|\right\} \\
&\times\left\{\left|\VEm^0(W)-\theta^0\right| + c|\epsilon|\right\}\Bigg].
\end{align*}
For simplicity we suppose that $\eta^\epsilon=0$ for the remainder of the proof, but the proof with $\eta^\epsilon\not=0$ is nearly identical. For each $w$,
\begin{align*}
&\left|\UR^{\beta^\epsilon,0}(w)-\UR^{\beta^0,0}(w)\right| = \left|\Ind_{\{\VEm^\epsilon(w)<\theta^\epsilon\}}-\Ind_{\{\VEm^0(w)<\theta^0\}}\right|[\upsilon^0(w)-\lambda^0(w)].
\end{align*}
Using that $\Ind_{\{\VEm^\epsilon(w)<\theta^\epsilon\}}\not=\Ind_{\{\VEm^0(w)<\theta^0\}}$ implies that $0\le |\VEm^0(w)-\theta^0|<|\VEm^\epsilon(w)-\theta^\epsilon-\VEm^0(w)+\theta^0|$ then shows that
\begin{align*}
|T_1^\epsilon|\le (1+c|\epsilon|)\E_{\star}^0\Bigg|&\Ind_{\{0\le |\VEm^0(w)-\theta^0|<|\VEm^\epsilon(w)-\theta^\epsilon-\VEm^0(w)+\theta^0|\}} \\
&\times\left\{[\upsilon^0(W)-\lambda^0(W)]|\VEm^0(W)-\theta^0| + c|\epsilon|\right\}\Bigg|.
\end{align*}
The indicator above can be replaced by $\Ind_{\{0< |\VEm^0(w)-\theta^0|<|\VEm^\epsilon(w)-\theta^\epsilon-\VEm^0(w)+\theta^0|\}}$ because \ref{it:omegacontnotflat} implies that $P_\star^0(\VEm^0(W)=\theta^0)=0$. Combining this with the fact that $\upsilon^0(W)-\lambda^0(W)\le 1$ and $|\VEm^\epsilon(w)-\VEm^0(w)|\le c|\epsilon|$ shows that
\begin{align*}
|T_1^\epsilon|&\le (1+c|\epsilon|)\E_{\star}^0\left[\Ind_{\{0< |\VEm^0(w)-\theta^0|<|\VEm^\epsilon(w)-\theta^\epsilon-\VEm^0(w)+\theta^0|\}}\left\{|\theta^\epsilon-\theta^0| + c|\epsilon|\right\}\right].
\end{align*}
By Lemma~\ref{lem:thetacons}, $\theta^\epsilon-\theta^0=O(\epsilon)$ so that, for all $\epsilon$ large enough,
\begin{align*}
|T_1^\epsilon|\le c|\epsilon|P_\star^0\left\{0< |\VEm^0(w)-\theta^0|<c|\epsilon|\right\}.
\end{align*}
The probability statement is $o(1)$. Thus, $T_1^\epsilon=o(\epsilon)$.

Lemma~\ref{lem:thetacons} establishes that $\theta^\epsilon$ is finite for all $\epsilon$ small enough, and 

The assumption that \ref{it:omegacontnotflat} establishes that $\theta^0$ is finite, and thus Lemma~\ref{lem:thetacons} establishes that $\theta^\epsilon$ is finite for all $\epsilon$ small enough. It follows that $\omega^{\beta^\epsilon,\epsilon}=\mu$ for all $\epsilon$ small enough, and thus $T_2^\epsilon=\omega^{\beta^\epsilon,\epsilon}-\omega^{\beta^0,0}=0$ for these $\epsilon$.
\end{proof}

\begin{lemma} \label{lem:F0notflat}
For each $t$, let $\tilde{\beta}_t\triangleq w\mapsto \Ind_{\{\VEm^0(w)<t\}}$. If \ref{it:omegacontnotflat} holds, then, for all $\theta$ in a neighborhood of $\theta^0$,
\begin{align*}
0&<\liminf_{t\rightarrow 0}\frac{\omega^{\tilde{\beta}_{\theta+t},0}-\omega^{\tilde{\beta}_{\theta},0}}{t}\le \limsup_{t\rightarrow 0}\frac{\omega^{\tilde{\beta}_{\theta+t},0}-\omega^{\tilde{\beta}_{\theta},0}}{t} <\infty,
\end{align*}
\end{lemma}
\begin{proof}[Proof of Lemma~\ref{lem:F0notflat}]
As $\upsilon^0(w)-\lambda^0(w)\ge \delta$ for all $w$, the facts that $t\mapsto P_\star^0\{\VEm^0(W)<\theta^0+t\}$ is monotonically non-decreasing and that \ref{it:omegacontnotflat} holds yield that
\begin{align*}
\liminf_{t\rightarrow 0}\frac{\omega^{\tilde{\beta}_{\theta+t},0}-\omega^{\tilde{\beta}_{\theta},0}}{t}&= \liminf_{t\rightarrow 0}\frac{\E_\star^0[\{\upsilon^0(W)-\lambda^0(W)\}I\{\theta^0\le \VEm^0(W)< \theta^0+t\}]}{t} \\
&\ge\delta\liminf_{t\rightarrow 0}\frac{P_\star^0\{\VEm^0(W)<\theta^0+t\}-P_\star^0\{\VEm^0(W)<\theta^0\}}{t}>0.
\end{align*}
The limit supremum result is proven similarly.
\end{proof}

\begin{lemma} \label{lem:thetacons}
Under the conditions of Theorem~\ref{thm:pdLambda}, $\theta^\epsilon= \theta^0 + O(\epsilon)$.
\end{lemma}
\begin{proof}[Proof of Lemma~\ref{lem:thetacons}]
In what follows we abuse notation and write $\omega^{\theta,\epsilon}$ for $\omega^{w\mapsto \Ind_{\{\VEm^\epsilon(w)<\theta\}},\epsilon}$ for all $\theta,\epsilon$. Note that, by \ref{it:omegacontnotflat}, $\eta^0=0$ so that $\omega^{\beta^0,0}=\omega^{\theta^0,0}$. We similarly abuse notation by letting $\UR^{\theta,\epsilon}\triangleq \UR^{w\mapsto \Ind_{\{\VEm^\epsilon(w)<\theta\}},\epsilon}$ and $D_{\UR,s}^{\theta,\epsilon}\triangleq D_{\UR,s}^{w\mapsto \Ind_{\{\VEm^\epsilon(w)<\theta\}},\epsilon}$.

We first establish that $\theta^\epsilon\rightarrow\theta^0$. For $t>0$, the definition of $\theta^\epsilon$ yields that $|\theta^\epsilon-\theta^0|>t$ implies that either $\omega^{\theta^0-t,\epsilon}>\mu$ or $\omega^{\theta^0+t,\epsilon}\le \mu$. We will establish that $\theta^\epsilon\rightarrow\theta^0$ by showing that, for any $t>0$, both (i) $\omega^{\theta^0-t,\epsilon}\le\mu$ and (ii) $\omega^{\theta^0+t,\epsilon}> \mu$, and therefore $|\theta^\epsilon-\theta^0|\le t$.

For a constant $c>0$ that may vary line by line,
\begin{align*}
\omega^{\theta^0-t,\epsilon}&\le (1+c|\epsilon|)\E_{\star}^0\left[\UR^{\theta^0-t,\epsilon}(W)\right] \\
&= (1+c|\epsilon|)\E_{\star}^0\left[\lambda^\epsilon(W) + \{\upsilon^\epsilon(W)-\lambda^\epsilon(W)\}\Ind\{\VEm^\epsilon(W)<\theta^0-t\}\right] \\
&\le (1+c|\epsilon|)\E_{\star}^0\left[\lambda^\epsilon(W) + \{\upsilon^\epsilon(W)-\lambda^\epsilon(W)\}\Ind\{\VEm^0(W)<\theta^0-t+c|\epsilon|\}\right] \\
&= (1+c|\epsilon|)\omega^{\theta^0-t+c|\epsilon|,0} + c|\epsilon|.
\end{align*}
Now, we know that $\omega^{\theta^0-t+c|\epsilon|,0}\le \mu$ by the definition of the supremum, and by \ref{it:omegacontnotflat} the inequality is strict. It follows that $\omega^{\theta^0-t,\epsilon}< \mu + c|\epsilon|(1+\mu)$, and so $\omega^{\theta^0-t,\epsilon}< \mu$ for all $|\epsilon|$ sufficiently small. Thus (i) holds. We now establish (ii). Similarly to the above, we have that $\omega^{\theta^0+t,\epsilon}\ge (1-c|\epsilon|)\omega^{\theta^0+t-c|\epsilon|,0}-c|\epsilon|$ for a constant $c>0$. By the definition of the supremum, $\omega^{\theta^0+t-c|\epsilon|,0}>\mu$ for all $|\epsilon|$ sufficiently small. Thus, (ii) holds for all $|\epsilon|$ sufficiently small. From the observation at the beginning of the proof, the fact that (i) and (ii) hold for all $|\epsilon|$ sufficiently small implies that $\limsup_{\epsilon\rightarrow 0}|\theta^\epsilon-\theta^0|\le t$, and as $t$ was arbitrary this implies that $\theta^\epsilon-\theta^0 = o(1)$.

We now establish that $\theta^\epsilon-\theta^0 = O(\epsilon)$. Using similar techniques to those used above, one can show that there exists a $c>0$ such that, for all $\theta$,
\begin{align}
(1-c|\epsilon|)\omega^{\theta-c|\epsilon|,0}-c|\epsilon|\le \omega^{\theta,\epsilon}\le (1+c|\epsilon|)\omega^{\theta+c|\epsilon|,0}+c|\epsilon|. \label{eq:F0bd}
\end{align}
Noting that $\mu\ge \omega^{\theta^\epsilon,\epsilon}$, we have that
\begin{align*}
0&\ge \omega^{\theta^\epsilon,\epsilon} - \omega^{\theta^0,0}\ge (1-c|\epsilon|)\omega^{\theta^\epsilon-c|\epsilon|,0}-c|\epsilon| - \omega^{\theta^0,0} \\
&= (1-c|\epsilon|)\left[\omega^{\theta^\epsilon-c|\epsilon|,0} - \omega^{\theta^0,0}\right] - c|\epsilon|(1+\omega^{\theta^0,0}).
\end{align*}
As \ref{it:omegacontnotflat} implies $\liminf_{t\rightarrow 0}\frac{\omega^{\theta^0+t,0}-\omega^{\theta^0,0}}{t}>0$ by Lemma~\ref{lem:F0notflat}, $B\triangleq \liminf_{\epsilon\rightarrow 0} B^\epsilon\triangleq \liminf_{\epsilon\rightarrow 0}\frac{\omega^{\theta^\epsilon-c|\epsilon|,0} - \omega^{\theta^0,0}}{\theta^\epsilon-\theta^0-c|\epsilon|}$ is greater than zero. The above implies that, for all $\epsilon$ small enough,
\begin{align*}
\frac{c|\epsilon|(1+\omega^{\theta^0,0})}{1-c|\epsilon|}\ge \omega^{\theta^\epsilon-c|\epsilon|,0} - \omega^{\theta^0,0}= (\theta^\epsilon-\theta^0-c|\epsilon|)B^\epsilon.
\end{align*}
Dividing both sides by $B^\epsilon$ and taking the limit supremum of the left-hand side as $\epsilon\rightarrow 0$ shows that $\theta^\epsilon-\theta^0\le O(\epsilon)$.

The proof of the fact that $\theta^\epsilon-\theta^0\ge O(\epsilon)$ is analogous, making use of the upper bound in (\ref{eq:F0bd}) and the fact that $\mu< \omega^{\theta^\epsilon+c|\epsilon|,\epsilon}$ if $\theta^\epsilon$ is finite (guaranteed for all $\epsilon$ small enough).
\end{proof}

\subsection{Asymptotic linearity (Theorem~\ref{thm:al})}\label{app:thmalproof}
In what follows, we write \ref{it:omegacontnotflat}/\ref{it:goodquantile}/\ref{it:betangood} to mean that \ref{it:omegacontnotflat}, \ref{it:goodquantile}, and \ref{it:betangood} hold. The proof under \ref{it:musmall} is essentially equivalent to that under \ref{it:mubig}, so for simplicity we only give the proof under \ref{it:mubig} and \ref{it:omegacontnotflat}/\ref{it:goodquantile}/\ref{it:betangood}. Outside of formal theorem statements, when we write that \ref{it:omegacontnotflat}/\ref{it:goodquantile}/\ref{it:betangood} holds or that \ref{it:mubig} holds, we mean that the stated condition(s) and the other conditions of Theorem~\ref{thm:al}, namely \ref{it:brcommonsupport}, \ref{it:brvebdd}, \ref{it:CSconsistency}, \ref{it:thetaDonsker}, \ref{it:thetaempproc}, and \ref{it:VEnbdd}, hold.

We prove Theorem~\ref{thm:al} by analyzing the terms on the right-hand side of the following decomposition, whose derivation is straightforward:
\begin{align}
\widehat{\phi}^n-\phi^0=& \frac{\widehat{\gamma}^{\widehat{\beta}^n,n}}{\widehat{\omega}^{\widehat{\beta}^n,n}} - \frac{\gamma^{\beta^0,0}}{\omega^{\beta^0,0}} \nonumber \\
=&\, \frac{\widehat{\gamma}^{\widehat{\beta}^n,n}-\gamma^{\widehat{\beta}^n,0}}{\omega^{\beta^0,0}} - \frac{\gamma^{\beta^0,0}}{[\omega^{\beta^0,0}]^2}\left[\widehat{\omega}^{\widehat{\beta}^n,n} - \omega^{\beta^0,0}\right] + \frac{\gamma^{\widehat{\beta}^n,0} - \gamma^{\beta^0,0}}{\omega^{\beta^0,0}} \nonumber \\
&+ \left[\widehat{\gamma}^{\widehat{\beta}^n,n}-\gamma^{\beta^0,0}\right]\left[\frac{1}{\widehat{\omega}^{\widehat{\beta}^n,n}}-\frac{1}{\omega^{\beta^0,0}}\right] + \frac{\gamma^{\beta^0,0}}{\widehat{\omega}^{\widehat{\beta}^n,n}[\omega^{\beta^0,0}]^2}\left[\widehat{\omega}^{\widehat{\beta}^n,n} - \omega^{\beta^0,0}\right]^2 \nonumber \\
\triangleq&\, \textnormal{Term 1} - \textnormal{Term 2} + \textnormal{Term 3} + \textnormal{Term 4} + \textnormal{Term 5}. \label{eq:term1to5def}
\end{align}
We will show that Term 1 is non-negligible, i.e. contributes to the dominant $O_P(\nmin^{-1/2})$ term in the asymptotically linear expansion, under both \ref{it:omegacontnotflat}/\ref{it:goodquantile}/\ref{it:betangood} and \ref{it:mubig}, Term 2 is negligible, i.e. $o_P(\nmin^{-1/2})$, under \ref{it:omegacontnotflat}/\ref{it:goodquantile}/\ref{it:betangood} and non-negligible under \ref{it:mubig}, and Term 3 is non-negligible under \ref{it:omegacontnotflat}/\ref{it:goodquantile}/\ref{it:betangood} and negligible under \ref{it:mubig}. 
Terms 4 and 5 are remainder terms that one can ignore if they invoke the delta method. Indeed, if we show that $\widehat{\omega}^{\widehat{\beta}^n,n} - \omega^{\beta^0,0}$ and $\widehat{\gamma}^{\widehat{\beta}^n,n}-\gamma^{\beta^0,0}$ are $O_P(\nmin^{-1/2})$, then Terms 4 and 5 are negligible. As this is what we will show under both \ref{it:omegacontnotflat}/\ref{it:goodquantile}/\ref{it:betangood} and \ref{it:mubig}, Terms 4 and 5 can be shown to be negligible by our analysis of the first three terms.

The following theorem establishes the behavior of Term 1 by studying $\widehat{\gamma}^{\widehat{\beta}^n,n}-\gamma^{\widehat{\beta}^n,0}$.
\begin{theorem}[Term 1] \label{thm:daparamal}
If \ref{it:brvebdd}, \ref{it:CSconsistency}, \ref{it:thetaempproc}, and \ref{it:VEnbdd}, then
\begin{align*}
\widehat{\gamma}^{\widehat{\beta}^n,n}-\gamma^{\widehat{\beta}^n,0}&= \sum_{s\in\mathcal{S}} (Q_s^n-P_s^0) \nabla \Gamma_s^{\beta^0,0} + o_P(\nmin^{-1/2}).
\end{align*}
\end{theorem}
The \hyperlink{proof:thmdaparamal}{proof of Theorem~\ref*{thm:daparamal}} is given in Appendix~\ref{app:term1}.

The fact that Term 2 is negligible under \ref{it:omegacontnotflat}/\ref{it:goodquantile}/\ref{it:betangood} is an immediate consequence of \ref{it:goodquantile}. We now present a theorem that establishes the behavior of Term 3 under \ref{it:omegacontnotflat}/\ref{it:goodquantile}/\ref{it:betangood}.
\begin{theorem}[Term 3 under \ref{it:omegacontnotflat}/\ref{it:goodquantile}/\ref{it:betangood}] \label{thm:phitheta}
If \ref{it:brvebdd}, \ref{it:omegacontnotflat}, \ref{it:CSconsistency}, \ref{it:thetaempproc}, \ref{it:goodquantile}, and \ref{it:betangood}, then
\begin{align*}
\gamma^{\widehat{\beta}^n,0} - \gamma^{\beta^0,0}&= - \theta^0\sum_{s\in\mathcal{S}} (Q_s^n-P_s^0) \nabla \Omega_s^{\beta^0,0} + o_P(\nmin^{-1/2}).
\end{align*}
\end{theorem}
The \hyperlink{proof:thmphitheta}{proof of Theorem~\ref*{thm:phitheta}} is given in Appendix~\ref{app:Terms23munotbig}.

We now present a lemma and a theorem that respectively establish the behavior of Terms 3 and 2 under \ref{it:mubig}. Both results are proven in Appendix~\ref{app:Terms23mubig}.
\begin{lemma}[Term 3 under \ref{it:mubig}]\label{lem:QstarnWeq1}
If \ref{it:brvebdd}, \ref{it:mubig}, \ref{it:CSconsistency}, and \ref{it:thetaDonsker}, then $\widehat{\theta}^n=+\infty$, and consequently $Q_\star^n\{\VEm^n(W)<\widehat{\theta}^n\}=1$, with probability approaching one.
\end{lemma}
The \hyperlink{proof:lemQstarnWeq1}{proof of Lemma~\ref*{lem:QstarnWeq1}} is given in Appendix~\ref{app:Terms23mubig}.
By \ref{it:mubig}, $\theta^0=+\infty$, so that the above shows that $\widehat{\theta}^n=\theta^0$, and thus that Term 2 is zero, with probability approaching one.
\begin{theorem}[Term 2 under \ref{it:mubig}]\label{thm:term2mubig}
If \ref{it:brvebdd}, \ref{it:mubig}, \ref{it:CSconsistency}, and \ref{it:thetaempproc}, then
\begin{align*}
\widehat{\omega}^{\widehat{\beta}^n,n} - \omega^{\beta^0,0}&= \sum_{s\in\mathcal{S}} (Q_s^n - P_s^0)\nabla \Omega_s^{\beta^0,0} + o_P(\nmin^{-1/2}).
\end{align*}
\end{theorem}
The \hyperlink{proof:thmterm2mubig}{proof of Theorem~\ref*{thm:term2mubig}} is given in Appendix~\ref{app:Terms23mubig}.

\subsubsection{Term 1 under \ref{it:omegacontnotflat}/\ref{it:goodquantile}/\ref{it:betangood} and \ref{it:mubig}.}\label{app:term1}
We give a lemma before proving Theorem~\ref{thm:daparamal}. 
\begin{lemma} \label{lem:daremsmall}
Under the conditions of Theorem~\ref{thm:daparamal},
\begin{align*}
\gamma^{\widehat{\beta}^n,n}-\gamma^{\widehat{\beta}^n,0} + \sum_{s\in\mathcal{S}} P_s^0 \nabla \Gamma_s^{\widehat{\beta}^n,n} = o_P(\nmin^{-1/2}).
\end{align*}
\end{lemma}
\begin{proof}[Proof of Lemma~\ref{lem:daremsmall}]
Note that
\begin{align}
&\gamma^{\widehat{\beta}^n,n}-\gamma^{\widehat{\beta}^n,0} + \sum_{s\in\mathcal{S}} P_s^0 \nabla \Gamma_s^{\widehat{\beta}^n,n} \nonumber \\
&= \E_\star^0 \left[\UR^{\widehat{\beta}^n,n}(W)\VE^n(W)\right] - \gamma^{\widehat{\beta}^n,0} + \sum_{s=1}^S P_s^0 \nabla \Gamma_s^{\widehat{\beta}^n,n}. \label{eq:phidaremdecomp}
\end{align}
Furthermore,
\begin{align}
\sum_{s=1}^S &P_s^0 \nabla \Gamma_s^{\widehat{\beta}^n,n} \nonumber \\
=&\, \sum_{s=1}^S \E_s^0\left[\frac{dP_{\star}^n}{dP_s^n}(W)\left\{D_{\UR,s}^{\widehat{\beta}^n,n}(O)\VE^n(W)+\UR^{\widehat{\beta}^n,n}(W)D_{\VE,s}^n(O)\right]\right] \nonumber \\
=&\, \E_\star^0\left[\sum_{s=1}^S\left\{D_{\UR,s}^{\widehat{\beta}^n,n}(O)\VE^n(W)+\UR^{\widehat{\beta}^n,n}(W)D_{\VE,s}^n(O)\right\}\right] \label{eq:dPstartrue} \\
&+ \sum_{s=1}^S \E_s^0\left[\left\{\frac{dP_{\star}^n}{dP_s^n}(W)-\frac{dP_{\star}^0}{dP_s^0}(W)\right\}\left\{D_{\UR,s}^{\widehat{\beta}^n,n}(O)\VE^n(W)+\UR^{\widehat{\beta}^n,n}(W)D_{\VE,s}^n(O)\right\}\right]. \nonumber
\end{align}
The law of total expectation and the fact that $D_{\UR,s}^{\widehat{\beta}^n,0}$ and $D_{\VE,s}^0$ are mean zero when applied to a random variable $O_s$ drawn from the conditional distribution $O_s|W$ under $P_s^0$ yield that the latter line is bounded above by
\begin{align*}
\sum_{s=1}^S \E_s^0\Bigg[&\left\{\frac{dP_{\star}^n}{dP_s^n}(W)-\frac{dP_{\star}^0}{dP_s^0}(W)\right\} \\
&\times\Big\{\left(\E_s^0[D_{\UR,s}^{\widehat{\beta}^n,n}(O)|W]-\E_s^0[D_{\UR,s}^{\widehat{\beta}^n,0}(O)|W]\right)\VE^n(W) \\
&\hspace{2em}+\UR^{\widehat{\beta}^n,n}(W)\left(\E_s^0[D_{\VE,s}^n(O)|W]-\E_s^0[D_{\VE,s}^0(O)|W]\right)\Big\}\Bigg].
\end{align*}
The above is $o_P(\nmin^{-1/2})$ by Cauchy-Schwarz, \ref{it:CSconsistency}, \ref{it:VEnbdd}, and the fact that $\UR^{\widehat{\beta}^n,n}$, $\VE^n$ have bounded range. Furthermore, (\ref{eq:dPstartrue}) simplifies to
\begin{align*}
\E_\star^0&\left[\sum_{s=1}^S\left\{D_{\UR,s}^{\widehat{\beta}^n,n}(O)\VE^n(W)+\UR^{\widehat{\beta}^n,n}(W)D_{\VE,s}^n(O)\right\}\right] \\
=&\, \E_\star^0\left[\left\{[\UR^{\widehat{\beta}^n,0}(W)-\UR^{\widehat{\beta}^n,n}(W)]\VE^n(W)+\UR^{\widehat{\beta}^n,n}(W)[\VEm^0(W)-\VE^n(W)]\right\}\right] \\
&+\E_\star^0\left[\left\{\UR^{\widehat{\beta}^n,n}(W)-\UR^{\widehat{\beta}^n,0}(W) + \sum_{s=1}^S D_{\UR,s}^{\widehat{\beta}^n,n}(O)\right\}\VE^n(W)\right] \\
&+\E_\star^0\left[\UR^{\widehat{\beta}^n,n}(W)\Rem_2(\mathcal{P}^n,\mathcal{P}^0)\right].
\end{align*}
The final line is $o_P(\nmin^{-1/2})$ by \ref{it:CSconsistency}. The magnitude of the second line is upper bounded by $\E_\star^0\left[\Rem_1(\mathcal{P}^n,\mathcal{P}^0)(W)|\VE^n(W)|\right]$, which is also $o_P(\nmin^{-1/2})$ by \ref{it:CSconsistency}. We have thus established that (\ref{eq:phidaremdecomp}) rewrites as
\begin{align*}
&\gamma^{\widehat{\beta}^n,n} -\gamma^{\widehat{\beta}^n,0} + \sum_{s\in\mathcal{S}} P_s^0 \nabla \Gamma_s^{\widehat{\beta}^n,n} \\
=&\, \E_\star^0\left[\{\UR^{\widehat{\beta}^n,n}(W)-\UR^{\widehat{\beta}^n,0}(W)\}\{\VEm^0(W)-\VE^n(W)\}\right] + o_P(\nmin^{-1/2}).
\end{align*}
By Cauchy-Schwarz and \ref{it:CSconsistency}, the leading term on the right is $o_P(\nmin^{-1/2})$, so that indeed the entire right-hand side is $o_P(\nmin^{-1/2})$.
\end{proof}
We now prove Theorem~\ref{thm:daparamal}.
\begin{proof}[Proof of Theorem~\ref{thm:daparamal}]\hypertarget{proof:thmdaparamal}{}
Rearranging the result of Lemma~\ref{lem:daremsmall} yields that
\begin{align*}
\gamma^{\widehat{\beta}^n,n}-\gamma^{\widehat{\beta}^n,0}&= - \sum_{s\in\mathcal{S}} P_s^0 \nabla \Gamma_s^{\widehat{\beta}^n,n} + o_P(\nmin^{-1/2}).
\end{align*}
Adding $\sum_{s\in\mathcal{S}} Q_s^n \nabla \Gamma_s^{\widehat{\beta}^n,n}$ to both sides shows that the one-step estimator $\widehat{\gamma}^{\widehat{\beta}^n,n}$ satisfies the identity
\begin{align*}
\widehat{\gamma}^{\widehat{\beta}^n,n}-\gamma^{\widehat{\beta}^n,0}&= \sum_{s\in\mathcal{S}} (Q_s^n-P_s^0) \nabla \Gamma_s^{\widehat{\beta}^n,n} + o_P(\nmin^{-1/2}).
\end{align*}
Applying \ref{it:thetaempproc} allows one to replace each instance of $\nabla \Gamma_s^{\widehat{\beta}^n,n}$ above by $\nabla \Gamma_s^{\beta^0,0}$, $s\in\mathcal{S}$. 
\end{proof}

\subsubsection{Term 3 under \ref{it:omegacontnotflat}/\ref{it:goodquantile}/\ref{it:betangood}.}\label{app:Terms23munotbig}
We first give two lemmas, and then we establish control of Term 3 by proving Theorem~\ref{thm:phitheta}.

\begin{lemma} \label{lem:thetaee}
If \ref{it:brvebdd} and \ref{it:CSconsistency}, then
\begin{align*}
\left|\widehat{\omega}^{\widehat{\beta}^n,n}-\omega^{\widehat{\beta}^n,0} - \sum_{s\in\mathcal{S}} (Q_s^n-P_s^0) \nabla \Omega_s^{\widehat{\beta}^n,n}\right|&= o_P(\nmin^{-1/2}).
\end{align*}
\end{lemma}
\begin{proof}[Proof of Lemma~\ref{lem:thetaee}]
Note that
\begin{align*}
&\omega^{\widehat{\beta}^n,n} + \sum_{s\in\mathcal{S}} P_s^0 \nabla \Omega_s^{\widehat{\beta}^n,n} \\
=&\, \E_\star^0\left[\UR^{\widehat{\beta}^n,n}(W)\right] + \sum_{s=1}^S \E_s^0\left[\frac{dP_{\star}^n}{dP_s^n}(W)D_{\UR,s}^{\widehat{\beta}^n,n}(O)\right] \\
=&\, \omega^{\widehat{\beta}^n,0} + \E_\star^0\left[\UR^{\widehat{\beta}^n,n}(W) - \UR^{\widehat{\beta}^n,0}(W) + \sum_{s=1}^S \E_s^0\left[D_{\UR,s}^{\widehat{\beta}^n,n}(O)\middle|W\right]\right] \\
&+ \sum_{s=1}^S \E_s^0\left[\left\{\frac{dP_{\star}^n}{dP_s^n}(W)-\frac{dP_{\star}^0}{dP_s^0}(W)\right\}\left\{\E_s^0\left[D_{\UR,s}^{\widehat{\beta}^n,n}(O)\middle|W\right] - \E_s^0\left[D_{\UR,s}^{\widehat{\beta}^n,0}(O)\middle|W\right]\right\}\right],
\end{align*}
where the final equality uses that $\E_s^0\left[D_{\UR,s}^{\widehat{\beta}^n,0}(O)\middle|W\right] = 0$. The second line is $o_P(\nmin^{-1/2})$ by \ref{it:CSconsistency} and Cauchy-Schwarz. The triangle inequality readily yields that
\begin{align*}
\E_\star^0\left|\UR^{\widehat{\beta}^n,n}(W) - \UR^{\widehat{\beta}^n,0}(W) + \sum_{s=1}^S \E_s^0\left[D_{\UR,s}^{\widehat{\beta}^n,n}(O)\middle|W\right]\right|&\le  \E_\star^0\left|\Rem_1(\mathcal{P}^n,\mathcal{P}^0)(W)\right|,
\end{align*}
and the right-hand side is $o_P(\nmin^{-1/2})$ by \ref{it:CSconsistency}. The fact that $\widehat{\omega}^{\widehat{\beta}^n,n}\triangleq \omega^{\widehat{\beta}^n,n} + \sum_{s\in\mathcal{S}} Q_s^n \nabla \Omega_s^{\widehat{\beta}^n,n}$ completes the proof.
\end{proof}

\begin{proof}[Proof of Theorem~\ref{thm:phitheta}]\hypertarget{proof:thmphitheta}{}
Note that
\begin{align*}
\gamma^{\widehat{\beta}^n,0} - \gamma^{\beta^0,0}=& \E_\star^0\left[\left\{\UR^{\widehat{\beta}^n,0}(W)-\UR^{\beta^0,0}(W)\right\}\VEm^0(W)\right] \\
=& \theta^0\left[\omega^{\widehat{\beta}^n,0} - \omega^{\beta^0,0}\right] + \left\{\gamma^{\widehat{\beta}^n,0} - \gamma^{\beta^0,0}-\theta^0\left[\omega^{\widehat{\beta}^n,0} - \omega^{\beta^0,0}\right]\right\}.
\end{align*}
By \ref{it:betangood}, the latter term above is $o_P(\nmin^{-1/2})$. By Lemma~\ref{lem:thetaee} and \ref{it:goodquantile},
\begin{align*}
\omega^{\widehat{\beta}^n,0}&= \omega^{\widehat{\beta}^n,n} + \sum_{s\in\mathcal{S}} Q_s^n \nabla \Omega_s^{\widehat{\beta}^n,n} - \sum_{s\in\mathcal{S}} (Q_s^n-P_s^0) \nabla \Omega_s^{\widehat{\beta}^n,0} + o_P(\nmin^{-1/2}) \\
&= \mu - \sum_{s\in\mathcal{S}} (Q_s^n-P_s^0) \nabla \Omega_s^{\widehat{\beta}^n,0} + o_P(\nmin^{-1/2}).
\end{align*}
By \ref{it:omegacontnotflat}, $\omega^{\beta^0,0}=\mu$, and so
\begin{align*}
\theta^0\left[\omega^{\widehat{\beta}^n,0} - \omega^{\beta^0,0}\right]
&= - \theta^0\sum_{s\in\mathcal{S}} (Q_s^n-P_s^0) \nabla \Omega_s^{\widehat{\beta}^n,0} + o_P(\nmin^{-1/2}).
\end{align*}
By \ref{it:thetaempproc}, we can replace each $\Omega_s^{\widehat{\beta}^n,0}$ above by $\Omega_s^{\beta^0,0}$.
\end{proof}

\subsubsection{Terms 2 and 3 under \ref{it:mubig}.}\label{app:Terms23mubig}
We first prove Lemma~\ref{lem:QstarnWeq1}, thereby showing that Term 3 is equal to zero with probability approaching $1$. We then establish control over Term 2 by proving Theorem~\ref{thm:term2mubig}.
\begin{proof}[Proof of Lemma~\ref{lem:QstarnWeq1}]\hypertarget{proof:lemQstarnWeq1}{}
By \ref{it:mubig}, $\theta^0=+\infty$. We show that $\widehat{\omega}^{\beta^0,n} \le \mu$, and consequently $\widehat{\beta}^n=+\infty$, with probability approaching one. Note that
\begin{align*}
\widehat{\omega}^{\beta^0,n} - \omega^{\beta^0,0}&= \sum_{s\in\mathcal{S}} (Q_s^n - P_s^0) \nabla \Omega_s^{\beta^0,n} + \left\{\widehat{\omega}^{\beta^0,n} - \omega^{\beta^0,0} + \sum_{s\in\mathcal{S}} P_s^0 \nabla \Omega_s^{\beta^0,n}\right\}.
\end{align*}
The first term on the right is $O_P(\nmin^{-1/2})$ by \ref{it:thetaDonsker}, and the latter term is $o_P(\nmin^{-1/2})$ by \ref{it:CSconsistency}. Hence, $\widehat{\omega}^{\beta^0,n} = \omega^{\beta^0,0} + O_P(\nmin^{-1/2})$. By \ref{it:mubig}, $\omega^{\beta^0,0}<\mu$, and thus $\widehat{\omega}^{\beta^0,n}< \mu$ with probability approaching one. As $\beta^0\equiv 1$, it must be the case that $\widehat{\theta}^n=+\infty$.
\end{proof}

\begin{proof}[Proof of Theorem~\ref{thm:term2mubig}]\hypertarget{proof:thmterm2mubig}{}
By Lemma~\ref{lem:QstarnWeq1}, $\widehat{\theta}^n=+\infty$ with probability approaching $1$. Suppose this holds. In this case the TMLE step ensures that $\widehat{\omega}^{\widehat{\beta}^n,n}=Q_\star^n \upsilon^n$. Furthermore, \ref{it:mubig} implies that $\omega^{\beta^0,0}=P_\star^0 \upsilon^0$. By the definition of $\Rem_1$, we have that
\begin{align*}
\widehat{\omega}^{\widehat{\beta}^n,n} - \omega^{\beta^0,0}&= (Q_\star^n-P_\star^0) \upsilon^n + P_\star^0 (\upsilon^n-\upsilon^0) \\
&= (Q_\star^n-P_\star^0) \nabla \Omega_\star^{\widehat{\beta}^n,n} - \sum_{s=1}^S P_\star^0 D_{\UR,s}^{\widehat{\beta}^n,n} + P_\star^0 \Rem_1(\mathcal{P}^n,\mathcal{P}^0).
\end{align*}
By \ref{it:CSconsistency}, $P_\star^0 \Rem_1(\mathcal{P}^n,\mathcal{P}^0)=o_P(\nmin^{-1/2})$. Furthermore,
\begin{align*}
\sum_{s=1}^S P_\star^0 D_{\UR,s}^{\widehat{\beta}^n,n}&= \sum_{s=1}^S P_s^0 \nabla \Omega_s^{\widehat{\beta}^n,n} + \sum_{s=1}^S P_s^0 \left[\frac{dP_\star^0}{dP_s^0}-\frac{dP_\star^n}{dP_s^n}\right] D_{\UR,s}^{\widehat{\beta}^n,n}.
\end{align*}
Using that each $D_{\UR,s}^{\beta^0}(O)$ is mean zero for $O\sim P_s^0$ (conditionally on $w$), the $D_{\UR,s}^{\widehat{\beta}^n,n}$ on the right-hand side above can be replaced by $D_{\UR,s}^{\widehat{\beta}^n,n}-D_{\UR,s}^{\beta^0}$, thereby showing that the latter term above is $o_P(\nmin^{-1/2})$ by \ref{it:CSconsistency}. Hence,
\begin{align*}
\widehat{\omega}^{\widehat{\beta}^n,n} - \omega^{\beta^0,0}&= (Q_\star^n-P_\star^0) \nabla \Omega_\star^{\widehat{\beta}^n,n} - \sum_{s=1}^S P_s^0 \nabla \Omega_s^{\widehat{\beta}^n,n} + o_P(\nmin^{-1/2}).
\end{align*}
By Steps \ref{it:tmle1}, \ref{it:tmle2}, and \ref{it:tmle3} of our estimation procedure and the fact that $\widehat{\theta}^n=+\infty$, $\sum_{s=1}^S Q_s^n \nabla \Omega_s^{\widehat{\beta}^n,n}=0$ \citep[these steps represent a standard logistic fluctuation submodel for a TMLE, see][]{vanderLaan&Rose11}. This shows that $\widehat{\omega}^{\widehat{\beta}^n,n} - \omega^{\beta^0,0}=\sum_{s\in\mathcal{S}} (Q_s^0-P_s^0) \nabla \Omega_s^{\widehat{\beta}^n,n} + o_P(\nmin^{-1/2})$. By \ref{it:thetaempproc}, it follows that $\widehat{\omega}^{\widehat{\beta}^n,n} - \omega^{\beta^0,0}=\sum_{s\in\mathcal{S}} (Q_s^0-P_s^0) \nabla \Omega_s^{\beta^0,0} + o_P(\nmin^{-1/2})$.
\end{proof}

\section{More interpretable condition for \ref{it:betangood}} \label{app:moreinterpretablebetangood}
We now provide a more interpretable sufficient condition for \ref{it:betangood}. First note that $\Rem_3^n$ rewrites as
\begin{align*}
\Rem_3^n = \E_\star^0\left[\left\{\UR^{\widehat{\beta}^n,0}(W)-\UR^{\beta^0,0}(W)\right\}\left\{\VEm^0(W)-\theta^0\right\}\right].
\end{align*}
The expression in the expectation above is small when either $\widehat{\beta}^n(W)=\beta^0(W)$ or $\VEm^0(W)$ is close to $\theta^0$. This observation gives some hope that the above expectation will be small when $\VEm^n$ and $\widehat{\theta}^n$ are good estimates of $w\mapsto \VEm^0(w)$ and $\theta^0$, because it will likely only be most difficult to correctly specify $\Ind_{\{\VEm^n(W)\ge\widehat{\theta}^n\}}$ when $\VEm^0(W)-\theta^0$ is small.

We now make this claim precise. The following margin condition is analogous to that used by \cite{Audibert&Tsybakov2007} in the classification context. In particular, for each $\alpha>0$ we define Condition $\alpha$ as follows:
\begin{align*}
P_0^\star\left\{0<|\VEm^0(W)-\theta^0|\le t\right\}&\lesssim t^{\alpha}\textnormal{ for all }t>0.
\end{align*}
The above states that $\VEm^0(W)$ does not place too much mass in the neighborhood of the decision boundary $\theta^0$ that appears in the worst-case unvaccinated risk $\UR^{\beta^0,0}$. The following theorem is an adaptation of Lemma 5.2 in \cite{Audibert&Tsybakov2007}. A similar adaptation was given in \cite{Luedtke&vanderLaan2016}.
\begin{lemma}
If \ref{it:omegacontnotflat} and Condition $\alpha$ holds for a given $\alpha>0$, then
\begin{align*}
&|\Rem_3^n|\lesssim \norm{(\VEm^n(W)-\widehat{\theta}^n)-(\VEm^0(W)-\theta^0)}_{2,P_\star^0}^{2(1+\alpha)/(2+\alpha)} \\
&|\Rem_3^n|\lesssim \norm{(\VEm^n(W)-\widehat{\theta}^n)-(\VEm^0(W)-\theta^0)}_{\infty,P_\star^0}^{1+\alpha}.
\end{align*}
\end{lemma}
Shortly we will show that we will show that we can get a faster rate of estimation on the univariate parameter $\widehat{\theta}^n$ than on the infinite-dimensional parameter $w\mapsto \VEm^0(w)$ when $\alpha\ge 1$. Hence, in this case the above allows us to map our rate of convergence of $w\mapsto \VEm^0(w)$ into a rate of decay for the remainder term $\Rem_3^n$. Suppose that $\alpha=1$, which holds if $w\mapsto \VEm^0(w)$ has bounded Lebesgue density in a neighborhood of $\theta^0$. The supremum norm result in the above lemma suggests that $\VEm^n-\widehat{\theta}^n$ converging to $\VEm^0-\theta^0$ at a rate faster than $\nmin^{-1/4}$ will suffice to make $\Rem_3^n$ negligible.

We close this section by showing the rate of estimation that we can obtain on $\theta^0$ provided the following additional regularity condition is satisfied:
\begin{enumerate}[resume*=en:regconds]
	\item\label{it:VEmL2} $\norm{\VEm^n-\VEm^0}_{2,P_\star^0}=o_P(1)$.
\end{enumerate}
The above is very mild if \ref{it:CSconsistency} holds.
\begin{lemma} \label{lem:thetancons}
If \ref{it:brvebdd}, \ref{it:omegacontnotflat}, \ref{it:CSconsistency}, \ref{it:thetaDonsker}, \ref{it:goodquantile}, and \ref{it:VEmL2}, then $\widehat{\theta}^n=\theta^0 + o_P(1)$. If Condition $\alpha$ also holds, then we have the stronger result that
\begin{align*}
\widehat{\theta}^n - \theta^0&= O_P\left(\nmin^{-1/2}+\norm{\VEm^n-\VEm^0}_{2,P_\star^0}^{\frac{2\alpha}{\alpha+1}}\right).
\end{align*}
\end{lemma}
Suppose, as is typical, that $\norm{\VEm^n-\VEm^0}_{2,P_\star^0}$ shrinks slower than $\nmin^{-1/2}$. In this case, if $\alpha\ge 1$, then the above gives conditions under which $\widehat{\theta}^n=\theta^0 + O_P(\norm{\VEm^n-\VEm^0}_{2,P_\star^0})$, where the big-oh can be replaced by a little-oh if $\alpha>1$.
\begin{proof}[Proof of Lemma~\ref{lem:thetancons}]
By \ref{it:goodquantile}, Lemma~\ref{lem:thetaee}, and \ref{it:thetaDonsker},
\begin{align*}
o_P(\nmin^{-1/2})=\widehat{\omega}^{\widehat{\beta}^n,n} - \mu
= \omega^{\widehat{\beta}^n,0}-\omega^{\beta^0,0} + O_P(\nmin^{-1/2}).
\end{align*}
Let $\beta^{n,0}$ denote the function $w\mapsto \Ind_{\{\VEm^0(w)<\widehat{\theta}^n\}}$. The above shows that
\begin{align*}
\omega^{\beta^{n,0},0}-\omega^{\beta^0,0}&= -[\omega^{\widehat{\beta}^n,0}-\omega^{\beta^{n,0},0}] + O_P(\nmin^{-1/2})
\end{align*}
At the end of this proof, we show that $\omega^{\widehat{\beta}^n,0} - \omega^{\beta^{n,0},0} = o_P(1)$ without using Condition $\alpha$. For now suppose we have established this. By \ref{it:omegacontnotflat}, $\theta\mapsto\omega^{w\mapsto\Ind\{\VEm^0(w)<\theta\},0}$ is continuous and increasing at $\theta^0$ so that $\omega^{\beta^{n,0},0}-\omega^{\beta^0,0} = o_P(1)$ is only possible if $\widehat{\theta}^n-\theta^0=o_P(1)$. If Condition $\alpha$ also holds, then we will show at the end of this proof that $\omega^{\widehat{\beta}^n,0} - \omega^{\beta^{n,0},0} = O_P(\norm{\VEm^n-\VEm^0}_{2,P_\star^0}^{\frac{2\alpha}{\alpha+1}})$. Noting that
\begin{align*}
\omega^{\beta^{n,0},0}-\omega^{\beta^0,0}&= \frac{\omega^{\beta^{n,0},0}-\omega^{\beta^0,0}}{\widehat{\theta}^n-\theta^0}[\widehat{\theta}^n-\theta^0],
\end{align*}
we see that \ref{it:omegacontnotflat} and $\widehat{\theta}^n-\theta^0=o_P(1)$ imply that $\widehat{\theta}^n-\theta^0=O_P(\omega^{\beta^{n,0},0}-\omega^{\beta^0,0})$, which we have shown to be $O_P(\nmin^{-1/2}+\norm{\VEm^n-\VEm^0}_{2,P_\star^0}^{\frac{2\alpha}{\alpha+1}})$ under Condition $\alpha$.

We now establish that $\omega^{\widehat{\beta}^n,0} = \omega^{\beta^{n,0},0} + o_P(1)$ regardless of the validity of Condition $\alpha$, and that $\omega^{\widehat{\beta}^n,0} = \omega^{\beta^{n,0},0} + O_P(\norm{\VEm^n-\VEm^0}_{2,P_\star^0}^{\frac{2\alpha}{\alpha+1}})$ under Condition $\alpha$. For simplicity we give the proof when $\widehat{\eta}^n=0$, though the proof for general $\widehat{\eta}^n\in[0,1]$ only differs slightly. Observe that
\begin{align*}
&\left|\omega^{\widehat{\beta}^n,0}-\omega^{\beta^{n,0},0}\right| \\
&= \left|\E_\star^0\left[\{\upsilon^0(W)-\lambda^0(W)\}\left\{\Ind_{\{\VEm^n(W)<\widehat{\theta}^n\}}-\Ind_{\{\VEm^0(W)<\widehat{\theta}^n\}}\right\}\right]\right| \\
&\le \E_\star^0\left[\left|\Ind_{\{\VEm^n(W)<\widehat{\theta}^n\}}-\Ind_{\{\VEm^0(W)<\widehat{\theta}^n\}}\right|\right] \\
&\le P_\star^0\left\{0\le |\VEm^0(W)-\widehat{\theta}^n|\le |\VEm^n(W)-\VEm^0(W)|\right\}.
\intertext{For any $t>0$, the inequality continues as}
&\le P_\star^0\left\{0\le |\VEm^0(W)-\widehat{\theta}^n|\le t\right\} + P_\star^0\left\{|\VEm^n(W)-\VEm^0(W)|\ge t\right\} \\
&\le P_\star^0\left\{0\le |\VEm^0(W)-\widehat{\theta}^n|\le t\right\} + \frac{\norm{\VEm^n-\VEm^0}_{2,P_\star^0}^2}{t^2}.
\end{align*}
By \ref{it:omegacontnotflat}, $P_\star^0\{\VEm^0(W)=\widehat{\theta}^n\}=0$, so that the former term satisfies
\begin{align*}
\lim_{t\downarrow 0} P_\star^0\left\{0\le |\VEm^0(W)-\widehat{\theta}^n|\le t\right\} = 0.
\end{align*}
As $\norm{\VEm^n-\VEm^0}_{2,P_\star^0}=o_P(1)$ by \ref{it:VEmL2}, one can choose a sequence $\{t_n\}$ and plug it in for $t$ in the preceding inequality for $\left|\omega^{\widehat{\beta}^n,0}-\omega^{\beta^{n,0},0}\right|$ to see that this quantity is $o_P(1)$. If Condition $\alpha$ holds, then one can choose $t_n=\norm{\VEm^n-\VEm^0}_{2,P_\star^0}^{2/(\alpha+1)}$, yielding the stronger result
\begin{align*}
\left|\omega^{\widehat{\beta}^n,0}-\omega^{\beta^{n,0},0}\right|&\lesssim \norm{\VEm^n-\VEm^0}_{2,P_\star^0}^{\frac{2\alpha}{\alpha+1}}.
\end{align*}
\end{proof}

\section{Alternative to Steps \ref{it:tmle1}, \ref{it:tmle2}, and \ref{it:tmle3} in our estimation scheme}\label{app:altalg}
We now present alternatives to Steps \ref{it:tmle1}, \ref{it:tmle2}, and \ref{it:tmle3}, to be used when neither (i) $\ell_s(w)$ is a constant multiple of $\mathbbmss{u}_s(w)$ nor (ii) $\ell_s\equiv 0$ for all $s$ holds. The estimation scheme is identical to that presented in the main text besides the modification of these three steps.
\begin{enumerate}
	\item[\ref*{it:tmle1}')]\hypertarget{it:tmle1prime}{} Fit a bivariate logistic regression with outcome $\left(y_s[i] : s=1,\ldots,S;\,i=1,\ldots,n_s\right)$, covariates $\left(\frac{\Ind_{\{a_s[i]=0\}}\ell_s(w)}{n_s P_s^{n,\textnormal{init}}(A=0|w_s[i])}\frac{dP_\star^{n,\textnormal{init}}}{dP_s^{n,\textnormal{init}}}(w_s[i]) : s=1,\ldots,S;\,i=1,\ldots,n_s\right)$ and $\left(\frac{\Ind_{\{a_s[i]=0\}}\mathbbmss{u}_s(w)}{n_s P_s^{n,\textnormal{init}}(A=0|w_s[i])}\frac{dP_\star^{n,\textnormal{init}}}{dP_s^{n,\textnormal{init}}}(w_s[i]) : s=1,\ldots,S;\,i=1,\ldots,n_s\right)$, and fixed, subject-level intercept $\left(\logit\left(\E_s^{n,\textnormal{init}}[Y|A=0,w_s[i]]\right) : s=1,\ldots,S;\,i=1,\ldots,n_s\right)$. Denote the fitted coefficient in front of the respective covariates by $\epsilon_n^\ell$ and $\epsilon_n^u$.
	\item[\ref*{it:tmle2}')]\hypertarget{it:tmle2prime}{} For each $s=1,\ldots,S$, let $(a,w)\mapsto \E_s^{n,\epsilon_n}[Y|a,w]$ denote the function
	\begin{align*}
	(a,w)\mapsto \logit^{-1}\Bigg[&\logit\left(\E_s^{n,\textnormal{init}}[Y|a,w]\right) \\
	&+ \epsilon_n^\ell\frac{\Ind_{\{a=0\}}\ell_s(w)}{n_s P_s^{n,\textnormal{init}}(A=0|w)}\frac{dP_\star^{n,\textnormal{init}}}{dP_s^{n,\textnormal{init}}}(w) \\
	&+ \epsilon_n^u\frac{\Ind_{\{a=0\}}\mathbbmss{u}_s(w)}{n_s P_s^{n,\textnormal{init}}(A=0|w)}\frac{dP_\star^{n,\textnormal{init}}}{dP_s^{n,\textnormal{init}}}(w)\Bigg].
	\end{align*}
	\item[\ref*{it:tmle3}')]\hypertarget{it:tmle3prime}{} Let $\mathcal{P}^n=(P_\star^n,P_1^n,\ldots,P_s^n)$ denote any collection of distribution satisfying that, for all $(a,w)$, $\E_s^n[Y|a,w] = \E_s^{n,\epsilon_n}[Y|a,w]$, $P_s^n(A=1|w) = P_s^{n,\textnormal{init}}(A=1|w)$, and $\frac{dP_\star^n}{dP_s^n}(w) = \frac{dP_\star^{n,\textnormal{init}}}{dP_s^{n,\textnormal{init}}}(w)$.
\end{enumerate}

\section{Extensions}\label{app:extensions}
\subsection{Two-phase sampling}\label{app:twophase}
Suppose now that the data is collected via a two-phase sampling scheme in a given trial $s\in\{1,\ldots,S\}$. In particular, suppose that $W$ is collected on only a subset of participants, whereas $(L,A,Y)$ is collected on all participants, where $L$ is a biomarker or collection of biomarkers that may be predictive of $W$. It is not essential that $L$ happens temporally before $A$ and $Y$. If no biomarker $L$ is observed, then one can set $L=0$ for all participants. Let $\Delta$ be an indicator of the missingness of $W$. In this setting, the full data distribution for trial $s$ is $(W,L,A,Y)\sim P_s^{0,F}$, and the observed data structure for trial $s$ is given by $O_s\triangleq (\Delta W,\Delta,L,A,Y)\sim P_s^{0,F}$. We suppose that $W$ is missing at random, in the sense that $\Delta\independent W | (L,A,Y)$ for each trial $s$, and further that censoring mechanism, i.e. the probability that $\Delta=1$ given each realization of $(L,A,Y)$, is known. This ensures that $P_s^{0,F}$ can be identified with $P_s^0$ by the G-computation formula \cite{Robins1986}. In particular, for any event $E$ on $O_s$, we have that $P_s^{0,F}\{E\}= \E_s^0[P_s^0\{E|\Delta=1,L\}]$.

A common sampling scheme that generates such data in vaccine efficacy trials is a nested case-control sampling scheme \citep{Breslow1996}, where the outcome $W$ is observed on all cases (participants with $Y=1$), and the outcome $W$ is only observed on a subset of controls ($Y=0$). Often these sampling schemes will take the form of an $m$:1 scheme, such that, for each case with $W$ observed, $W$ is observed for $m$ controls. While technically the indicator $\Delta$ is drawn without replacement for these $m$ individuals, one can typically ignore this dependence in the data with little impact on precision or coverage.

A simple modification of our procedure via inverse probability weighting allows estimation of our efficacy lower bound $\phi^0$. For efficiency gains, we recommend estimating the censoring mechanism even though it is known \citep[see Theorem~2.3 in][]{vdL02}. While estimating the censoring mechanism improves the precision of our estimator, it does not reduce the width of our confidence intervals. This thus leads to a conservative inferential procedure. The proof of correctness of this approach is beyond the scope of this work, though closely follows the arguments given in \cite{Rose2011}, with minor tweaks to account for the multiple sample nature of the problem. Indeed, Steps~\ref{it:tmle1twophase}, \ref{it:tmle2twophase}, and \ref{it:tmle3twophase} constitute an IPCW-TMLE, as presented in \cite{Rose2011}. For simplicity we only give the algorithm for the case presented in Section~\ref{sec:eststeps} in the main text, namely where either (i) the chosen $\ell_s(w)$ is a constant multiple of the chosen $\mathbbmss{u}_s(w)$, where the multiple is independent of $s$ and $w$, or (ii) the chosen $\ell_s\equiv 0$ for all $s$ so that $\lambda^0\equiv 0$. The modification of Steps~\ref{it:tmle1twophase}, \ref{it:tmle2twophase}, and \ref{it:tmle3twophase} is analogous to the modifications made to the algorithm in the main text given in Appendix~\ref{app:altalg}. We start at Step~\ref{it:censmechtwophase} so that the other steps parallel those given for the algorithm in the main text.

To emphasize the fact that many of the parameters below depend on the full data structure $(W,L,A,Y)$ rather than the (censored) observed data structure $(\Delta W,\Delta,L,A,Y)$, we replace the ``$0$'' in the superscript by ``$0,F$'' when the parameter is defined for the full data distribution, e.g. we write $\lambda^{0,F}$ rather than $\lambda^0$. Similarly, we write $\lambda^{n,F}$ rather than $\lambda^n$ when denoting estimates of parameters of the full data distribution. Finally, we note that below we denote the observed value of $L$ and $\Delta$ for participant $i$ from trial $s$ by $l_s[i]$ and $\delta_s[i]$, respectively.
\begin{enumerate}[start=0]
	\item\label{it:censmechtwophase} Estimate each $(l,a,y)\mapsto P_s^0\{\Delta=1|l,a,y\}$ using a completed trial $s$-specific correctly specified parametric model. Standardize these estimates by an appropriate constant so that, for each $s=1,\ldots,S$,
	\begin{align*}
	 \sum_{i=1}^{n_s} \frac{\delta_s[i]}{P_s^n\{\Delta=1|l_s[i],a_s[i],y_s[i]\}} = n_s.
	\end{align*}
	Note: correct parametric model specification is possible in this context due to the presumed knowledge (by experimental design) of the censoring mechanism $P_s^0\{\Delta=1|l,a,y\}$.
	\item\label{it:initeststwophase} Let $(a,w)\mapsto \E_s^{n,F,\textnormal{init}}[Y|a,w]$, $w\mapsto P_s^{n,F,\textnormal{init}}(A=1|w)$, and $w\mapsto \frac{dP_\star^{n,F,\textnormal{init}}}{dP_s^{n,F,\textnormal{init}}}(w)$ represent estimates of $(a,w)\mapsto \E_s^{0,F}[Y|a,w]$, $w\mapsto P_s^{0,F}(A=1|w)$, and $w\mapsto \frac{dP_\star^0}{dP_s^{0,F}}(w)$, respectively.\\
	Note: \cite{Rose2011} describe a weighted loss-based estimation scheme that leverages the information in the biomarker $L$ when estimating these quantities. This procedure makes use of the estimate of the censoring mechanism from Step~\ref{it:censmechtwophase}. \cite{Rose2011} also extend the super-learner of \cite{vanderLaan&Polley&Hubbard07} to two-phase sampling designs.
	\item\label{it:tmle1twophase} Fit a weighted univariate logistic regression with weights\\ $\left(\frac{\Delta}{P_s^n\{\delta_s[i]=1|l_s[i],a_s[i],y_s[i]\}} : s=1,\ldots,S;\,i=1,\ldots,n_s\right)$, outcome\\
	 $\left(y_s[i] : s=1,\ldots,S;\,i=1,\ldots,n_s\right)$, covariate\\
	 $\left(\frac{\Ind_{\{a_s[i]=0\}}\mathbbmss{u}_s(w)}{n_s P_s^{n,F,\textnormal{init}}(A=0|w_s[i])}\frac{dP_\star^{n,F,\textnormal{init}}}{dP_s^{n,F,\textnormal{init}}}(w_s[i]) : s=1,\ldots,S;\,i=1,\ldots,n_s\right)$, fixed, subject-level intercept $\left(\logit\left(\E_s^{n,F,\textnormal{init}}[Y|A=0,w_s[i]]\right) : s=1,\ldots,S;\,i=1,\ldots,n_s\right)$. Denote the fitted coefficient in front of the covariate by $\epsilon_n$.
	\item\label{it:tmle2twophase} For each $s=1,\ldots,S$, let $(a,w)\mapsto \E_s^{n,F,\epsilon_n}[Y|a,w]$ denote the function
	\begin{align*}
	(a,w)\mapsto \logit^{-1}\left[\logit\left(\E_s^{n,F,\textnormal{init}}[Y|a,w]\right) + \epsilon_n\frac{\Ind_{\{a=0\}}\mathbbmss{u}_s(w)}{n_s P_s^{n,F,\textnormal{init}}(A=0|w)}\frac{dP_\star^{n,F,\textnormal{init}}}{dP_s^{n,F,\textnormal{init}}}(w)\right].
	\end{align*}
	\item\label{it:tmle3twophase} Let $\mathcal{P}^{n,F}=(P_\star^n,P_1^{n,F},\ldots,P_s^{n,F})$ denote any collection of distribution satisfying that, for all $(a,w)$, $\E_s^{n,F}[Y|a,w] = \E_s^{n,F,\epsilon_n}[Y|a,w]$, $P_s^{n,F}(A=1|w) = P_s^{n,F,\textnormal{init}}(A=1|w)$, and $\frac{dP_\star^n}{dP_s^{n,F}}(w) = \frac{dP_\star^{n,F,\textnormal{init}}}{dP_s^{n,F,\textnormal{init}}}(w)$. Furthermore, for each completed trial $s$, let $Q_s^{n,F}$ denote the distribution that puts mass proportional to $\delta_s[i]/(n_s P_s^n\{\Delta=1|l_s[i],a_s[i],y_s[i]\})$ at each observation $i=1,\ldots,n_s$, and zero mass elsewhere.
	\item\label{it:omegantwophase} For each $\beta : \mathcal{W}\rightarrow[0,1]$, let $\widehat{\omega}^{\beta,n,F}\triangleq \omega^{\beta,n,F} + \sum_{s\in\mathcal{S}} Q_s^{n,F} \nabla  \Omega_s^{\beta,n,F}$, and note that $\widehat{\omega}^{\beta,n,F}$ rewrites as $Q_\star^{n,F} \UR^{\beta,n,F} + \sum_{s=1}^S Q_s^{n,F} \nabla \Omega_s^{\beta,n,F}$.
	\item\label{it:thetantwophase} Let $\widehat{\theta}^{n,F}\triangleq\sup\{\theta : \widehat{\omega}^{w\mapsto \Ind_{\{\VEm^{n,F}(w)<\theta\}},n,F} \le \mu\}$, where $\sup\emptyset=-\infty$.
	\item\label{it:etantwophase} Let $\widehat{\eta}^{n,F}$ be any element of the set $\argmin_{\eta\in[0,1]}\left(\widehat{\omega}^{\beta_\eta,n,F}-\mu\right)^2$, where $\beta_\eta\triangleq w\mapsto \Ind_{\{\VEm^{n,F}(w)<\widehat{\theta}^{n,F}\}} + \eta \Ind_{\{\VEm^{n,F}(w)=\widehat{\theta}^{n,F}\}}$.
	\item\label{it:betantwophase} Let $\widehat{\beta}^{n,F}\triangleq \beta_{\widehat{\eta}^{n,F}}$.
	\item\label{it:onesteptwophase} Estimate $\gamma^{\widehat{\beta}^{n,F},0}$ with
	\begin{align*}
	\widehat{\gamma}^{\widehat{\beta}^{n,F},n,F}&\triangleq \gamma^{\widehat{\beta}^{n,F},n,F} + \sum_{s\in\mathcal{S}} Q_s^{n,F} \nabla \Gamma_s^{\widehat{\beta}^{n,F},n,F} \\
	&= n_\star^{-1}\sum_{i=1}^{n_\star} \UR^{\widehat{\beta}^{n,F},n,F}(w_\star[i])\VE^{n,F}(w_\star[i]) + \sum_{s=1}^S Q_s^{n,F} \nabla \Gamma_s^{\widehat{\beta}^{n,F},n,F}.
	\end{align*}
	\item\label{it:phintwophase} Estimate $\phi^0$ with $\widehat{\phi}^{n,F}\triangleq \frac{\widehat{\gamma}^{\widehat{\beta}^{n,F},n,F}}{\widehat{\omega}^{\widehat{\beta}^{n,F},n,F}}$.
\end{enumerate}
For confidence interval construction, we note that similar conditions to those used in \ref{thm:al} yield that our estimator is asymptotically linear, with the same $P_\star^0$ gradient as in the main text (though with the observed data parameters for the completed trial distributions replaced by full data parameters) and, for $s=1,\ldots,S$, $P_s^0$ gradients $o_s\mapsto \frac{\delta}{P_s^0\{\Delta=1|l,a,y\}}\nabla \Phi_s^0(o_s)$, again replacing the observed data parameters by the full data parameters in the definition of $\nabla \Phi_s^0(o_s)$.

\begin{remark}
Consider a randomized trial where $L$ precedes randomization and $A$ is independent of $L$ conditional on $W$. In this case, it is straightforward to improve the efficiency of the above procedure if $L$ is predictive of $Y$ after accounting for $A$ and $W$ in at least one of the trials. In particular, efficiency could be improved by leveraging this biomarker when estimating both the (known) propensity score $P_s^{0,F}\{A|W\}$ and the outcome regression $\E_s^{0,F}[Y|A,W]$, namely by adding $L$ to the both of the corresponding conditioning statements. The same efficiency gain of course holds for the algorithm in the main text, since setting $\Delta=1$ with probability one shows that $W$ being observed on all individuals is a special case of the results in this appendix.
\end{remark}

\subsection{Monotonic vaccine efficacy curve}\label{app:mono}
In this section, we describe a situation in which one can replace the condition that $P_\star^0\{\VEm^0(W)=\theta^0\}=0$ with the following new condition:
\begin{enumerate}[label=M),ref=M]
	\item\label{it:monotonic} $W$ is real-valued and $w\mapsto \VEm^0(w)$ is monotonic.
\end{enumerate}

We break this section into three parts. First, we present a new partial bridging formula specific to the monotonic $\VEm^0$ case. Then, we give formal conditions that will allow one to establish the pathwise derivative of the parameter specified by this formula. Finally, we describe how to modify our estimator from the main text so that the validity of the confidence intervals neither relies on $P_\star^0\{\VEm^0(W)=\theta^0\}=0$ nor on \ref{it:betangood}.

\subsubsection{Partial bridging formula.}
We will use the notation from the main text to express our partial bridging formula. We will define alternatives to $\theta^0$, $\eta^0$, $\beta^0$, and $\phi^0$, which we will respectively denote by $\underline{\theta}^0$, $\underline{\eta}^0$, $\underline{\beta}^0$, and $\underline{\phi}^0$.

Define
\begin{align*}
&\underline{\theta}^0\triangleq \sup\Big\{\theta\in\mathbb{R} : \omega^{w\mapsto\Ind_{\{w<\theta\}},0}\le \mu\Big\},
\end{align*}
where $\sup\emptyset = -\infty$ by convention. Let $\underline{\beta}_{\eta}\triangleq w\mapsto  \Ind_{\{w<\underline{\theta}^0\}} + \eta \Ind_{\{w=\underline{\theta}^0\}}$, and define $\underline{\eta}^0$ to be the smallest element of the set $\argmin_{\eta\in[0,1]}\left(\omega^{\underline{\beta}_{\eta},0}-\mu\right)^2$. 
Let $\underline{\beta}^0\triangleq \underline{\beta}_{\underline{\eta}^0}$. One can show that, if $\E_\star^0[\lambda^0(W)]\le \mu\le \E_\star^0[\upsilon^0(W)]$, then $\omega^{\underline{\beta}^0,0}=\mu$. Our partial bridging parameter is given by
\begin{align*}
\underline{\Phi}(\mathcal{P}^0)\triangleq \frac{\gamma^{\underline{\beta}^0,0}}{\omega^{\underline{\beta}^0,0}}\triangleq\underline{\phi}^0.
\end{align*}
The proof of the following lemma is nearly identical to the proof of Lemma~\ref{lem:lb} and so is omitted.
\begin{lemma}
Suppose \ref{it:monotonic}, \ref{it:brVEbridge}, \ref{it:brdatadepub}, \ref{it:brcommonsupport}, and \ref{it:brvebdd} hold. If $\mu=\E_\star^0[\E_\star^{0,F}[Y|A=0,W]]$, then $\Psi(P_\star^{0,F})\ge \underline{\phi}^0$.
\end{lemma}

\subsubsection{First-order expansion of $\underline{\phi}^0$.}
Our first-order expansion also replaces \ref{it:omegacontnotflat} in the main text with two alternative assumptions, the first of which is given below.
\begin{enumerate}[label=D\arabic*),ref=D\arabic*,series=en:regcondsmono]
	\item\label{it:Wcdfcontnotflat} $\E_\star^0[\lambda^0(W)]<\mu<\E_\star^0[\upsilon^0(W)]$ and either (i) $P_\star^0\{W=\theta^0\}>0$ or (ii) $P_\star^0\{W=\theta^0\}=0$, $w\mapsto \VEm^0(w)$ is Lipschitz at $\theta^0$, and
	\begin{align*}
	0&<\liminf_{t\rightarrow 0} \frac{P_{\star}^0\{W< \theta+t\}-P_{\star}^0\{W< \theta\}}{t}. 
	\end{align*}
\end{enumerate}
Note that $\E_\star^0[\lambda^0(W)]<\mu<\E_\star^0[\upsilon^0(W)]$ implies that $\theta^0$ is finite.

We give a theorem presenting the gradients of our parameter. For brevity, we only sketch the proof.
\begin{theorem} \label{thm:pdmono}
If \ref{it:brcommonsupport} and \ref{it:brvebdd} hold and either \ref{it:Wcdfcontnotflat}, \ref{it:mubig}, or \ref{it:musmall} holds, then $\underline{\Phi}$ is pathwise differentiable and, for each $s\in\mathcal{S}$, the $P_s^0$ gradient is given by
\begin{align*}
\nabla \underline{\Phi}_s^0(o_s)=& \begin{cases}
\dfrac{\nabla \Gamma_s^{\underline{\beta}^0,0}(o_s)}{\omega^{\underline{\beta}^0,0}} - \VEm^0(\underline{\theta}^0)\dfrac{\nabla \Omega_s^{\underline{\beta}^0,0}(o_s)}{\omega^{\underline{\beta}^0,0}},&\mbox{ if }\E_\star^0[\lambda^0(W)]<\mu<\E_\star^0[\upsilon^0(W)], \\[1.3em]
\dfrac{\nabla \Gamma_s^{\underline{\beta}^0,0}(o_s)}{\omega^{\underline{\beta}^0,0}} - \gamma^{\underline{\beta}^0,0}\dfrac{\nabla \Omega_s^{\underline{\beta}^0,0}(o_s)}{[\omega^{\underline{\beta}^0,0}]^2},&\mbox{ otherwise.}
\end{cases}.
\end{align*}
\end{theorem}
\begin{proof}[Sketch of Proof of Theorem~\ref{thm:pdmono}]
The proof under \ref{it:mubig} or \ref{it:musmall} is essentially identical to that of Theorem~\ref{thm:pd} under the same conditions, we suppose \ref{it:Wcdfcontnotflat} in the remainder.

We first outline the convergence result of $\underline{\beta}^\epsilon$ to $\underline{\beta}^0$. If (i), then we can instead show that $\underline{\theta}^\epsilon=\underline{\theta}^0$ for all $\epsilon$ small enough and $\underline{\eta}^\epsilon=\underline{\eta}^0+O(\epsilon)$. If (ii), then \ref{it:Wcdfcontnotflat} yields that $\underline{\theta}^\epsilon=\underline{\theta}^0+O(\epsilon)$. In either case, $\underline{\theta}^0$ is finite for all $\epsilon$ small enough so that $\omega^{\underline{\beta}^\epsilon,\epsilon}=\mu$ for all $\epsilon$ small enough. Hence,
\begin{align*}
\frac{\phi^\epsilon-\phi^0}{\epsilon} = \mu^{-1}\frac{\gamma^{\underline{\beta}^\epsilon,\epsilon}-\gamma^{\underline{\beta}^0,0}}{\epsilon},
\end{align*}
so it is enough to study the right-hand side multiplied by $\mu\epsilon$. We will use that
\begin{align*}
\gamma^{\underline{\beta}^\epsilon,\epsilon}-\gamma^{\underline{\beta}^0,0}=&\, \left\{\gamma^{\underline{\beta}^0,\epsilon}-\gamma^{\underline{\beta}^0,0} - \VEm^0(\underline{\theta}^0)[\omega^{\underline{\beta}^0,\epsilon}-\omega^{\underline{\beta}^0,0}]\right\} \\
&+ \VEm^0(\underline{\theta}^0)\left[\omega^{\underline{\beta}^\epsilon,\epsilon}-\omega^{\underline{\beta}^0,0}\right] \\
&+ \left\{\gamma^{\underline{\beta}^\epsilon,\epsilon}-\gamma^{\underline{\beta}^0,\epsilon} - \VEm^0(\underline{\theta}^0)[\omega^{\underline{\beta}^\epsilon,\epsilon}-\omega^{\underline{\beta}^0,\epsilon}]\right\}
\end{align*}
A delta method argument shows that dividing the leading term by $\epsilon$ and taking the limit as $\epsilon\rightarrow 0$ yields $\sum_{s\in\mathcal{S}} P_s^0 h_s [\nabla \Gamma_s^{\underline{\beta}^0,0}(o_s) + \VEm^0(\underline{\theta}^0)\nabla \Omega_s^{\underline{\beta}^0,0}]$. As we established earlier in this proof, $\omega^{\underline{\beta}^\epsilon,\epsilon}=\mu=\omega^{\underline{\beta}^0,0}$ for all $\epsilon$ small enough, so that the second term above is zero for all $\epsilon$ small enough . The remainder of this proof aims to show that the final term is $o(\epsilon)$. We will use a constant $c$ that may vary line by line. We now study the final term above, whose numerator bounds as
\begin{align*}
&\left|\gamma^{\underline{\beta}^\epsilon,\epsilon}-\gamma^{\underline{\beta}^0,\epsilon} - \VEm^0(\underline{\theta}^0)[\omega^{\underline{\beta}^\epsilon,\epsilon}-\omega^{\underline{\beta}^0,\epsilon}]\right| \\
&= (1+c|\epsilon|)\E_\star^0\left|\{\upsilon^\epsilon(W)-\lambda^\epsilon(W)\}\{\underline{\beta}^\epsilon(W)-\underline{\beta}^0(W)\}\left\{\VEm^\epsilon(W)-\VEm^0(\underline{\theta}^0)\right\}\right| \\
&\le (1+c|\epsilon|)\E_\star^0\left[|\underline{\beta}^\epsilon(W)-\underline{\beta}^0(W)|\left\{\left|\VEm^0(W)-\VEm^0(\underline{\theta}^0)\right| + c|\epsilon|\right\}\right].
\end{align*}
Call the right-hand side $\underline{T}_1^\epsilon$. If (i) holds, then the fact that $\underline{\theta}^\epsilon=\underline{\theta}^0$ for all $\epsilon$ small enough shows that, for all such $\epsilon,$
\begin{align*}
\underline{T}_1^\epsilon&\le c|\epsilon|(1+c|\epsilon|)|\underline{\eta}^\epsilon-\underline{\eta}^0|P_\star^0\{W=\underline{\theta}^0\}.
\end{align*}
As $\underline{\eta}^\epsilon-\underline{\eta}^0=O(\epsilon)$, the right-hand side must be $o(\epsilon)$. If, instead, (ii) holds, then similar techniques to those used to control $T_1^\epsilon$ in Theorem~\ref{thm:pdLambda} show that
\begin{align*}
\underline{T}_1^\epsilon&\le (1+c|\epsilon|)\E_\star^0\left[\Ind_{\{0<|W-\underline{\theta}^0|<|\underline{\theta}^\epsilon-\underline{\theta}^0|\}}\left\{\left|\VEm^0(W)-\VEm^0(\underline{\theta}^0)\right| + c|\epsilon|\right\}\right].
\end{align*}
By the fact that $\VEm^0$ is Lipschitz at $\underline{\theta}^0$ and the bound on $|W-\underline{\theta}^0|$ given in the indicator, we have that
\begin{align*}
\underline{T}_1^\epsilon&\le (1+c|\epsilon|)\left\{c\left|\underline{\theta}^\epsilon-\underline{\theta}^0\right| + c|\epsilon|\right\}P_\star^0\left\{0<|W-\underline{\theta}^0|<|\underline{\theta}^\epsilon-\underline{\theta}^0|\right\}.
\end{align*}
As $\underline{\theta}^\epsilon-\underline{\theta}^0=O(\epsilon)$ and the probability statement is $o(1)$, the right-hand side is $o(\epsilon)$.
\end{proof}

\subsubsection{Modification to our estimator.}
We now present the modification to the estimator presented in the main text that allows us to replace the condition that $\VEm^0(W)$, $W\sim P_\star^0$, does not concentrate mass at a decision boundary with the condition that $W\sim P_\star^0$ does not concentrate mass at a decision boundary and the monotonicity condition \ref{it:monotonic}.

The first \ref*{it:omegan} steps of our estimation procedure are identical to that in the main text, and the modification of the remaining steps is given below.
\begin{enumerate}
	\item[\ref*{it:thetan}.]\label{it:thetanmono} Let $\widehat{\underline{\theta}}^n\triangleq\sup\{\theta : \widehat{\omega}^{w\mapsto \Ind_{\{w<\theta\}},n} \le \mu\}$, where $\sup\emptyset=-\infty$.
	\item[\ref*{it:etan}.]\label{it:etanmono} Let $\widehat{\underline{\eta}}^n$ be any element of the set $\argmin_{\eta\in[0,1]}\left(\widehat{\omega}^{\beta_\eta,n}-\mu\right)^2$, where $\beta_\eta\triangleq w\mapsto \Ind_{\{w<\widehat{\underline{\theta}}^n\}} + \eta \Ind_{\{w=\widehat{\underline{\theta}}^n\}}$.
	\item[\ref*{it:betan}.]\label{it:betanmono} Let $\widehat{\underline{\beta}}^n\triangleq \beta_{\widehat{\underline{\eta}}^n}$.
	\item[\ref*{it:onestep}.]\label{it:onestepmono} Estimate $\gamma^{\widehat{\underline{\beta}}^n,0}$ with
	\begin{align*}
	\widehat{\gamma}^{\widehat{\underline{\beta}}^n,n}&\triangleq \gamma^{\widehat{\underline{\beta}}^n,n} + \sum_{s\in\mathcal{S}} Q_s^n \nabla \Gamma_s^{\widehat{\underline{\beta}}^n,n} \\
	&= n_\star^{-1}\sum_{i=1}^{n_\star} \UR^{\widehat{\underline{\beta}}^n,n}(w_\star[i])\VE^n(w_\star[i]) + \sum_{s=1}^S Q_s^n \nabla \Gamma_s^{\widehat{\underline{\beta}}^n,n}.
	\end{align*}
	\item[\ref*{it:phin}.]\label{it:phinmono} Estimate $\underline{\phi}^0$ with $\widehat{\underline{\phi}}^n\triangleq \frac{\widehat{\gamma}^{\widehat{\underline{\beta}}^n,n}}{\widehat{\omega}^{\widehat{\underline{\beta}}^n,n}}$.
\end{enumerate}
The conditions for the asymptotic linearity of this approach are similar to those of Theorem~\ref{thm:al}, but weaker because they do not rely on the condition that $P_\star^0\{\VEm^0(W)=\theta^0\}=0$ (nor any analogue thereof). To start, we consider why, at least for $W$ continuous, $P_\star^0\{\VEm^0(W)=\theta^0\}=0$ was essentially necessary for it to be plausible that the procedure in the main text was asymptotically linear. For asymptotic linearity to hold, the indicators that $\VEm^n(w)<\widehat{\theta}^n$ and $\VEm^n(w)=\widehat{\theta}^n$ that appear in the empirical gradients must converge to a fixed limit. The fact that $\VEm^n$ is only an estimate of $\VEm^0$ suggests that it is not plausible that the indicators of the events $\VEm^n(w)<\widehat{\theta}^n$ or $\VEm^n(w)=\widehat{\theta}^n$ converge to a fixed limit for any $w$ for which $\VEm^0(w)=\theta^0$. Thus, if $W$ is continuous and $P_\star^0\{\VEm^0(W)=\theta^0\}>0$, then it is not plausible that these indicators converge to any fixed limit.

Consider now the procedure above. The procedure above replaces the indicators that $\VEm^n(w)<\widehat{\theta}^n$ and $\VEm^n(w)=\widehat{\theta}^n$ with indicators that $w<\widehat{\underline{\theta}}^n$ and $w=\widehat{\underline{\theta}}^n$: in particular, it is expected to be valid regardless of the value of $P_\star^0\{W=\underline{\theta}^0\}$. The conditions needed to ensure that these new indicators converge to a fixed limit are weaker than those needed for the convergence of $\VEm^n(w)<\widehat{\theta}^n$ or $\VEm^n(w)=\widehat{\theta}^n$. If $P_\star^0$ concentrates mass at $\underline{\theta}^0$, then we expect that $\widehat{\underline{\theta}}^n$ will equal $\underline{\theta}^0$ with probability approaching one, so that clearly the indicators that $w<\widehat{\underline{\theta}}^n$ and $w=\widehat{\underline{\theta}}^n$ are fixed with probability approaching one. If, on the other hand, $P_\star^0$ does not concentrate mass at $\underline{\theta}^0$, then indicators that $w<\widehat{\underline{\theta}}^n$ and $w=\widehat{\underline{\theta}}^n$ will generally converge to a fixed limit if $\widehat{\underline{\theta}}^n\rightarrow \underline{\theta}^0$. For some intuition on why $\widehat{\underline{\theta}}^n\rightarrow \underline{\theta}^0$ is to be expected, note that $\widehat{\underline{\theta}^n}$ is essentially an empirical quantile for an $(\upsilon^0-\lambda^0)$-weighted version of $P_\star^0$. We say ``essentially'' because conditions can be given under which we can replace the estimated weights $\upsilon^n-\lambda^n$ by the true weights $\upsilon^0-\lambda^0$ at the expense of an $O_P(\nmin^{-1/2})$ term. Given this replacement, we can show that \ref{it:Wcdfcontnotflat} yields that $\widehat{\underline{\theta}}^n=\theta^0 + O_P(\nmin^{-1/2})$. Note that this is stronger than the slower than root-$n$ rate of convergence of $\widehat{\theta}^n$ to $\theta^0$ in the nonmonotonic case given in Lemma~\ref{lem:thetancons}.

\section{Nested case-control simulation results} \label{app:nestedccsim}
Figure~\ref{fig:twophasecovg} shows that our the coverage of our estimation scheme decreased under the 1:1 nested case-control design, though also that this coverage improves with sample size. At the Moderate and Tight settings, Figure~\ref{fig:twophaseavglb} shows that our estimate of the partially bridged lower bound has positive bias for the nested case-control design, though that this bias reduces with sample size.

Figure~\ref{fig:twophasetruevar} compares the Monte Carlo variance of the estimator that uses the known censoring mechanism versus the estimator that estimates the censoring mechanism. The two procedures appear to have similar variance in this particular simulation setting. 

\begin{figure}
	\centering
	\includegraphics[width=\linewidth]{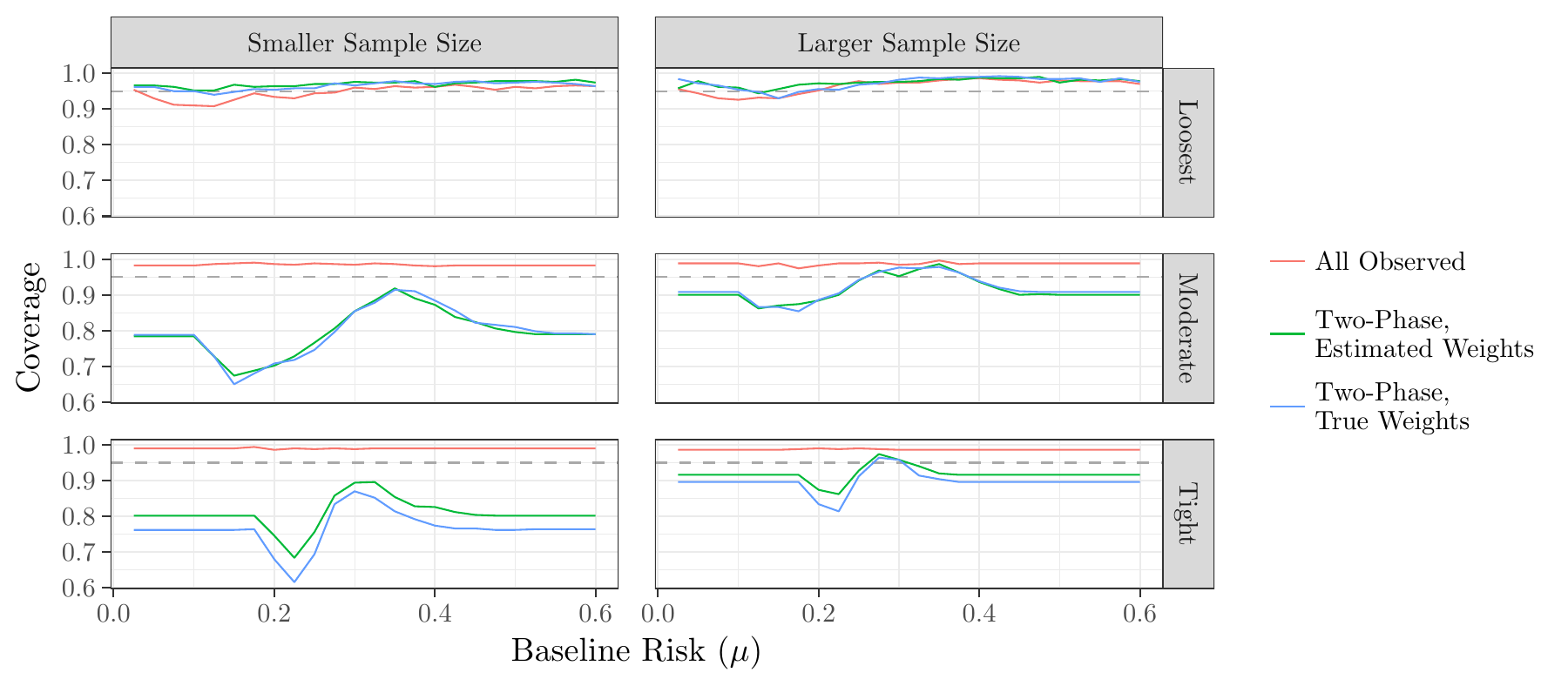}
	\caption{Comparison of the coverage of our lower confidence bound for $\phi^0$, i.e. the lower bound on the vaccine efficacy, when there is and is not data missing due to a nested case-control sampling design. Conducted at both smaller and larger sample sizes, respectively with $(n_\star,n_1,n_2)$ equal to $(100,2000,2000)$ and $(200,4000,4000)$, and for different choices of $\ell_s$ and $\mathbbmss{u}_s$, determining the tightness of the unvaccinated risk bounds. Horizontal dashed lines drawn at 95\% coverage, vertical dashed lines drawn at the true baseline risk value $\mu$.}
	\label{fig:twophasecovg}
\end{figure}

\begin{figure}
	\centering
	\includegraphics[width=\linewidth]{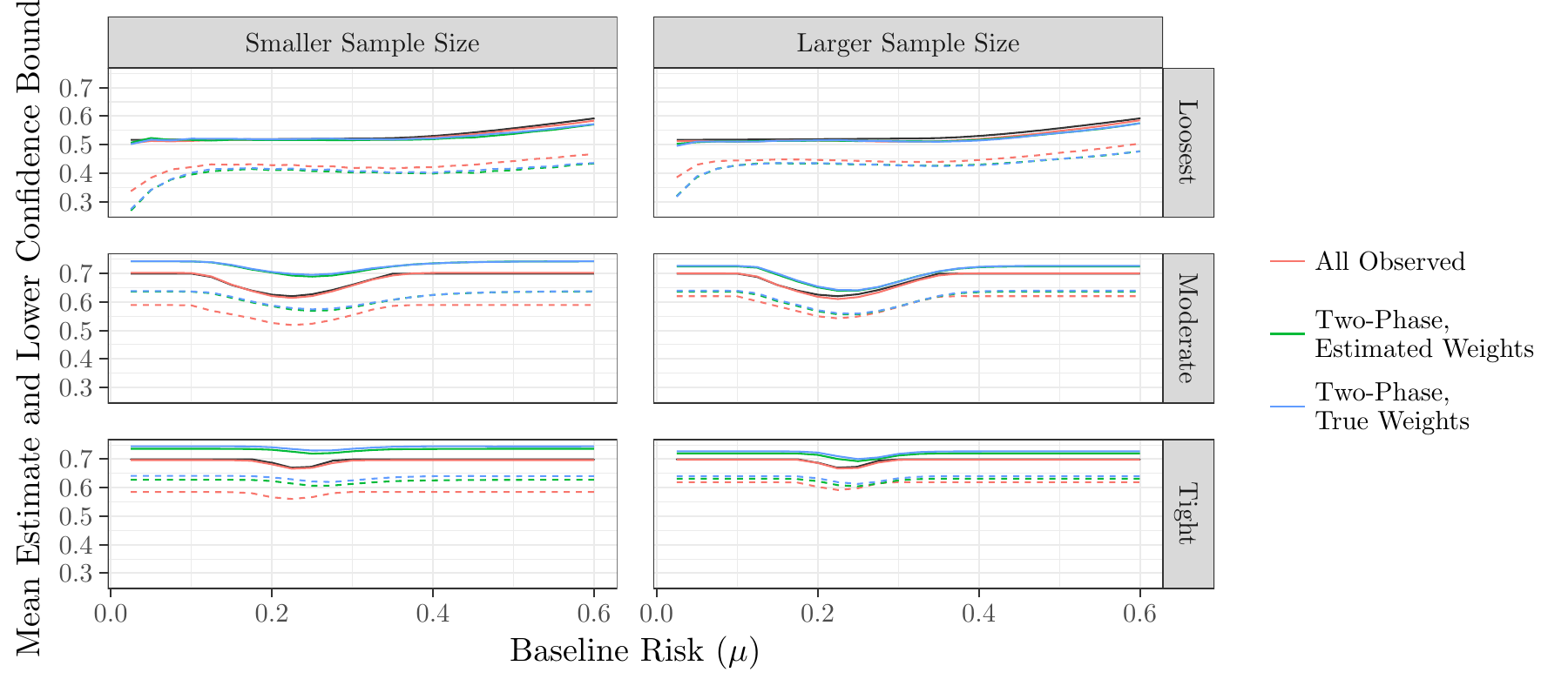}
	\caption{Comparison of average estimates (solid lines) and lower confidence bounds (dashed lines) for $\phi^0$, i.e. the lower bound on the vaccine efficacy, when there is and is not data missing due to a nested case-control sampling design. Conducted at both smaller and larger sample sizes, respectively with $(n_\star,n_1,n_2)$ equal to $(100,2000,2000)$ and $(200,4000,4000)$, and for different choices of $\ell_s$ and $\mathbbmss{u}_s$, determining the tightness of the unvaccinated risk bounds. Black trend lines denote true $(\ell_s,\mathbbmss{u}_s,\mu)$-specific lower bound.}
	\label{fig:twophaseavglb}
\end{figure}

\begin{figure}
	\centering
	\includegraphics[width=\linewidth]{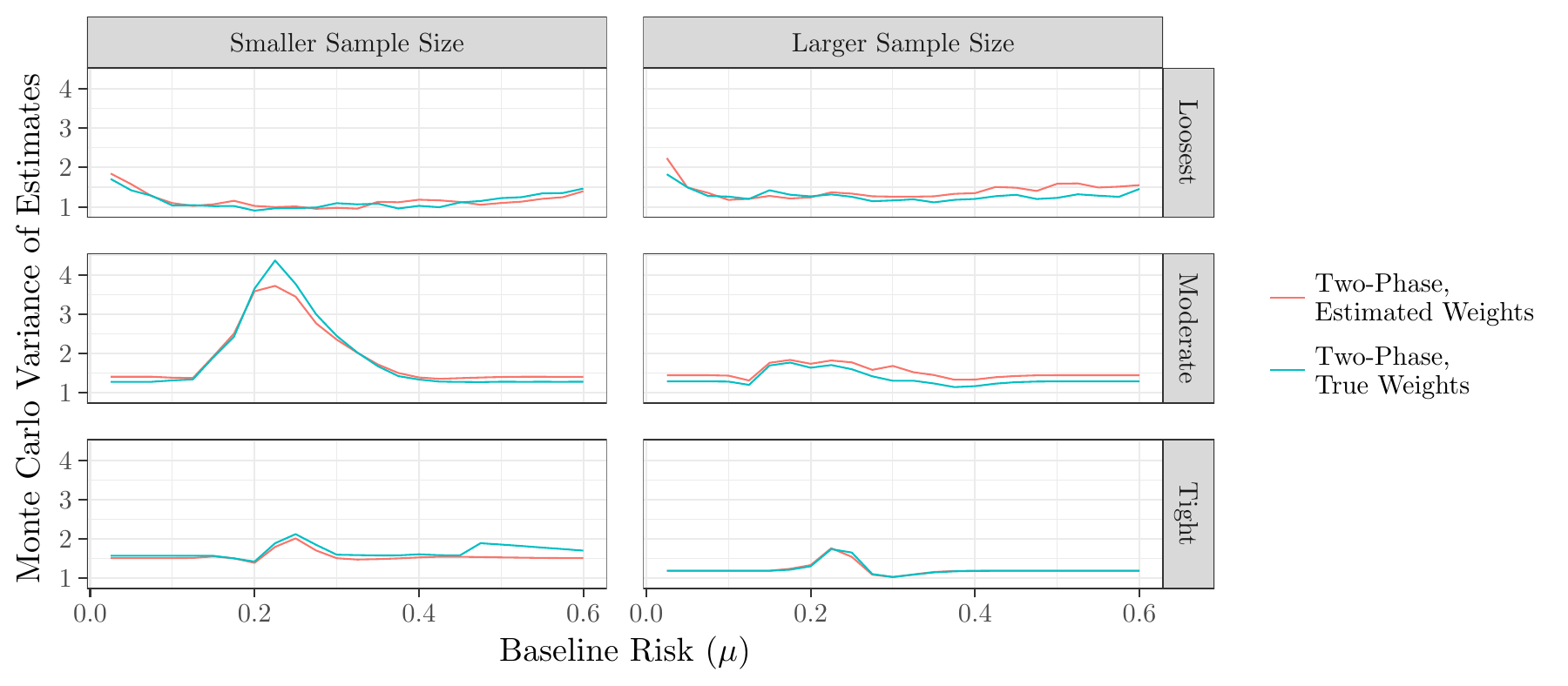}
	\caption{Comparison of estimator's Monte Carlo variance in the presence of two-phase sampling (standardized by Monte Carlo variance of estimator that observed $W$ for all individuals). Conducted at both smaller and larger sample sizes, respectively with $(n_\star,n_1,n_2)$ equal to $(100,2000,2000)$ and $(200,4000,4000)$, and for different choices of $\ell_s$ and $\mathbbmss{u}_s$, determining the tightness of the unvaccinated risk bounds.}
	\label{fig:twophasetruevar}
\end{figure}

\end{document}